%% file: main.tex
\documentclass[senames,dvipsnames,svgnames,table,11pt]{article}
\usepackage[a4paper, total={7in, 8in}]{geometry}
\newcommand\blfootnote[1]{
  \begin{NoHyper}
  \renewcommand\thefootnote{}\footnote{#1}%
  \addtocounter{footnote}{-1}%
  \end{NoHyper}
}

\input{macro}

\title{NP-hardness of testing equivalence to sparse polynomials and to constant-support polynomials}

\author{{Omkar Baraskar}\\
		\normalsize{University of Waterloo}\\
		\normalsize{\tt{obaraska@uwaterloo.ca}}
		\and {Agrim Dewan}\\
		\normalsize{Indian Institute of Science}\\
	    \normalsize{\tt{agrimdewan@iisc.ac.in}}
		\and {Chandan Saha\footnote{Partially supported by a MATRICS grant of the Science and Engineering Research Board, DST, India.}}\\
		\normalsize{Indian Institute of Science}\\
		\normalsize{\tt{chandan@iisc.ac.in}}
            \and {Pulkit Sinha\footnote{Partially supported by a Mike and Ophelia Lazardis Fellowship, and partially supported by NSERC Canada.}}\\
            \normalsize{University of Waterloo} \\
            \normalsize{\tt{psinha@uwaterloo.ca}}
  }
\date{}

\begin{document}

\maketitle
\thispagestyle{empty}
\begin{abstract}
An $s$-sparse polynomial has at most $s$ monomials with nonzero coefficients. The Equivalence Testing problem for sparse polynomials (ETsparse) asks to decide if a given polynomial $f$ is equivalent to (i.e., in the orbit of) some $s$-sparse polynomial. In other words, given $f \in \F[\vecx]$ and $s \in \N$, ETsparse asks to check if there exist $A \in \GL(|\vecx|, \F)$ and $\vecb \in \F^{|\vecx|}$ such that $f(A\vecx + \vecb)$ is $s$-sparse. We show that ETsparse is NP-hard over any field $\F$, if $f$ is given in the sparse representation, i.e., as a list of nonzero coefficients and exponent vectors. This answers a question posed by Gupta, Saha and Thankey (SODA 2023) and also, more explicitly, by Baraskar, Dewan and Saha (STACS 2024). The result implies that the Minimum Circuit Size Problem (MCSP) is NP-hard for a \emph{dense} subclass of depth-$3$ arithmetic circuits if the input is given in sparse representation. We also show that approximating the smallest $s_0$ such that a given $s$-sparse polynomial $f$ is in the orbit of some $s_0$-sparse polynomial to within a factor of $s^{\frac{1}{3} - \epsilon}$ is NP-hard for any $\epsilon >0$; observe that $s$-factor approximation is trivial as the input is $s$-sparse. Finally, we show that for any constant $\sigma \geq 5$, checking if a polynomial (given in sparse representation) is in the orbit of some support-$\sigma$ polynomial is NP-hard. Support of a polynomial $f$ is the maximum number of variables present in any monomial of $f$. These results are obtained via direct reductions from the $\SAT$ problem.\blfootnote{A preliminary version of this paper appeared in the proceedings of the 51\textsuperscript{st} International Colloquium on Automata, Languages and Programming (ICALP), 2024.}   
\end{abstract}
\clearpage
\newpage

\setcounter{tocdepth}{2}
\tableofcontents
\thispagestyle{empty}
\newpage
\setcounter{page}{1}

\input{Intro}
\input{Results}
\input{related_works}
\input{roadmap}
\input{Prem}
\input{ETSparse_NPhard}
\input{Gap_ETSparse}
\input{ETCSP}
\input{SparseShift}
\input{Conclusion}
\nocite{*}
\bibliographystyle{alpha}
\bibliography{bibl}

\input{Appendix}

\end{document}

%% file: macro.tex
\usepackage{fullpage}
\usepackage{amssymb}
\usepackage{amsmath}
\usepackage{amsthm}
\usepackage{multirow}
\usepackage{enumerate}
\usepackage{graphicx}
\usepackage{mathrsfs}
\usepackage[utf8]{inputenc}

\usepackage{float}
\usepackage{pdfpages}
\usepackage{listings}

\newcommand{\ith}{{^\text{th}}}

\renewcommand{\cal}[1]{{\mathcal{#1}}}

\newcommand{\et}{\mathrm{ETsparse}}
\newcommand{\etcsp}{\mathrm{ETsupport}}
\newcommand{\gapet}{\alpha\text{-}\mathrm{gap}\text{-}\et}
\newcommand{\prometcsp}{(\sigma+1)\mathrm{\text{-}to\text{-}}\sigma\mathrm{\ ETsupport}}
\newcommand{\sparseshift}{\mathrm{SETsparse}}
\newcommand{\gapsparseshift}{\alpha\text{-}\mathrm{gap}\text{-}\mathrm{SETsparse}}

\usepackage[T1]{fontenc}  
\usepackage[inline]{enumitem}  
\usepackage{lmodern}  
\usepackage{braket}

\makeatletter%
\newsavebox{\Exbox}\newlength{\len}%
{\end{itemize}}  
\makeatother%

\setlist[description]{font=\normalfont\itshape\space}
\newlength{\jeroenlen}
\newenvironment{example}
{\settowidth{\jeroenlen}{\textit{Remarks.}}%
	\begin{description}[leftmargin=\jeroenlen,labelwidth=0pt,labelsep=0pt]
		\item[\textit{Remarks.\qquad}]%
		\begin{enumerate}[leftmargin=0em,labelsep=0.5em]}
		{\end{enumerate}\end{description}}

\usepackage{color}
\usepackage[pdfstartview=FitH,pdfpagemode=UseNone,colorlinks=true,citecolor=blue,linkcolor=blue,backref=page]{hyperref}

\usepackage{cleveref}
\usepackage{mathtools}
\usepackage{relsize}
\usepackage{amsfonts}
\usepackage[T1]{fontenc}
\usepackage{mathpazo}
\usepackage{bm}
\usepackage{changepage}
\usepackage{tcolorbox}
\usepackage{tikz}
\usepackage{thmtools}
\usepackage{thm-restate}
\usepackage{subcaption}
\usepackage{enumitem}
\usepackage{algorithm}
\usepackage{algcompatible}

\usepackage{dirtytalk}
\usepackage[thinc]{esdiff}
\usepackage[numbers]{natbib}
\usepackage{bbm}

\usepackage{tikz}
\usetikzlibrary{positioning}
\usetikzlibrary{calc}
\usetikzlibrary{decorations.shapes}
\usetikzlibrary{decorations.pathreplacing,angles,quotes}
\usepackage{todonotes}
\usepackage{BOONDOX-cal} 
\usepackage{textcomp}
\usepackage{blkarray}

\allowdisplaybreaks

\definecolor{blueviolet}{RGB}{60,50,200}
\definecolor{oliveg}{RGB}{40,200,30}
\hypersetup{colorlinks=true,       
    linkcolor=blueviolet,
    citecolor=oliveg,
}

\newtheorem{theorem}{Theorem}
\theoremstyle{definition}
\newtheorem{definition}{Definition}[section]
\newtheorem{problem}{Problem}[section]
\newtheorem{proposition}{Proposition}[section] \newtheorem{lemma}{Lemma}[section]

\newtheorem{corollary}{Corollary}[section]
\newtheorem{claim}{Claim}[section]

\newtheorem{observation}{Observation}[section]



\algrenewcommand\alglinenumber[1]{\footnotesize #1.}

\renewcommand{\bar}{\overline}

\newcommand{\veca}{{\mathbf{a}}}
\newcommand{\vecb}{{\mathbf{b}}}

\newcommand{\vecu}{{\mathbf{u}}}

\newcommand{\vecw}{{\mathbf{w}}}
\newcommand{\vecx}{{\mathbf{x}}}
\newcommand{\vecy}{{\mathbf{y}}}
\newcommand{\vecz}{{\mathbf{z}}}

\newcommand{\vecalpha}{{\boldsymbol{\alpha}}}

\newcommand{\C}{\mathbb{C}}
\newcommand{\F}{\mathbb{F}}

\newcommand{\N}{\mathbb{N}}
\newcommand{\Q}{\mathbb{Q}}
\newcommand{\R}{\mathbb{R}}

\newcommand{\Z}{\mathbb{Z}}

\newcommand{\cC}{\mathcal{C}}

\newcommand{\cF}{\mathcal{F}}

\newcommand{\var}{\mathrm{var}}

\newcommand{\NP}{\mathsf{NP}}

\newcommand{\AM}{\mathsf{AM}}
\newcommand{\coAM}{\mathsf{coAM}}

\newcommand{\SAT}{\mathrm{3\text{-}SAT}}
\newcommand{\CNF}{\mathrm{3\text{-}CNF}}

\newcommand{\supp}{\textnormal{Supp}}

\newcommand{\poly}{\textnormal{poly}}

\newcommand{\IMM}{\mathsf{IMM}}
\newcommand{\SP}{\mathsf{SP}}
\newcommand{\PSym}{\mathsf{PSym}}

\newcommand{\GL}{\mathrm{GL}}

\newcommand{\orb}{\textnormal{orb}}

\newcommand{\hdthree}{\textnormal{hom-}\Sigma\Pi\Sigma}

\renewcommand\thefootnote{\textcolor{red}{\arabic{footnote}}}
\makeatletter
\newcommand\footnoteref[1]{\protected@xdef\@thefnmark{\ref{#1}}\@footnotemark}
\makeatother

\definecolor{DSgray}{cmyk}{0,0,0,0.7}

\definecolor{anti-flashwhite}{rgb}{0.95, 0.95, 0.96}
\definecolor{britishracinggreen}{rgb}{0.0, 0.26, 0.15}
\definecolor{cornellred}{rgb}{0.7, 0.11, 0.11}
\definecolor{cream}{rgb}{1.0, 0.99, 0.82}
\definecolor{lavender}{rgb}{0.9, 0.9, 0.98}
\definecolor{midnightblue}{rgb}{0.1, 0.1, 0.44}	
\definecolor{deepsaffron}{rgb}{1.0, 0.6, 0.2}
\definecolor{piggypink}{rgb}{0.99, 0.87, 0.9}
\definecolor{dollarbill}{rgb}{0.52, 0.73, 0.4}
\definecolor{crimsonred}{rgb}{0.6, 0.0, 0.0}
\definecolor{persiangreen}{rgb}{0, 0.5, 0.0}
\definecolor{celadon}{rgb}{0.87, 0.99, 0.8}
\definecolor{dukeblue}{rgb}{0.53, 0.0, 0.69}
\definecolor{sealbrown}{rgb}{0.44, 0.28, 0.1}
\definecolor{wheat}{rgb}{0.94, 0.92, 0.84}	
\definecolor{patriarch}{rgb}{0.03, 0.57, 0.82}
\definecolor{cerise}{rgb}{0.96, 0.0, 0.63}
\definecolor{scarlet}{rgb}{0.96, 0.0, 0.63}
\definecolor{almond}{rgb}{0.94, 0.87, 0.8}
\definecolor{aliceblue}{rgb}{0.94, 0.97, 1.0}

%% file: Intro.tex
\section{Introduction} \label{Section: Introduction}
The Polynomial Equivalence (PE) problem asks to decide if two polynomials, given as lists of coefficients, are equivalent. Polynomials $f, g \in \F[\vecx]$ are \emph{equivalent}, denoted as $f \sim g$, if there is an $A \in \GL(|\vecx|, \F)$ and a $\vecb \in \F^{|\vecx|}$ such that $f = g(A\vecx + \vecb)$. Equivalent polynomials represent the same function up to a change of the coordinate system.\footnote{Over $\R$, an invertible map $\vecx \mapsto A\vecx + \vecb$ is simply a combination of rotation, reflection, scaling, and translation.} The PE problem is thus regarded as the algebraic analog of the graph isomorphism (GI) problem. PE is at least as hard as GI \cite{AgrawalS05, Kayal11}, but we do not know if it is much harder than GI. There is, in fact, a cryptographic authentication scheme based on the presumed average-case hardness of PE \cite{Patarin96}. Is PE $\NP$-hard? Over finite fields, PE is not $\NP$-hard unless the polynomial hierarchy collapses \cite{Saxenaphd, Thierauf98}. In contrast, PE is not even known to be decidable over $\Q$. With the aim of gaining more insight into the complexity of testing polynomial equivalence, a natural variant of PE has been studied in the literature. This variant is known as \emph{equivalence testing}. 

In the following discussion, whenever we write "circuit(s)" and "formula(s)", we mean arithmetic circuit(s) and arithmetic formula(s), respectively, unless mentioned otherwise. \footnote{An \emph{arithmetic circuit} is like a Boolean circuit but with AND and OR replaced by $\times$ and $+$ gates, and with edges labelled by $\F$-elements. It computes a polynomial over $\F$. A \emph{formula} is a circuit whose underlying graph is a tree.}  \\

\noindent{\textbf{Equivalence testing.}} Equivalence testing (ET) comes in two flavors -- ET for polynomial families and ET for circuit classes. ET for a polynomial family $\cF$ is defined as follows: given a \emph{single} polynomial $f$, check if it is equivalent to some $g \in \cF$. This variant of PE was introduced in \cite{Kayal12, Kayal11}, wherein randomized polynomial-time ET algorithms were provided for the permanent, determinant, and elementary and power symmetric polynomial families. Subsequently, efficient ET algorithms were given for various other important polynomial families, such as the iterated matrix multiplication ($\IMM$) family \cite{KayalNST17} (see Section \ref{sec:related work}). These algorithms are efficient even if $f$ is provided as a circuit or a black-box.\footnote{Black-box access to $f$ means oracle access to $f$, we get $f(\veca)$ from a query point $\veca$ in one unit time. It is as if $f$ is given as a ``hidden'' circuit and the only operation we are allowed to do is evaluate the circuit at chosen points.} ET for a circuit class $\cC$ (a.k.a testing equivalence to $\cC$) is defined similarly: given a polynomial $f$, decide if it is equivalent to some polynomial $g$ that is computable by a circuit in $\cC$. Recently, efficient ET algorithms have been given for read-once formulas \cite{GuptaST23} and a special subclass of sparse polynomials, namely $t$-design polynomials for constant $t$ \cite{BDS24}. Sparse polynomials are depth-$2$ circuits.\footnote{We assume that a depth-$2$ circuit has a $+$ gate on top and a bottom layer of $\times$ gates. If the top gate is a $\times$ gate, then ET can be solved efficiently using polynomial factorization algorithms \cite{KaltofenT90}.} It is natural to ask whether or not ET can be solved efficiently for \emph{general} sparse polynomials. This question was posed in \cite{GuptaST23} and also, more explicitly, in \cite{BDS24}.

Before proceeding to discuss ET for sparse polynomials, we point out a subtle difference between ET for polynomial families and that for circuit classes. The polynomial families for which ET has been studied so far are such that if $f$ is equivalent to some $g$ in the family, then $g$ is unique and it can be readily identified from $f$. For example, if $f$ is equivalent to some determinant polynomial\footnote{The $n^2$-variate determinant polynomial is the determinant of the matrix $(x_{i,j})_{i,j \in [n]}$ of formal variables.}, then we know which one simply from the number of variables of $f$. Moreover, polynomials in most of these families admit well-known polynomial-size circuits. So, a circuit for $g$ can be derived once it is identified. Thus, if $f$ is also given as a circuit, then ET for such a family reduces to PE with the input polynomials given as circuits. Over finite fields, this version of PE is in $\AM \cap \coAM$ and hence unlikely to be $\NP$-hard. On the other hand, in the case of ET for a circuit class, if $f$ is equivalent to some circuit $C$ in the class, then $C$ need \emph{not} be unique, and further, $C$ may not be easily deducible from $f$. This leaves us with the prospect of proving that ET is hard for some natural circuit class. Do sparse polynomials form such a class? \\

\noindent{\textbf{ET for sparse polynomials.}} 
An $n$-variate, degree-$d$ polynomial is \emph{$s$-sparse} if it has at most $s$ monomials with nonzero coefficients. An $s$-sparse polynomial is computable by a depth-$2$ circuit having top fan-in $s$. Sparse polynomials have been extensively studied in algebraic complexity, particularly with regard to identity testing \cite{KlivansS01, LiptonV03}, interpolation \cite{Ben-OrT88, GrigorievKS90, KlivansS01, BlaserJ14}, and factorization \cite{GathenK85, BhargavaSV20} (see the tutorial \cite{Roche18} and the references therein for more algorithms involving sparse polynomials). ET provides yet another avenue to understand these ``basic" polynomials better. ET for sparse polynomials asks to check if a given polynomial is sparse in some coordinate system. More formally, given a polynomial $f$ as an arithmetic circuit and an $s \in \N$, decide if there is an $s$-sparse polynomial $g$ such that $f \sim g$. This problem was studied in \cite{GrigorievK93} over $\Q$, wherein an exponential in $n^4$ time algorithm was provided. There has not been any significant progress on this problem since that work. The lack of improvements in the complexity for over three decades makes one wonder:
\begin{center}
\emph{Is ET for sparse polynomials NP-hard?}
\end{center}
In this work, we answer this question in the affirmative over \emph{any} field (see the first part of Theorem \ref{Theorem: ETSparse NP-hard}) even if the input $f$ is provided as a depth-$2$ circuit. The result answers the question posed in \cite{GuptaST23, BDS24}. To our knowledge, the theorem gives the first example of a natural circuit class for which ET is provably hard.

Although ET for sparse polynomials (ETsparse) is a fairly natural problem, there is a deeper reason to study ETsparse that originates from the expressive power of affine projections of sparse polynomials and the \emph{Minimum Circuit Size Problem} (MCSP) for depth-$3$ circuits. We discuss this reason below to motivate ETsparse when the input is a \emph{homogeneous} polynomial.

\subsection{ETsparse and MCSP for depth-$3$ circuits}
First, we need a few definitions: A polynomial $g$ is an \emph{affine projection} of $f$ if $g = f(A\vecx + \vecb)$ for some $A \in \F^{|\vecx| \times |\vecx|}$ and $\vecb \in \F^{|\vecx|}$. If $\vecb = 0$, we say $g$ is a \emph{linear projection} of $f$; additionally, if $A \in \GL(|\vecx|)$, we say $g$ is in the \emph{orbit} of $f$, denoted as $\orb(f)$. Depth-$3$ circuits form a highly expressive class \cite{GKKS16, Tavenas15}. A depth-$3$ ($\Sigma \Pi \Sigma$) circuit is a circuit with a $+$ gate on top, a middle layer of $\times$ gates, and a bottom layer of $+$ gates. A depth-$3$ circuit with a top fan-in of $s$ is an affine projection of an $s$-sparse polynomial. Thus, the problem of deciding if a given $f$ is an affine projection of an $s$-sparse polynomial is closely related to MCSP for depth-$3$ circuits. We say ``closely related to" instead of ``the same as" because the size of a depth-$3$ circuit is determined by not only its top fan-in but also its formal degree. \\

\noindent{\textbf{MCSP}.} The complexity of MCSP for Boolean circuits has baffled researchers for over six decades. MCSP for a Boolean circuit class $\cC$ ($\cC$-MCSP) takes input the truth table of an $n$-variate Boolean function $f$ and a parameter $s \in \N$ and asks to check if $f$ is computable by a circuit in $\cC$ of size at most $s$. There are intriguing connections between MCSP and several other areas such as cryptography \cite{KabanetsC00, AllenderD17}, learning theory \cite{CarmosinoIKK16}, average-case complexity \cite{Hirahara18}, and proof complexity \cite{PichS19}. Whether or not MCSP for general Boolean circuits is NP-hard is a long-standing open question. It is known that MCSP is NP-hard for DNF \cite{Masek79, AllenderHMPS06} and DNF $\circ$ XOR formulas \cite{HiraharaOS18}. But no NP-hardness result is known (under deterministic polynomial-time reductions) for more general circuit models such as $\text{AC}^0$ circuits.\footnote{However, strong hardness results are known for several powerful circuit models under randomized or quasi-polynomial time or subexponential time reductions \cite{Ilango20, IlangoLO20, Ilango21, Hirahara22}.} This is not too surprising as \cite{KabanetsC00} showed that NP-hardness of $\cC$-MCSP under \emph{natural}\footnote{i.e., the size of the output of the reduction and the output parameter $s$ depend only on the size of the input instance. Almost all reductions that show NP-hardness of problems are natural.} deterministic polynomial-time reductions implies a $2^{\Omega(n)}$ lower bound for $\cC$, unless NP $\subseteq$ SUBEXP. Unfortunately, such strong lower bounds are not known even for depth-$3$ Boolean circuits. However, a $2^{\Omega(n)}$ lower bound is known for XOR $\circ$ AND $\circ$ XOR formulas \cite{Razborov87}, which are depth-$3$ \underline{arithmetic} circuits over $\F_2$ and are like DNF $\circ$ XOR formulas but with the top OR gate replaced by an XOR gate. In fact, a $2^{\Omega(n)}$ lower bound is known for depth-$3$ arithmetic circuits over any fixed finite field \cite{GrigorievR98}. This raises hope that we will be able to prove the hardness of MCSP for depth-$3$ arithmetic circuits over finite fields. But how is the input given in the case of MCSP for arithmetic circuits? And what about depth-$3$ circuits over fields of characteristic $0$? \\

\noindent{\textbf{MCSP for arithmetic circuits: Input representation and model of computation.}} In the Boolean setting of MCSP, one of the main reasons for assuming that the input is a truth table is that the assumption puts MCSP in NP. Analogously, in the algebraic setting, we could assume that the polynomial is given in the dense representation as a list of ${n+d \choose n}$ coefficients. But observe that even if the input is given as an arithmetic circuit, MCSP is in the complexity class MA over finite fields. This is because verifying if two circuits compute the same polynomial is the polynomial identity testing problem, which admits a randomized polynomial-time algorithm \cite{DemilloL78, Zippel79, Schwartz80}. Furthermore, class MA equals NP, assuming a widely believed circuit lower bound \cite{ImpagliazzoW97}. A succinct input representation also opens up the possibility of proving NP-hardness of MCSP for models, such as depth-$3$ circuits over fields of characteristic $0$, for which strong exponential lower bounds are unknown (the MCSP hardness to lower bound implication in \cite{KabanetsC00} needs the input in the dense format). The current best lower bound for depth-$3$ circuits over fields of characteristic $0$ is quasi-polynomial in $n$ \cite{Limaye0T21,AmireddyGKST23}.

It is, therefore, reasonable to assume that the input polynomial is given succinctly as a circuit which should only facilitate our efforts in proving NP-hardness of MCSP for arithmetic circuit classes. For example, there is an instance in the Boolean setting wherein succinct representation of the input helped prove NP-hardness of MCSP long before such a hardness result was shown with respect to the dense representation -- it is the case of the \emph{partial} MCSP problem \cite{HancockJLT96, Hirahara22}. In this work, we assume that the input is given as a depth-$2$ circuit, i.e., as a list of nonzero coefficients, and exponent vectors in unary -- this is the \emph{sparse representation}.\footnote{Sparse representations of polynomials are also used in computer algebra systems wherein the exponent vector is given in binary. As the degree is $n^{O(1)}$ in this work (except on one occasion; see the remark following Theorem \ref{Theorem: ETsupport NP-hard}), whether or not the exponent vector is given in unary or binary makes little difference.} 

A few remarks are in order concerning the model of computation. Over finite fields, we assume the Turing machine model. However, over arbitrary fields of characteristic $0$, it is natural to consider an arithmetic model of computation (similar to the Blum-Shub-Smale machine model \cite{BlumSS89}) that allows us to store a field element in unit space and perform an arithmetic operation in unit time. Over $\Q$, it is not clear if MCSP for arithmetic circuits is even decidable in the Turing machine model. But, if we confine our search to size-$s$ circuits whose field constants are $s^{O(1)}$ bit rational numbers, then we can work with the Turing machine model. \\

\noindent{\textbf{MCSP for homogeneous depth-$3$ circuits.}} The size of a $\Sigma\Pi\Sigma$ circuit is primarily determined by its formal degree and its top fan-in, whereas the size of a homogeneous depth-$3$ ($\hdthree$) circuit is mainly decided by its top fan-in (the formal degree of a $\Sigma\Pi\Sigma$ circuit is the maximum fan-in of the middle layer of $\times$ gates). MCSP for $\Sigma\Pi\Sigma$ circuits can be defined as follows: given $f$ and $D,s \in \N$, decide if there is a $\Sigma\Pi\Sigma$ circuit with formal degree bounded by $D$ and top fan-in bounded by $s$ that computes $f$. Similarly, MCSP for $\hdthree$ circuits is defined as: given a homogeneous $f$ and $s \in \N$, check if there is a $\hdthree$ circuit with top fan-in at most $s$ that computes $f$. In order to prove NP-hardness of $\Sigma\Pi\Sigma$-MCSP, it is \emph{necessary} to prove NP-hardness of $\hdthree$-MCSP. The reason is: a polynomial $f(x_1, x_2, \ldots, x_n)$ has a $\Sigma\Pi\Sigma$ circuit with formal degree bounded by $D$ and top fan-in bounded by $s$ if and only if the homogeneous polynomial $z^Df(x_1z^{-1}, x_2z^{-1}, \ldots, x_nz^{-1})$ has a $\hdthree$ circuit with top fan-in bounded by $s$. Also, if the reduction in a hypothetical proof of NP-hardness of $\hdthree$-MCSP has a certain simple feature, then it would imply NP-hardness of $\Sigma\Pi\Sigma$-MCSP (see the last remark following Proposition \ref{Proposition: bwd direction homogeneous}). Hence, it is natural to study the hardness of $\hdthree$-MCSP first.

NP-hardness of MCSP is known for two interesting subclasses of $\hdthree$ circuits, namely depth-$3$ powering circuits \cite{Shitov16} and set-multilinear $\Sigma\Pi\Sigma$ circuits \cite{Hastad90}; the top fan-in's of circuits in these two classes correspond to Waring rank and tensor rank, respectively. Perhaps an appealing evidence in favor of NP-hardness of $\hdthree$-MCSP is a proof of NP-hardness of MCSP for a ``dense" subclass of $\hdthree$ circuits. Intuitively, $\cC$ is a \emph{dense} subclass of $\hdthree$ circuits if every $\hdthree$ circuit can be approximated ``infinitesimally closely" by circuits in $\cC$.\footnote{Formally, a subclass $\cC$ of $\hdthree$ circuits is \emph{dense} if there are polynomial functions $p, q: \N \rightarrow \N$ such that the following holds: For $n, d, s \in \N$, the coefficient vector of every $n$-variate degree-$d$ polynomial computable by a size-$s$ $\hdthree$ circuit is in the \emph{Zariski closure} of the set of coefficient vectors of $p(nds)$-variate degree-$d$ polynomials computable by size-$q(nds)$ circuits in $\cC$. Here, ``size" means ``top fan-in".} Unfortunately, depth-$3$ powering circuits and set-multilinear $\Sigma\Pi\Sigma$ circuits are \emph{not} dense inside $\hdthree$ circuits.\footnote{Circuits of these two classes have small read-once algebraic branching programs (ROABPs), and the class ROABP is closed under Zariski closure \cite{Forbes16}. So, the closures of these two classes are also contained inside ROABPs. But, there are explicit $O(n)$ size $\hdthree$ circuits that require $2^{\Omega(n)}$ size ROABPs \cite{SahaT21, KayalNS20}.} On the other hand, \emph{orbits of homogeneous sparse polynomials} form a dense subclass of $\hdthree$ circuits.\footnote{Every $n$-variate degree-$d$ $\hdthree$ circuit of size-$s$ is a linear projection of an $s$-sparse degree-$d$ homogeneous polynomial in at most $sd$ variables. It is well known that linear projections of $f$ are contained in the Zariski closure of the orbit of $f$ over fields of characteristic $0$ (see \cite{SahaT21} for a proof of this fact).} It is natural to ask:

\begin{center}
\emph{Is MCSP for orbits of homogeneous sparse polynomials NP-hard?}
\end{center}
MCSP for orbits of homogeneous sparse polynomials is exactly the ETsparse problem on inputs that are homogeneous polynomials. The second part of Theorem \ref{Theorem: ETSparse NP-hard} answers the question positively over any field. \\ 

\noindent{\textbf{Approximating the sparse-orbit complexity.}} Call the smallest $s_0$ such that $f$ is in the orbit of an $s_0$-sparse polynomial, the \emph{sparse-orbit complexity} of $f$. Theorem \ref{Theorem: ETSparse NP-hard} shows that sparse-orbit complexity is hard to compute in the worst case. 
\begin{center}
\emph{Is sparse-orbit complexity easy to approximate?}
\end{center}
In Theorem \ref{Theorem: gap ETSparse NP-hard}, we show that approximating the sparse-orbit complexity of a given $s$-sparse polynomial (homogeneous or not) to within a $s^{1/3 - \epsilon}$ factor is NP-hard for any $\epsilon \in (0,1/3)$. As the input is $s$-sparse, approximating the sparse-orbit complexity to within a factor $s$ is trivial. 

\subsection{ET for constant-support polynomials}
ET is efficiently solvable for two special sparse polynomial families, namely the power symmetric polynomial $\PSym := x_1^d + \ldots + x_n^d$ \cite{Kayal11} and the sum-product polynomial $\SP : = \sum_{i \in [s]}{\prod_{j \in [d]}{x_{i,j}}}$ \cite{MediniS21, Kayal11}. What makes ET easy for these sparse polynomials? Explanations were provided in \cite{GuptaST23, BDS24}: $\SP$ is a read-once formula; it is also a $1$-design polynomial. $\PSym$ is a $1$-design polynomial, but it is also a support-$1$ polynomial.
\begin{center}
\emph{Is ET easy for constant-support polynomials?}
\end{center}
In Theorem \ref{Theorem: ETsupport NP-hard}, we show that checking if a given $f$ is in the orbit of a support-$5$ polynomial is NP-hard; this answers the question in the negative.

%% file: Results.tex
\subsection{Our results} \label{Section: Results}
We now state our results formally. The ETsparse problem is defined as follows.
\begin{problem}[$\et$] \label{Problem: ETSparse}
    Given a polynomial $f \in \F[\vecx]$ in its sparse representation and an integer $s$, check if there exist an $A \in \GL(|\vecx|,\F)$ and a $\vecb \in \F^{|\vecx|}$ such that $f(A\vecx + \vecb)$ is $s$-sparse.
\end{problem}
\noindent Our first result, Theorem \ref{Theorem: ETSparse NP-hard}, shows the $\NP$-hardness of $\et$ over any field.
\begin{theorem}[$\et$ is $\NP$-hard] \label{Theorem: ETSparse NP-hard}
\begin{enumerate}
    \item Let $\F$ be any field. There is a deterministic polynomial-time many-one reduction from $\SAT$ to $\et$ over $\F$.
    \item Let $\F$ be any field. There is a deterministic polynomial-time many-one reduction from $\SAT$ to $\et$ over $\F$ where the input polynomial to the $\et$ problem is homogeneous. 
\end{enumerate}
\end{theorem}
\begin{example}
        \item Part $2$ of the theorem subsumes part $1$. We state parts $1$ and $2$ separately because of two reasons: One, part $1$ has a simpler proof. Two, the degree parameters in the proof of part $1$ have a better upper bound in comparison to that in the proof of part $2$.

        \item The reduction is \emph{natural}\footnote{unless $\textnormal{char}(\F) = 2$. See the remark following Observation \ref{Obs: poly sparsity char 2 case} in Section \ref{section-inhomogeneous-finite}.} and has the feature that a satisfying assignment can be mapped to a sparsifying invertible $A \in \{-1,0,1\}^{|\vecx|\times |\vecx|}$ and vice versa. So, ETsparse is NP-hard even when $A$ is restricted to having only $\{-1,0,1\}$ entries.

        \item The authors of \cite{ChillaraGS23} showed the undecidability over $\Z$ of testing if a given $f$ is shift equivalent to some sparse polynomial ($f$ is shift equivalent to a polynomial $g$, if there exists a $\vecb \in \F^{|\vecx|}$ s.t $f = g(\vecx + \vecb)$). However, their result does not imply the intractability of ETsparse as testing shift equivalence to a sparse polynomial is a special case of $\et$ when $A$ is the identity map.
        
        \item The authors of \cite{BDS24} gave a randomized polynomial-time ET algorithm for \emph{random}\footnote{A random $s$-sparse degree-$d$ polynomial in their work was defined to be a polynomial where each monomial is formed independently of the others by selecting $d$ variables uniformly at random from the variable set; the coefficients are allowed to be arbitrary.} sparse polynomials, assuming black-box access to the input. Such average-case results for hard problems are not unusual in both algebraic and Boolean settings. In the algebraic setting, MCSP is $\NP$-hard for depth-$3$ powering circuits \cite{Shitov16} and for set-multilinear depth-$3$ circuits \cite{Hastad90}.\footnote{In a depth-$3$ powering circuit, each term is a power of a linear form. In a set-multilinear depth-$3$ circuit, the variable set is partitioned into $d$ sets such that each term is a product of $d$ linear forms, the $i\ith$ linear form being a linear form in the $i\ith$ set.} Yet, \cite{KayalS19} gave average-case learning algorithms for both these circuit models. In the Boolean setting, \cite{DyerF89} gave polynomial-time algorithms for average cases of $\NP$-hard problems like Graph $3$-colorability.

        \item Depth-$3$ power circuits, set-multilinear depth-$3$ circuits, and shifted sparse polynomials are all contained inside ROABPs. So, these models admit polynomial-time (improper) learning algorithms \cite{BeimelBBKV00, KlivansS06} and quasi-polynomial-time hitting sets \cite{AgrawalGKS15, ForbesS13}. Orbits of sparse polynomials require exponential size ROABPs \cite{SahaT21}; we cannot expect to improperly learn them via ROABPs. Theorem \ref{Theorem: ETSparse NP-hard} suggests that proper learning orbits of sparse polynomials is likely hard. Nonetheless, there is a quasi-polynomial time hitting set for orbits of sparse polynomials \cite{MediniS21, SahaT21}.        
\end{example}
We prove Theorem \ref{Theorem: ETSparse NP-hard} in Section \ref{Subsection: Reduction ETSparse}. Next, we define the gap version of $\et$.
\begin{problem}[$\gapet$] \label{Problem: gap ETSparse} 
Let $\alpha > 1$ be a parameter. Given a polynomial $f \in \F[\vecx]$ in its sparse representation and an integer $s_0$, output: 
\begin{itemize}
    \item YES, if there exist an $A \in \GL(|\vecx|,\F)$ and $\vecb \in \F$ such that $f(A\vecx + \vecb)$ is $s_0$-sparse.
    \item NO, if for all $A \in \GL(|\vecx|,\F)$ and $\vecb \in \F$, $f(A\vecx + \vecb)$ has sparsity at least $\alpha s_0$.
\end{itemize}
\end{problem}

Our second result, Theorem \ref{Theorem: gap ETSparse NP-hard}, shows that $\gapet$ is $\NP$-hard for $\alpha = s^{\frac{1}{3}-\epsilon}$, where $s$ is the sparsity of the input polynomial $f$ and $\epsilon \in (0,\frac{1}{3})$ is an arbitrary constant. Theorem \ref{Theorem: gap ETSparse NP-hard} is proven in Section \ref{Subsection: gap reduction}. From Theorem \ref{Theorem: gap ETSparse NP-hard}, we get Corollary \ref{Corollary: Sparse orbit complexity NP-hard} which states that $s^{\frac{1}{3}-\epsilon}$ factor approximation of the sparse-orbit complexity of an $s$-sparse polynomial is $\NP$-hard.

\begin{theorem}[$s^{\frac{1}{3}}\text{-gap-}\et$ is $\NP$-hard] \label{Theorem: gap ETSparse NP-hard}
Let $\epsilon \in (0,\frac{1}{3})$ be an arbitrary constant. 
\begin{enumerate}
    \item Let $\F$ be any field. There exists a deterministic polynomial-time many-one reduction from $\SAT$ to $s^{\frac{1}{3} - \epsilon}\text{-}\mathrm{gap}\text{-}\et$ over $\F$ where the input polynomial in $s^{\frac{1}{3} - \epsilon}\text{-}\mathrm{gap}\text{-}\et$ is $s$-sparse.

    \item Let $\F$ be any field. There exists a deterministic polynomial-time many-one reduction from $\SAT$ to $s^{\frac{1}{3} - \epsilon}\text{-}\mathrm{gap}\text{-}\et$ over $\F$ where the input polynomial in $s^{\frac{1}{3} - \epsilon}\text{-}\mathrm{gap}\text{-}\et$ is homogeneous and $s$-sparse.

\end{enumerate}
\end{theorem}
\begin{example}

    \item Like Theorem \ref{Theorem: ETSparse NP-hard}, part $2$ of Theorem \ref{Theorem: gap ETSparse NP-hard} subsumes part $1$. We state parts $1$ and $2$ separately because of two reasons. One, part $1$ has a simpler proof. Two, the degree parameters in the proof of part $1$ have a better bound in comparison to that in the proof of part $2$.

        \item It may be possible to improve the constant $\frac{1}{3}$ in $s^{\frac{1}{3}-\epsilon}$ using a more careful analysis.

        \item Interestingly, the above results are obtained without invoking the celebrated PCP theorem \cite{AroraS98, AroraLMSS98, Dinur07}.

\end{example}
\begin{corollary} \label{Corollary: Sparse orbit complexity NP-hard}
Let $0 < \epsilon < \frac{1}{3}$ be an arbitrary constant. 
\begin{enumerate}
    \item Let $\F$ be any field. It is $\NP$-hard to compute $s^{\frac{1}{3} - \epsilon}$ factor approximation of the sparse-orbit complexity when the input is an $s$-sparse polynomial over $\F$.
    \item Let $\F$ be any field. It is $\NP$-hard to compute $s^{\frac{1}{3} - \epsilon}$ factor approximation of the sparse-orbit complexity when the input is an $s$-sparse homogeneous polynomial over $\F$.
\end{enumerate}
\end{corollary}
\begin{example}
     Thus, approximating the sparse-orbit complexity within a certain super-constant factor is $\NP$-hard over any field. In contrast, \cite{SongWZ17, BlaserIJL18,Swernofsky18} showed that approximating the tensor rank (which corresponds to the smallest top fan-in of a set-multilinear depth-$3$ circuit) within a $1+\delta$ factor, where $\delta \approx 0.0005$, is $\NP$-hard over any field. We do not know of any hardness of approximation result for the Waring rank (which corresponds to the smallest top fan-in of a depth-$3$ powering circuit).
\end{example}

Now, we formally define the support of a polynomial.
\begin{definition}[Support of a polynomial] \label{Definition: support}
For a monomial $\vecx^{\vecalpha}$, where $\vecalpha$ is the exponent vector, the support of $\vecx^{\vecalpha}$, $\supp(\vecx^{\vecalpha})$, is the number of variables with non-zero exponent. The support of a polynomial $f$, $\supp(f)$, is the maximum support size over all the monomials of $f$.
\end{definition}
Thus, a polynomial has support $\sigma$ if there exists a monomial with support $\sigma$ and no other monomial has support $> \sigma$. The ET problem for constant-support polynomials and a stronger version of it are defined next (henceforth, $\sigma$ is assumed to be a constant). 
\begin{problem}[$\etcsp$] \label{Problem: ETcsp}
    Given a polynomial $f \in \F[\vecx]$ in its sparse representation and an integer $\sigma$, check if there exists an $A \in \GL(|\vecx|, \F)$ such that $\supp(f(A\vecx)) \leq \sigma$.
\end{problem}
\begin{problem}[$\prometcsp$] \label{Problem: Promise ETcsp}
    Given a polynomial $f \in \F[\vecx]$ with support $\sigma + 1$ in its sparse representation, check if there exists an $A \in \GL(|\vecx|, \F)$ such that $\supp(f(A\vecx)) \leq \sigma$.
\end{problem}
\begin{example}
    \item Unlike $\et$, checking if $f$ is in the \emph{orbit} of a constant-support polynomial is the same as checking if $f$ is equivalent to a constant-support polynomial. This follows from the observation that $\supp(f(\vecx)) = \supp(f(\vecx+\vecb))$ for any $\vecb \in \F^{|\vecx|}$.
\end{example}

Our third result, Theorem \ref{Theorem: ETsupport NP-hard}, shows that $\etcsp$ and $\prometcsp$ are $\NP$-hard. We prove Theorem \ref{Theorem: ETsupport NP-hard} in Section \ref{Section: ETsupport NP-hard}.

\begin{theorem}[$\etcsp$ is $\NP$-hard] \label{Theorem: ETsupport NP-hard}
Let $\sigma \geq 5$ be a constant and $\F$ be a field with $\textnormal{char}(\F) = 0$ or $> \sigma + 1$. There is a deterministic polynomial-time many-one reduction from $\SAT$ to $\etcsp$ over $\F$. In particular, $\SAT$ reduces to $\prometcsp$ in deterministic polynomial time.
\end{theorem}
\begin{example} \label{Remarks: Theorem ETsupport}
        \item Over fields of finite characteristic, it is assumed that the exponent vectors corresponding to the monomials of the input polynomial are given in binary.
\end{example}
We now consider a variant of $\et$ and $\gapet$ which involves testing equivalence to sparse polynomials under only translations. Formally,

\begin{problem}[$\sparseshift$] \label{Problem: Shift equivalence}
Given a polynomial $f \in \F[\vecx]$ in its sparse representation and an integer $s$, decide if there exists a $\vecb \in \F^{|\vecx|}$, such that $f(\vecx + \vecb)$ is $s$-sparse.
\end{problem}

\begin{problem}[$\gapsparseshift$] \label{Problem: Gap Shift equivalence}
Let $\alpha$ be a parameter $> 1$. Given a polynomial $f \in \F[\vecx]$ in its sparse representation and an integer $s_0$, output:
\begin{itemize}
    \item YES, if there exists a $ \vecb \in \F^{|\vecx|}$ such that $f(\vecx + \vecb)$ has at most $s_0$ monomials.
    \item NO, if for all $\vecb \in \F^{|\vecx|}$, $f(\vecx + \vecb)$ has sparsity at least $\alpha s_0$.
\end{itemize}
\end{problem}
The problem of testing equivalence to polynomials under only translations has been studied and called the Shift Equivalence Testing problem (SET) in \cite{DvirOS14}, based on which we use the name $\sparseshift$. Refer Section \ref{sec:related work} for a more detailed discussion on SET and $\sparseshift$. Briefly, the authors of \cite{LakshmanS94} studied $\sparseshift$ for univariate polynomials and gave efficient algorithms for it along with a criterion for the existence of such shifts. Later, the authors of \cite{GrigorievL00} extended the criteria to multivariate polynomials and also gave deterministic and randomized algorithms for $\sparseshift$, where if the input is an $n$-variate degree-$d$ polynomial and an integer $s$, then the running time is super-polynomial in $d,s$ and $n$ for both algorithms when $n$ is not a constant. The lack of progress in developing efficient algorithms for $\sparseshift$ motivated the authors of \cite{ChillaraGS23} to study the hardness of $\sparseshift$. More precisely, they studied the problem of deciding, for a given polynomial $f(\vecx)$, the existence of a $\vecb \in \F^{|\vecx|}$ such that $f(\vecx+\vecb)$ has strictly lesser monomials than $f(\vecx)$. Note that this is a variant of $\sparseshift$, which we will refer to as $s\text{-to-}(s-1)\sparseshift$, where the inputs are a polynomial of sparsity $s$ and an integer $s-1$. They also studied $\gapsparseshift$.

In our fourth and last result, we adapt the proofs of Theorems \ref{Theorem: ETSparse NP-hard} and \ref{Theorem: gap ETSparse NP-hard} to show the NP-hardness of $\sparseshift$ and $\gapsparseshift$ (for super-constant $\alpha$) over any integral domain including fields. We highlight the main differences between our results and those of \cite{ChillaraGS23} in the remarks following the theorem. A more detailed comparison is made in Section \ref{Sec:SETsparse prev work compare}.
\begin{theorem}[$\sparseshift$ and $s^{\frac{1}{3} - \epsilon}\text{-}\mathrm{gap}\text{-}\sparseshift$ are $\NP$-hard.] \label{Theorem: ETSparse_shift NP-hard}
Let $R$ be an integral domain that can also be a field. Let $\epsilon \in (0,1/3)$ be an arbitrary constant.
\begin{enumerate}
    \item There is a deterministic polynomial-time many-one reduction from $\SAT$ to $\sparseshift$ over $R$.
    \item There is a deterministic polynomial-time many-one reduction from $\SAT$ to $s^{\frac{1}{3} - \epsilon}\text{-}\mathrm{gap}\text{-}\sparseshift$ over $R$ where the input polynomial in $s^{\frac{1}{3} - \epsilon}\text{-}\mathrm{gap}\text{-}\sparseshift$ is $s$-sparse.
\end{enumerate}
\end{theorem}
\begin{example}
    \item The authors of \cite{ChillaraGS23} showed that polynomial solvability reduces to $s\text{-to-}(s-1)\sparseshift$ over any integral domain which is \emph{not} a field. This implies that $\sparseshift$ is at least as hard as polynomial solvability over integral domains but not fields. In contrast, we show $\sparseshift$ is $\NP$-hard over any integral domain \emph{including} fields.
    \item The authors of \cite{ChillaraGS23} also showed that for any \emph{constant} $\alpha > 1$, $\gapsparseshift$ is $\NP$-hard over $\R,\Q,\F_p$ and $\Z_q$ (Abelian group with $q$ elements) when the input polynomial is in the sparse representation, and undecidable over $\Z$ for $\alpha = s^{o(1)}$, where $s$ is the sparsity of the input polynomial given in the sparse representation. Both their results use gap amplification, and the result over fields additionally invokes the PCP theorem. In contrast, we show $s^{\frac{1}{3} - \epsilon}\text{-}\mathrm{gap}\text{-}\sparseshift$, where $s^{\frac{1}{3} - \epsilon}$ is super-constant, is $\NP$-hard over any integral domain \emph{including} fields. We prove our result without using gap amplification or invoking the PCP theorem.  
    \item Like in Theorem \ref{Theorem: ETSparse NP-hard}, the reduction is natural, except when the characteristic is $2$. In contrast, the reduction in \cite{ChillaraGS23} is not natural.

\end{example}
\begin{corollary} \label{Corollary: Sparse orbit shift complexity NP-hard}
Let $0 < \epsilon < \frac{1}{3}$ be an arbitrary constant and $R$ be any integral domain. It is $\NP$-hard to compute $s^{\frac{1}{3} - \epsilon}$ factor approximation of the sparse-orbit complexity under only translations when the input is an $s$-sparse polynomial over $R$.
\end{corollary}
We prove Theorems \ref{Theorem: ETSparse NP-hard}, \ref{Theorem: gap ETSparse NP-hard}, \ref{Theorem: ETsupport NP-hard}, and parts 1 and 2 of Theorem \ref{Theorem: ETSparse_shift NP-hard} in Sections \ref{Subsection: Reduction ETSparse}, \ref{Subsection: gap reduction}, \ref{Section: ETsupport NP-hard}, \ref{section-sparseshiftNPH} and \ref{section-gapsparseshiftNPH} respectively. We give proof sketches of the respective theorems at the beginning of the sections. This paper is an extended version of \cite{BDSS24}. In this extended version, we prove Theorem \ref{Theorem: ETSparse_shift NP-hard}, which does not appear in \cite{BDSS24}, and strengthen Theorem \ref{Theorem: ETsupport NP-hard} by proving it for $\sigma \geq 5$ instead of $\sigma \geq 6$.

%% file: related_works.tex
\subsection{Related work} \label{sec:related work}
\paragraph{Results on ET.} As mentioned in Section \ref{Section: Introduction}, the study of ET was initiated in \cite{Kayal11} where efficient ET algorithms were given for the power symmetric and the elementary symmetric polynomials. Following this, efficient ET algorithms were given for several other important polynomial families and circuit classes such as the permanent \cite{Kayal12}, the determinant \cite{Kayal12, grochowPhD, GargGK019}, the iterated matrix multiplication (IMM) polynomial \cite{KayalNST17,MurthyNS20}, the continuant polynomial \cite{MediniS21}, read-once formulas (ROFs) \cite{GuptaST23}, and design polynomials \cite{BDS24,GuptaS19}. ET algorithms have also been used to give efficient reconstruction algorithms; for example, \cite{KayalNS19} gave an efficient average-case reconstruction algorithm for low-width ABPs based on ET for the determinant. 

The sum-product polynomial $\SP : = \sum_{i \in [s]}{\prod_{j \in [d]}{x_{i,j}}}$ is a rare example for which three different ET algorithms are known. The SP polynomial can be computed by an ROF. So, the ET algorithm for ROFs \cite{GuptaST23, Kayal11}, which is based on analyzing the Hessian determinant, gives ET for SP. Also, SP is a design polynomial, so the ET algorithm of \cite{BDS24}, which uses the vector space decomposition framework of \cite{KayalS19, GKS20}, holds for SP. The authors of \cite{MediniS21} also observed that ET for SP follows from the reconstruction algorithm in \cite{KayalS19}. A third ET algorithm for SP can be designed by analyzing its Lie algebra. Observe that the orbit of SP is a dense subclass of homogeneous depth-$3$ circuits. However, as ET for SP is easy, it does not provide any supporting evidence for the hardness of MCSP for homogeneous depth-$3$ circuits. 

\paragraph{Results on PE.} \label{PE variants} Quadratic form equivalence can be solved in polynomial time over $\R$, $\C$, finite fields and $\Q$ (assuming access to integer factoring oracle) \cite{Saxenaphd,Wallenborn13}. These algorithms are based on well-known classification of quadratic forms \cite{Lam04,Ara11}. In contrast, \cite{AgrawalS05} showed that cubic form equivalence (CFE) is at least as hard as graph isomorphism. The authors of \cite{GrochowQ23a} showed that CFE is polynomial time equivalent to several other problems like group isomorphism for $p$-groups, algebra isomorphism, trilinear form equivalence, etc.

A variant of PE is the Shift Equivalence Testing problem (SET), as it was called by \cite{DvirOS14}, where given two $n$-variate polynomials $f$ and $g$, one needs to check if there exists $\vecb \in \F^n$ such that $f(\vecx) = g(\vecx + \vecb)$. The author of \cite{Grigoriev97} gave a deterministic algorithm over characteristic $0$ fields, a randomized algorithm over prime residue fields and a quantum algorithm over characteristic $2$ fields for SET. All these algorithms have running time polynomial in the dense representation of the input, that is, for $n$-variate, degree-$d$ polynomials given in the verbose representation as input, the running time is $\poly(\binom{n+d}{d})$. The authors of \cite{DvirOS14} gave a randomized algorithm for SET assuming black-box access to $n$-variate polynomials $f$ and $g$ with degree bound $d$ and circuit size bound $s$. Their algorithm runs in $\poly(n,d,s)$ time. Another randomized polynomial-time algorithm for SET is given in \cite{Kayal12}.

A variant of SET, which we call $\sparseshift$, is where a single $n$-variate polynomial $f(\vecx)$ and a positive integer $t$ are given as inputs, and the objective is to decide if there exists $\vecb \in \F^n$ such that $f(\vecx+\vecb)$ is $t$-sparse. The authors of \cite{LakshmanS94} studied $\sparseshift$ for univariate polynomials over $\Q$ and gave sufficient conditions for the uniqueness and rationality of a $t$-sparsifying shift. The authors of \cite{GrigorievL00} extended these conditions to multivariate polynomials and gave two algorithms for computing $t$-sparsifying shifts for $n$-variate, degree-$d$ polynomials, one where the input polynomial has finitely many $t$-sparsifying shifts and the other for polynomials without any finiteness restriction on the number of $t$-sparsifying shifts. The running time of the first algorithm is $(dt)^{O(n)}$ without randomization and $t^{O(n)}$ with randomization, while that of the second one is $(nt)^{O(n^2)}$. As mentioned in the remarks following Theorem \ref{Theorem: ETSparse_shift NP-hard}, the authors of \cite{ChillaraGS23} showed that $s\text{-to-}(s-1)\sparseshift$ is undecidable over $\Z$ by showing a reduction from polynomial solvability over $\Z$ to $\sparseshift$. They also showed that $\gapsparseshift$ is $\NP$-hard over $\R,\Q, \F_p$ and $\Z_q$ for constant $\alpha$ and undecidable over $\Z$ for suitable $\alpha$. A detailed comparison with their results is presented in Section \ref{Sec:SETsparse prev work compare}.

The scaling equivalence problem is yet another variant of PE, which involves checking for given $n$-variate polynomials $f$ and $g$ whether there exists a diagonal matrix $S \in \GL(n,\F)$ such that $f(\vecx) = g(S\vecx)$. The authors of \cite{BlaserRS17} gave a randomized polynomial-time algorithm for the scaling equivalence problem over $\R$. 

\paragraph{Hardness results.} The author of \cite{Kayal12} showed that the problem of checking if a polynomial is an affine projection of another polynomial is $\NP$-hard via a reduction from Graph $3$-Colorability. Computing the tensor rank (which is MCSP for depth-$3$ set multilinear circuits) is $\NP$-hard \cite{Hastad90}, so is computing the Waring rank for a polynomial (which is MCSP for depth-$3$ powering circuits) \cite{Shitov16}. In the Boolean world, \cite{KhotS08} showed that there is no polynomial-time algorithm to $n^{1-\delta}$-approximate, where $\delta > 0$ is an arbitrarily small constant, a $\mathrm{DNF}$ with minimum number of terms for any $n$-variate Boolean function given as a truth table, unless $\NP$ is decidable in quasi-polynomial time. It is also known that $(1+\delta)$-approximate MCSP, where $\delta \approx 0.0005$ is a constant, is $\NP$-hard for set-multilinear depth three circuits \cite{SongWZ17,Swernofsky18,BlaserIJL18}. In \cite{KlivansS09}, it was shown that depth-$3$ arithmetic circuits cannot be PAC-learned in polynomial time unless the length of a shortest nonzero vector of an $n$-dimensional lattice can be approximated to within a factor of $\tilde{O}(n^{1.5})$ in polynomial time by a quantum algorithm. This means it is hard to PAC-learn the class of Boolean functions that match the output of depth-$3$ arithmetic circuits on the Boolean hypercube. 

\paragraph{Hitting sets and lower bounds for orbits of sparse polynomials.} The authors of \cite{MediniS21} gave a quasi-polynomial time hitting set\footnote{A hitting set for a circuit class $\mathcal{C}$ is a set $S \subseteq \F^{|\vecx|}$ such that for every non-zero polynomial $f(\vecx)$ computable by a circuit $C \in \mathcal{C}$, $f(\veca) \neq 0$ for some $\veca \in S$.} construction for the orbits of sparse polynomials. Orbits of sparse polynomials form a subclass of homogeneous depth-$3$ circuits. The authors of \cite{NisanW97} showed that any homogeneous depth-$3$ circuit computing the $n$-variate elementary symmetric polynomial of degree $2d$ has size  $\Omega((\frac{n}{4d})^d)$. The authors of \cite{KayalST16} showed the existence of an explicit polynomial family in $n$ variables and degree $d$, with $d \geq n$, for which any homogeneous depth-$3$ circuit computing it must be of size at least $2^{\Omega(n)}$.

%% file: roadmap.tex
\subsection{Roadmap of the paper}
In Section \ref{sec:prelim}, we state a few useful observations and claims, the proofs of which appear in Section \ref{Section: Proofs prelim} of the appendix. The proof of part one of Theorem \ref{Theorem: ETSparse NP-hard} for fields of characteristic zero is given in Sections \ref{subsubsection: sparse f construction}-\ref{subsubsection: sparse bwd direction}. Section \ref{section-ETSparse-homogenous} has the proof of part two of the same theorem for characteristic zero fields. In Section \ref{section-extension-finite}, we prove Theorem \ref{Theorem: ETSparse NP-hard} for fields of finite characteristics. Similarly, the proofs of parts one and two of Theorem \ref{Theorem: gap ETSparse NP-hard} for characteristic zero fields appear in Sections \ref{sec:gap non-homog} and \ref{sec:gap homog}, respectively. Section \ref{sec:gap finite char} contains the proof of Theorem \ref{Theorem: gap ETSparse NP-hard} over fields of finite characteristics. In Section \ref{Section: ETsupport NP-hard}, we prove Theorem \ref{Theorem: ETsupport NP-hard}. For simplicity, we ignore the effect of translation vectors in the above-mentioned sections. In Section \ref{Section: Translations} of the appendix, we show how to handle translation vectors. In Sections \ref{section-sparseshiftNPH} and \ref{section-gapsparseshiftNPH}, we prove parts one and two of Theorem \ref{Theorem: ETSparse_shift NP-hard} respectively. The missing proofs of the observations, claims, lemmas, and propositions in Sections \ref{Subsection: Reduction ETSparse}, \ref{Subsection: gap reduction}, \ref{Section: ETsupport NP-hard} and \ref{section-shifteqvhard} appear in Sections \ref{Section: ETSparse proofs}, \ref{Section: Gap ETSparse proofs}, \ref{Section: ETsupport proofs} and \ref{Section: sparse-shift proofs} of the appendix, respectively. 

%% file: Prem.tex
\section{Preliminaries} \label{sec:prelim}
\subsection{Definitions and notations} \label{Section: Definitions}
For $n,a,b \in \N$, $[n]$ denotes the set $\{1,2 \dots, n\}$ and $[a,b]$ denotes the integers from $a$ to $b$, both inclusive. A polynomial is \emph{homogeneous} if all its monomials have the same total degree. The set of invertible linear transforms in $n$ variables over a field $\F$ is denoted by $\GL(n,\F)$. For a polynomial $f \in \F[\vecx]$, the action of a linear transform $A \in \F^{|\vecx| \times |\vecx|}$ on its variables is denoted by $f(A\vecx)$ as well as by $A(f)$. The \emph{sparsity} of a polynomial $f$, denoted as $\mathcal{S}(f)$, is the number of monomials in $f$ with non-zero coefficients. For a polynomial $f$, $\var(f)$ denotes the set of variables that occur in at least one monomial of $f$.  We have used the notation $f \sim g$ earlier to denote $f = g(A\vecx + \vecb)$. Henceforth, we will ignore the translation vector $\vecb$ in the main body of the discussion for simplicity but mention the necessary changes in the proofs or point to appropriate sections when translations are involved. Thus, for polynomials $f$ and $g$, $f \sim g$ will mean $f(\vecx) = g(A\vecx)$ where $A \in \GL(|\vecx|,\F)$. Similarly, the \emph{orbit} of a polynomial $f$ will denote the set $\{f(A\vecx), A \in \GL(|\vecx|,\F)\}$. The \emph{degree} of a monomial is its total degree, and the degree of a polynomial $f$ is the maximum degree amongst all monomials in $f$. The $x$-degree of a monomial is the degree of the variable $x$ in the monomial.

\begin{definition}[Degree separated polynomials]
Polynomials $f$ and $g$ are \textit{degree separated} if no monomial of $f$ has the same degree as a monomial of $g$. Similarly, $f$ and $g$ are degree separated \emph{with respect to a variable $x$} if no monomial of $f$ has the same $x$-degree as a monomial of $g$.
\end{definition}
The \emph{set of degrees} of a polynomial is the set of distinct degrees of all the monomials in the polynomial. For example, the set of degrees of $f(x_1,x_2) = x_1^2 + x_1x_2 + 4x_2$ is $\{2,1\}$. A \emph{linear form} is a homogeneous degree one polynomial. An \emph{affine form} is a degree one polynomial.

\subsection{Algebraic preliminaries} \label{Section: Preliminaries}
The proofs of the observations and claims stated in this section can be found in Appendix \ref{Section: Proofs prelim}.
\begin{observation} \label{Lemma: equiv degrees same}
    Let $f$ and $g$ be polynomials such that $f \sim g$. Then, $f$ and $g$ have the same set of degrees for the monomials. Thus, if $f$ and $g$ are degree separated, then $f \not \sim g$.
\end{observation}
\begin{observation} \label{lem_deg_sep_sum}
    If $f$ and $g$ are degree separated (or degree separated with respect to some variable), then $\cal S(f+g)= \cal S (f) +\cal S(g)$.
\end{observation}
\begin{observation} \label{Lemma: equiv deg sep}
    If $f$ and $g$ are degree separated, $f_1 \sim f$ and $g_1 \sim g$, then $\cal S(f_1+g_1) = \cal S(f_1)+\cal S(g_1)$.
\end{observation}

Observation \ref{Lemma: linear form sparsity} analyzes the sparsity of powers of linear forms. Observation \ref{lem_binom} is a special case of Observation \ref{Lemma: linear form sparsity} and is stated separately because it is simpler and is invoked many times. Observation \ref{Lemma: affine form sparsity} analyzes the sparsity of powers of affine forms. Note that all of these observations, along with Observation \ref{lem_deg_sep_sum}, also hold more generally over integral domains.\footnote{Integral domains are commutative rings with a multiplicative identity and no zero divisors.} This property of the aforementioned observations is used in proving Theorem \ref{Theorem: ETSparse_shift NP-hard}.

\begin{observation} \label{Lemma: linear form sparsity}
Let $\ell$ be a linear form in $m$ variables and $d \in \N$. If $\textnormal{char}(\F) = 0$, $\cal S(\ell^d) = \binom{d+m-1}{m-1}$, and if $\textnormal{char}(\F) = p$, $\cal S(\ell^d) = \prod_{i=0}^{k}\binom{e_i+m-1}{m-1}$, where $d = \sum_{i=0}^{k} e_ip^i$, $e_i \in [0,p-1]$. 
\end{observation}
\begin{observation} \label{lem_binom}
    If $\textnormal{char}(\F) = 0$ and $\ell$ be a linear form in exactly two variables, then $\cal S(\ell^d)=d+1$. The result holds for characteristic $p$ fields if $p > d$ or if $d = p^k-1$ for some $k\in \N$. Further, if $\ell$ is a linear form in more than two variables and $d$ is as before, then $\cal S(\ell^d)\geq d+1$.
\end{observation}

\begin{observation}
\label{Lemma: affine form sparsity}
    Let $h = \ell + c_0$, where $\ell$ is a linear form in at least one variable and $c_0 \in \F \backslash \{0\}$, then $\cal S(h^d) \geq \cal S(\ell^d) + 1$. More precisely, $\cal S(h^d) \geq d+1$ holds if $\textnormal{char}(\F) = 0$ or if $\textnormal{char}(\F) = p$ and $p > d$ or $d = p^k-1$ for some $k\in \N$.
\end{observation}

Claim \ref{lem_div} analyzes the sparsity of polynomials divisible by a power of some linear form in at least two variables and is used to prove part two of Theorems \ref{Theorem: ETSparse NP-hard} and \ref{Theorem: gap ETSparse NP-hard}. Claim \ref{FullSupport} analyzes the support of monomials under invertible linear transforms and is used to prove Theorem \ref{Theorem: ETsupport NP-hard}.  

\begin{claim}\label{lem_div}
 Let $\textnormal{char}(\F) = 0$. If $f \in \F[\vecx]$ is a non-zero polynomial divisible by $\ell^d$ for some linear form $\ell$ in at least two variables, then $\cal S(f)\geq d+1$. The claim also holds for characteristic $p$ fields, where the degree of $f$ is less than $p$.
\end{claim}

\begin{claim} \label{FullSupport}
     Let $\sigma, d, n \in \N$, $d \geq \sigma$, $f =  (x_1\cdots x_n)^d$, and $\ell_1,\dots,\ell_n$ be linearly independent linear forms in $x_1, \ldots, x_n$. If $|\cup_{i=1}^{n}\var(\ell_i)| \geq \sigma$ and $g := f(\ell_1\cdots \ell_n)$, then $\supp(g) \geq \sigma$. The claim holds if $\textnormal{char}(\F) = 0$, or $\textnormal{char}(\F) = p$ with $p > d$, or $p > \sigma$ and $d = p^k -1$ for some $k \in \N$.
\end{claim}

%% file: ETSparse_NPhard.tex
\section{$\NP$-hardness of $\et$} \label{Subsection: Reduction ETSparse}
In this section, we prove Theorem \ref{Theorem: ETSparse NP-hard}. We first show the reduction over characteristic $0$ fields in the non-homogeneous case without considering translations for ease of understanding. Section \ref{section-ETSparse-homogenous} shows the reduction over characteristic $0$ fields in the homogeneous case. In Section \ref{section-extension-finite}, the reduction is shown to hold over finite characteristic fields for both the non-homogeneous and the homogeneous case. In Appendix \ref{Section: ETSparse proofs}, we prove the lemmas and the observations of this section. Appendix \ref{Section: Affine ETsparse} shows how the reduction holds while also considering translations.\footnote{Note that for two homogeneous polynomials $f$ and $g$, $f(\vecx) = g(A\vecx + \vecb)$ implies $f(\vecx) = g(A\vecx)$, where $A \in \GL(|\vecx|,\F)$ and $\vecb \in \F^{|\vecx|}$. Hence, it suffices to prove part 2 of Theorem \ref{Theorem: ETSparse NP-hard} without translations. \label{footnote: homogeneous no translate}}

\paragraph{Proof sketch.}The reduction maps each variable and clause of a $\CNF$\footnote{We assume, without loss of generality, that each clause of a $\CNF$ has $3$ distinct variables. This can be achieved by introducing extra variables for clauses with $< 3$ variables. } $\psi$ to distinct degree separated polynomials which, summed together, give the polynomial $f$. As the summands are degree separated, the sparsity of $f$ under invertible transforms can be analyzed by doing so for individual polynomials. The degrees are chosen such that $f$ is equivalent to an $s$-sparse polynomial (for a suitable sparsity parameter $s$) if and only if $\psi \in \SAT$. 

\subsection{Constructing $f$ and $s$} \label{subsubsection: sparse f construction}
Let $\psi$ be a $\CNF$ in variables $\vecx := \{x_1, x_2 \dots x_n\}$ and $m$ clauses:
\[ \psi = \land_{k=1}^{m} \lor_{j \in C_k} (x_j \oplus a_{k,j}), \] 
where $C_k$ denotes the set of indices of the variables in the $k\ith$ clause and $a_{k,j} \in \{0,1\}$. Let $\vecy := \{y_1,y_2 \dots y_n \}$, $x_0$ be a new variable and $\vecz := \{x_0\} \sqcup \vecx \sqcup \vecy$. For $d_1, d_2, d_3, d_4 \in \N$, consider the following polynomials:

\begin{itemize}
    \item Corresponding to variable $x_i$, where $i\in [n]$, define $Q_i(\vecz)$ as: 
\begin{equation*} \label{cl_poly_2}
\begin{split}
    Q_i(\vecz) &:= Q_{i,1}(\vecz) + Q_{i,2}(\vecz) + Q_{i,3}(\vecz), \text{ where}\\
    Q_{i,1}(\vecz) &:= x_0^{(3i-2)d_1}x_i^{d_2}, \ Q_{i,2}(\vecz) := x_0^{(3i-1)d_1}(y_i+x_i)^{d_3} \text{ and } Q_{i,3}(\vecz) := x_0^{3id_1}(y_i-x_i)^{d_3}. 
\end{split}
\end{equation*}
Intuitively, $Q_{i,2}$ and $Q_{i,3}$ correspond to assigning $0$ and $1$, respectively, to $x_i$ in $\psi$. $Q_{i,1}$ is used to establish a mapping between satisfying assignments and sparsifying transforms. 
\item For the $k\ith$ clause, $k \in [m]$,  define $R_k(\vecz) := x_0^{(3n+k)d_1}\prod_{j\in C_k}(y_j+ (-1)^{a_{k,j}}x_j)^{d_4}.$ 
\end{itemize}
Define $s := 1+n(3+d_3) + m(d_4+1)^2$ and the polynomial $f$ as:
\begin{equation} \label{Definition: 3SAT poly}
f(\vecz) := x_0^{d_1} + \sum_{i=1}^{n} Q_i(\vecz) + \sum_{k=1}^{m} R_k(\vecz).
\end{equation}
The following conditions are imposed on the $d_i$'s:
\begin{equation} \label{ineq}
            \begin{split}
            d_1&\geq \max(s,d_2+1), \  d_2 \geq 2d_3, \ d_3 \geq m(d_4+1)^2 + 1, \text{ and } d_4 \geq m.
            \end{split}
\end{equation}
For characteristic $0$ fields, the inequalities of \eqref{ineq} can be converted to equalities. Thus, we get
\begin{equation} \label{di-inhomogeneous-char0}
    \begin{split}
        d_4 &= m, \ d_3 = m(m+1)^2+1=O(m^3)  \implies s= O(nm^3)\\
        d_2 &= 2m(m+1)^2+2=O(m^3), \ d_1 = 1+n(4+m(m+1)^2)+m(m+1)^2=O(nm^3).
    \end{split}
\end{equation}
Note, for the above choices of $d_3$ and $d_2$, $s \geq d_2 + 1$. Hence, $d_1$ is set to $s$. Under the conditions of \eqref{ineq} the following observations hold.
\begin{observation} \label{Obs: deg separated polys}
For all $i \in [n]$, $k \in [m]$, the polynomials $x_0^{d_1}$, $Q_{i,1}(\vecz)$, $Q_{i,2}(\vecz)$, $Q_{i,3}(\vecz)$ and $R_k(\vecz)$ are degree separated from one another. Also, $Q_i(\vecz)$ is degree separated from other $Q_j(\vecz)$'s, for $i,j \in [n]$ and $i \neq j$. Similarly, $R_k(\vecz)$ is degree separated from $R_l(\vecz)$ for $k,l \in [m]$ and $k \neq l$.
\end{observation}
\begin{observation} \label{Obs: ETSparse poly degree}
The degree of $f$ is $(3n+m)d_1 + 3d_4 = (mn)^{O(1)}$.
\end{observation}
\begin{observation} \label{Obs: ETsparse poly sparsity }
$\cal S(f(\vecz)) = 1 + n(2d_3 + 3) + m(d_4+1)^3$ and $\supp(f) = 7$.
\end{observation}
\subsection{The forward direction} \label{subsubsection: sparse fwd direction}
Proposition \ref{Proposition: fwd direction} shows how a satisfiable $\psi$ implies the existence of an invertible $A$, such that $\cal S(f(A\vecz)) \leq s$ by constructing $A$ from a satisfying assignment $\vecu \in \{0,1\}^n$ of $\psi$.  
\begin{proposition} \label{Proposition: fwd direction}
Let $\vecu = (u_1, \dots ,u_n) \in \{0,1\}^n$ be such that $\psi(\vecu) = 1$. Then $\cal S(f(A\vecz)) \leq s$, where $A$ is as:
    \begin{equation} \label{action}
        \begin{split}
           A: x_0 \mapsto x_0, x_i\mapsto x_i, y_i\mapsto y_i +(-1)^{u_i}x_i, \ \ \  \forall i\in [n].
        \end{split}
    \end{equation}
\end{proposition}
\begin{proof}
It follows from the definition of $f$ in \eqref{Definition: 3SAT poly}, Observations \ref{Obs: deg separated polys} and \ref{Lemma: equiv deg sep} that 
\[\cal S(f(A\vecz)) = \cal S(A(x_0^{d_1})) + \sum_{i=1}^{n} \cal S(Q_i(A\vecz)) + \sum_{k=1}^{m} \cal S(R_k(A\vecz)).\]
Thus, it suffices to analyze the sparsity of $A(x_0^{d_1}), Q_i(A\vecz)$'s and $R_k(A\vecz)$'s. Now, $\cal S(A(x_0^{d_1})) = 1$ as $A(x_0^{d_1}) = x_0^{d_1}$. We now analyze $\cal S(Q_i(A\vecz))$ for $i \in [n]$. If $u_i = 0$, then
    \begin{align*}
     Q_{i,1}(A\vecz) = x_0^{(3i-2)d_1}x_i^{d_2}, \ Q_{i,2}(A\vecz) =  x_0^{(3i-1)d_1} (y_i+2x_i)^{d_3} \text{ and }  Q_{i,3}(A\vecz) = x_0^{3id_1} y_i^{d_3}.
    \end{align*} 
    If $u_i = 1$, then 
    \begin{align*}
     Q_{i,1}(A\vecz) = x_0^{(3i-2)d_1}x_i^{d_2}, \ Q_{i,2}(A\vecz) =  x_0^{(3i-1)d_1} y_i^{d_3} \text{ and } Q_{i,3}(A\vecz) = x_0^{3id_1} (y_i-2x_i)^{d_3}. 
    \end{align*}  
By Observation \ref{lem_binom} (for linear forms in two variables over characteristic $0$ fields), if $u_i = 0$ then $\cal S(Q_{i,2}(A\vecz)) = d_3 + 1$ and $\cal S(Q_{i,3}(A\vecz)) = 1$ and, if $u_i = 1$ then $\cal S(Q_{i,2}(A\vecz)) = 1$ and $\cal S(Q_{i,3}(A\vecz)) = d_3 + 1$. In either case, by Observations \ref{Obs: deg separated polys} and \ref{Lemma: equiv deg sep}, 
\[\cal S(Q_i(A\vecz)) = \cal S(Q_{i,1}(A\vecz)) + \cal S(Q_{i,2}(A\vecz)) + \cal S(Q_{i,3}(A\vecz)) = d_3+3.\] 
For the $k\ith$ clause, $k \in [m]$, the action of $A$ on the corresponding polynomial $R_k$ is:
\begin{equation*}
    \label{cl_op_poly}
   R_k(A\vecz) = x_0^{(3n + k)d_1}\prod_{j \in C_k}(y_j+ ((-1)^{a_{k,j}}+(-1)^{u_j})x_j)^{d_4}. 
\end{equation*}
As the multiplicands in $R_k(A\vecz)$ do not share any variables, $\cal S(R_{k}(A\vecz))$ is the product of the sparsity of the multiplicands. Since $\psi(\vecu) = 1$, therefore in the $k\ith$ clause there exists $j \in C_k$ such that $a_{k,j} \neq u_j$. For that $j$, $(y_j+((-1)^{a_{k,j}} + (-1)^{u_j})x_j)^{d_4} = y_j^{d_4}$. As at least one literal is true in every clause under $\vecu$, $\cal S(R_{k}(A\vecz)) \leq (d_4+1)^2$ using Observation \ref{lem_binom}. Thus,
\[\cal S(f(A\vecz)) = \cal S(A(x_0^{d_1})) + \sum_{i=1}^{n} \cal S(Q_i(A\vecz)) + \sum_{k=1}^{m} \cal S(R_k(A\vecz)) \leq 1 + n(d_3+3) + m(d_4+1)^2 = s.\qedhere\] 
\end{proof}

\subsection{The reverse direction} \label{subsubsection: sparse bwd direction}
Now, we show that $(f,s)\in \et$ implies $\psi\in \SAT$ by showing that the permuted and scaled versions of the transform of \eqref{action} form all the viable sparsifying invertible linear transforms. This is where the constraints on the $d_i$'s are used. So, let $A \in \GL(|\vecz|,\F)$ be such that $\cal S(f(A\vecz)) \leq s$. Lemma \ref{Lemma: fix x0} shows that $A(x_0)$ is just a variable by leveraging $d_1 \geq s$.

\begin{lemma} \label{Lemma: fix x0}
     Without loss of generality, $A(x_0) = x_0$.
\end{lemma}
The proof of Lemma \ref{var_spars_lem} uses $d_2 \geq 2d_3$ while that of Lemma \ref{Lemma: equality condition Qi} uses $d_3 \geq m(d_4+1)^2 + 1$. 
\begin{lemma} \label{var_spars_lem}
    For any invertible $A$ and $i \in [n]$:
        \begin{equation*} \label{var_spars_eqn}
            \cal S(Q_i(A\vecz)) = \cal S(Q_{i,1}(A\vecz))+\cal S(Q_{i,2}(A\vecz))+\cal S(Q_{i,3}(A\vecz))\geq d_3 + 3,
        \end{equation*}
    where $Q_i$, $Q_{i,1}$, $Q_{i,2}$ and $Q_{i,3}$ are as defined in Section \ref{subsubsection: sparse f construction}. Equality holds if and only if under $A$
        \[x_i \mapsto X_i  \text{ and } y_i \mapsto Y_i + (-1)^{u_i}X_i \] 
    for some scaled variables $X_i, Y_i \in \vecz$ and $u_i\in \{0,1\}$. Further, if $\cal S(Q_i(A\vecz)) \neq d_3 + 3$, then $\cal S(Q_i(A\vecz)) \geq 2d_3 + 3$.
\end{lemma}
            
\begin{lemma} \label{Lemma: equality condition Qi}
Under the given $A$, $\cal S(Q_i(A\vecz)) = d_3 + 3$ holds for all $i \in [n]$.
\end{lemma}
Lemmas \ref{Lemma: fix x0}, \ref{var_spars_lem} and \ref{Lemma: equality condition Qi} together show that $A$ is a permuted scaled version of the transform of \eqref{action}. We can assume $A$ to be as described in \eqref{action} without loss of generality as permutation and non-zero scaling of variables do not affect the sparsity of a polynomial. Proposition \ref{Proposition: bwd direction} shows how a satisfying assignment can be derived from $A$ using $d_4 \geq m$.
\begin{proposition} \label{Proposition: bwd direction}
   With $A$ as described in \eqref{action}, $\vecu = (u_1, \dots , u_n)$ is a satisfying assignment for $\psi$.  
\end{proposition}
\begin{proof}
Suppose not; then there exists $k \in [m]$ such that the $k\ith$ clause, $\lor_{j \in C_k}(x_j\oplus a_{k,j})$, in $\psi$ is unsatisfied. Since this clause is unsatisfied, $u_j = a_{k,j}$ for all $j \in C_k$. Thus, $R_k(A\vecz) = x_0^{(3n+k)d_1} \prod_{j \in C_k}(y_j\pm 2x_j)^{d_4}$, where $R_k$ is as defined in Section \ref{subsubsection: sparse f construction}, and $\cal S(R_k(A\vecz)) = (d_4+1)^3\geq (m+1)(d_4+1)^2$ by Observation \ref{lem_binom}, the fact that $R_k(A\vecz)$ is a product of linear forms not sharing variables, and the condition $d_4 \geq m$. By the definition of $f$ and $s$ in Section \ref{subsubsection: sparse f construction}, Observations \ref{Obs: deg separated polys} and \ref{Lemma: equiv deg sep}, it holds that
\begin{equation*}
    \begin{split}
            \cal S(f(A\vecz)) & \geq  \cal S(A(x_0)^{d_1}) + \sum_{i = 1}^{n} \cal S(Q_i(A\vecz)) +\cal S(R_k(A\vecz)) \\
            &\geq  1 + n(3+d_3) + m(d_4+1)^2 + (d_4+1)^2 = s+ (d_4+1)^2 > s,
    \end{split}
\end{equation*}
a contradiction. Thus, $\vecu$ is a satisfying assignment for $\psi$.
\end{proof}
\subsection{The homogeneous case} \label{section-ETSparse-homogenous}
We show a modification of the construction in Section \ref{subsubsection: sparse f construction} which, along with arguments similar to those in Sections \ref{subsubsection: sparse fwd direction} and \ref{subsubsection: sparse bwd direction}, can be used to prove Theorem \ref{Theorem: ETSparse NP-hard} for homogeneous polynomials over characteristic $0$ fields. Because the polynomials are homogeneous, we cannot use degree separation like in the non-homogeneous case. Instead, we introduce a new variable $y_0$, a new degree parameter $d_5 \in \N$, and redefine $Q_i(\vecz)$ and $R_k(\vecz)$ of Section \ref{subsubsection: sparse f construction} along with modified constraints on the $d_i$'s so that:
\begin{enumerate}
    \item Each polynomial is homogeneous with the same degree and is divisible by $x_0^{d_1}$ and $y_0^{d_2}$.
    \item Each polynomial has a distinct $x_0$ degree. 
\end{enumerate}
The divisibility condition ensures that both $x_0$ and $y_0$ map to scaled variables under a sparsifying invertible linear transform (see Lemma \ref{Lemma: y0,x0 fixed homogeneous} and its proof). Due to the second condition, the polynomials are degree separated with respect to $x_0$ (see Observation \ref{Obs: deg separated polys homogeneous}). This fact is used to show that, under a sparsifying invertible linear transform, the polynomials are degree separated with respect to $x_0$ (see Lemma \ref{Lemma: y0 separated} and its proof). Formally, let $x_0$, $\vecx$ and $\vecy$ be as defined in Section \ref{subsubsection: sparse f construction} and $y_0$ be a new variable. Define $\vecz := \vecx \sqcup \vecy \sqcup \{x_0\} \sqcup \{y_0\}$. Let $d_1, d_2, d_3, d_4, d_5 \in \N$. Consider the following polynomials:
\begin{enumerate}
    \item For each variable $x_i$, $i \in [n]$, define $Q_i(\vecz) := Q_{i,1}(\vecz) + Q_{i,2}(\vecz) + Q_{i,3}(\vecz),$ where
\begin{equation*}
 \begin{split}
Q_{i,1}(\vecz) &:= x_0^{d_1 + (3i-2)(d_3+1)}y_0^{d_2 + (3n+m-3i+3)(d_3+1)-d_3} x_i^{d_3}, \\ Q_{i,2}(\vecz) &:=x_0^{d_1 + (3i-1)(d_3+1)}y_0^{d_2 + (3n+m-3i+2)(d_3+1)-d_4}(y_i+x_i)^{d_4}, \\ Q_{i,3}(\vecz) &:=x_0^{d_1 + 3i(d_3+1)}y_0^{d_2 + (3n+m-3i+1)(d_3+1)-d_4}(y_i-x_i)^{d_4}. 
\end{split}
\end{equation*}
   \item For the $k\ith$ clause, $k \in [m]$, define 
    \[R_k(\vecz) := x_0^{d_1 + (3n+k)(d_3+1)}y_0^{d_2 +(m-k+1)(d_3+1)-3d_5}\prod_{j\in C_k}(y_j+ (-1)^{a_{k,j}}x_j)^{d_5}.\]
\end{enumerate}
Define $s := 1 + n(d_4 + 3) + m(d_5+1)^2$ and impose the following conditions on the $d_i$'s: 
\begin{equation} \label{ineq_homogeneous}
\begin{split}
           d_1 &\geq d_2 + (3n+m+1)(d_3+1) + 1, \ d_2 \geq \max((3n+m+1)(d_3+1),s) + 1, \\ d_3 &\geq 2d_4, \ d_4 \geq m(d_5+1)^2+1, \ d_5 \geq m.    
\end{split}
\end{equation}
For characteristic $0$ fields, the inequalities of \eqref{ineq_homogeneous} can be converted to equalities to get
\begin{equation} \label{di-homogeneous-char0}
    \begin{split}
    d_5 &= m, \ d_4 = m(d_5+1)^2 + 1, \ d_3 = 2d_4,\\
    d_2 &= \max((3n+m+1)(d_3+1),s) + 1, \\
    d_1 &= d_2 + (3n+m+1)(d_3+1) + 1.
    \end{split}
\end{equation}
For these choices, $s = 1 + n(d_4+3) +m(d_5+1)^2 = O(nm^3)$, while $(3n+m+1)(d_3+1) = \Theta((n+m)m^3)$. Hence, 
\begin{equation*}
    \begin{split}
    d_5 &= O(m), \ d_4 = O(m^3), \ d_3 = O(m^3),\\
    d_2 &= O((n+m)m^3), \ d_1 = O((n+m)m^3).
    \end{split}
\end{equation*}
Using the conditions in \eqref{ineq_homogeneous}, it is easy to verify that the individual degree of $x_0$ and $y_0$ in every polynomial defined above is at least $d_1$ and $d_2$, respectively. Define $f$ as:
\begin{equation} \label{Homogeneous poly}
  f(\vecz) :=  x_0^{d_1}y_0^{d_2 + (3n+m+1)(d_3+1)} + \sum_{i=1}^{n}Q_i(\vecz) + \sum_{k=1}^{m}R_k(\vecz).  
\end{equation}
Clearly, $f$ is a homogeneous polynomial of degree $d_1 + d_2 + (3n+m+1)(d_3+1)$ and is divisible by $x_0^{d_1}$ and $y_0^{d_2}$. Further, the following observations hold under the constraints of \eqref{ineq_homogeneous}.

\begin{observation} \label{Obs: deg separated polys homogeneous}
For all $i \in [n]$, $k \in [m]$, the polynomials $x_0^{d_1}y_0^{d_2 + (3n+m + 1)(d_3+1)}$, $Q_{i,1}(\vecz)$, $Q_{i,2}(\vecz)$, $Q_{i,3}(\vecz)$ and $R_k(\vecz)$ are degree separated with respect to $x_0$ from one another. Also, $Q_i(\vecz)$ is degree separated with respect to $x_0$ from other $Q_j(\vecz)$'s, for $i,j \in [n]$ and $i \neq j$. Similarly, $R_k(\vecz)$ is degree separated with respect to $x_0$ from $R_l(\vecz)$ for $k,l \in [m]$ and $k \neq l$.
\end{observation}

\begin{observation} \label{Obs: sparsity poly homogeneous}
$\cal S(f(\vecz)) = 1 + n(2d_4 + 3) + m(d_5+1)^3$ and $\supp(f) = 8$.
\end{observation}
\vspace{-4mm}
\paragraph{The forward direction.} Let $\vecu \in \{0,1\}^n$ be such that $\psi(\vecu) = 1$ and $f$, as described in \eqref{Homogeneous poly}, be the polynomial corresponding to $\psi$. Proposition \ref{Proposition: fwd direction homogeneous} shows how $\vecu$ can be used to construct a sparsifying transform. The proof of Proposition \ref{Proposition: fwd direction homogeneous} is very similar to that of Proposition \ref{Proposition: fwd direction}.

\begin{proposition} \label{Proposition: fwd direction homogeneous}
$\cal S(f(A\vecz)) \leq s$ where $A \in \GL(|\vecz|,\F)$ is as follows:
\begin{equation} \label{action homogeneous case}
  A: y_0 \mapsto y_0, \ x_0 \mapsto x_0, \ x_i \mapsto x_i, \ y_i \mapsto y_i + (-1)^{u_i}x_i \ \  i \in [n].  
\end{equation}

\end{proposition}
\vspace{-4mm}
\paragraph{The reverse direction.} Let $\cal S(f(A\vecz)) \leq s$ for some $A \in \GL(|\vecz|, \F)$. Lemma \ref{Lemma: y0,x0 fixed homogeneous}, the proof of which requires Claim \ref{lem_div}, shows that $A(x_0)$ and $A(y_0)$ have only one variable each. With this established, Lemma \ref{Lemma: y0 separated} shows that the summands of $f(A\vecz)$ must be degree separated with respect to $x_0$.

\begin{lemma} \label{Lemma: y0,x0 fixed homogeneous}
Without loss of generality, $A(x_0) = x_0$ and $A(y_0) = y_0$.
\end{lemma}

\begin{lemma} \label{Lemma: y0 separated}
    For all $i \in [n]$, $k \in [m]$, the polynomials $x_0^{d_1}y_0^{d_2 +(3n+m+1)(d_3+1)}$, $Q_{i,1}(A\vecz)$, $Q_{i,2}(A\vecz)$, $Q_{i,3}(A\vecz)$ and $R_k(A\vecz)$ are degree separated from one another with respect to $x_0$. Also, $Q_i(A\vecz)$ is degree separated with respect to $x_0$ from other $Q_j(A\vecz)$'s, for $i,j \in [n]$ and $i \neq j$. Similarly, $R_k(A\vecz)$ is degree separated with respect to $x_0$ from $R_l(A\vecz)$ for $k,l \in [m]$ and $k \neq l$.
\end{lemma}
\vspace{-6.5mm}
\begin{equation*}
\therefore \ \  \cal S(f(A\vecz)) = \cal S(x_0^{d_1}y_0^{d_2 + (3n+m+1)(d_3+1)}) + \sum_{i=1}^{n} \cal S(Q_i(A\vecz)) + \sum_{k=1}^{m} \cal S(R_k(A\vecz)), \text{ by Lemma \ref{Lemma: y0 separated}.}
\end{equation*}
Lemmas \ref{Lemma: Qi(Az) analysis} and \ref{Lemma: equality condition Qi homogeneous} are modified versions of Lemmas \ref{var_spars_lem} and \ref{Lemma: equality condition Qi} respectively and have similar proofs as the original lemmas. Together, Lemmas \ref{Lemma: y0,x0 fixed homogeneous}, \ref{Lemma: Qi(Az) analysis} and \ref{Lemma: equality condition Qi homogeneous} show that $A$ is a permuted scaled version of the transform of \eqref{action homogeneous case}. Proposition \ref{Proposition: bwd direction homogeneous} then shows how to obtain a satisfying assignment from $A$ and can be proved similarly as Proposition \ref{Proposition: bwd direction}.

\begin{lemma} \label{Lemma: Qi(Az) analysis}
    For any invertible $A$ and $i \in [n]$:
        \begin{equation*} \label{var_spars_eqn_homogeneous}
            \cal S(Q_i(A\vecz)) = \cal S(Q_{i,1}(A\vecz))+\cal S(Q_{i,2}(A\vecz))+\cal S(Q_{i,3}(A\vecz))\geq d_4 + 3,
        \end{equation*}
    where $Q_i$, $Q_{i,1}$, $Q_{i,2}$ and $Q_{i,3}$ are as defined earlier. Equality holds if and only if under $A$
        \[x_i \mapsto X_i  \text{ and } y_i \mapsto Y_i + (-1)^{u_i}X_i \] 
    for some scaled variables $X_i, Y_i \in \vecz$ and $u_i\in \{0,1\}$. Further, if $\cal S(Q_i(A\vecz)) \neq d_4 + 3$, then $\cal S(Q_i(A\vecz)) \geq 2d_4 + 3$.
\end{lemma}

\begin{lemma} \label{Lemma: equality condition Qi homogeneous}
    Under the given $A$, $\cal S(Q_i(A\vecz)) = d_4 + 3$ for all $i \in [n]$.
\end{lemma}

\begin{proposition} \label{Proposition: bwd direction homogeneous}
   With $A$ as described in \eqref{action homogeneous case}, $\vecu = (u_1, \dots , u_n)$ is a satisfying assignment for $\psi$.
\end{proposition}

\begin{example} \label{Remark: homogeneous reduction feature}
    \item In the definition of $f$ in Section \ref{subsubsection: sparse f construction} and this section, an extra summand is present besides $Q_i$'s and $R_k$'s. We can drop the summand by suitably modifying $f$, the current parameters and arguments to make the reduction work. In particular, for an $f$ divisible by a suitable power of $x_0$ (and $y_0$ for the homogeneous case), Lemmas \ref{Lemma: fix x0} and \ref{Lemma: y0 separated} can be proved using Claim \ref{lem_div} (over characteristic $0$ fields) or by an argument as in Section \ref{proof-x0,y0-fixed-homogeneous-finitechar} (over finite characteristic fields). We preserve the extra summand here for two reasons: One, it leads to a simpler argument and better bounds on the $d_i$'s for the non-homogeneous case over finite characteristic fields. Two, it proves useful in showing the reduction when also considering translations (see Appendix \ref{Section: Affine ETsparse}). 

    \item A simpler construction of $f$ for the homogeneous case is possible with four degree parameters $d_1$, $d_2$, $d_3$ and $d_4$ under the constraints of \eqref{ineq}. In this construction, the extra summand, $Q_i$'s and $R_k$'s are defined very similarly as in Section \ref{subsubsection: sparse f construction}, with each polynomial multiplied by an appropriate power of $y_0$ such that $f$ is homogeneous and is divisible by $x_0^{d_1}$ and $y_0^{d_1}$. The arguments presented in this section go through with some changes for the simpler construction. The reason we present the current construction is to have a single construction with which the reduction goes through for finite characteristic fields (as shown in Section \ref{section-homogeneous-finite}) and characteristic $0$ fields.
    
    \item \label{hd3 desired feature} A feature of our reduction is that we can easily alter the output polynomial to $w^Df(\vecz)$, where $w \notin \vecz$. This can be achieved by multiplying the output polynomial $f$ of the current reduction by $w^D$, where $D$ is greater than the sparsity parameter $s$ in the reduction. If a proof of $\NP$-hardness of $\hdthree$-MCSP has this feature, then it would imply NP-hardness of $\Sigma\Pi\Sigma$-MCSP (via a homogenization trick). 
\end{example}
\subsection{Extension to finite characteristic fields} \label{section-extension-finite}
In this section, we will show how the construction of Sections \ref{subsubsection: sparse f construction} and \ref{section-ETSparse-homogenous}, with some changes for appropriate cases, can be used to show the $\NP$-hardness of $\et$ over finite characteristic fields for the non-homogeneous case (in the following section) and the homogeneous case (in Section \ref{section-homogeneous-finite}), respectively. 

\subsubsection{The non-homogeneous case} \label{section-inhomogeneous-finite}
We first show how the construction of Section \ref{subsubsection: sparse f construction} also proves the reduction over fields $\F$ where the characteristic is greater than $2$ and then give a modified construction to prove the reduction over characteristic $2$ fields. Note that the degrees are chosen in Section \ref{section-inhomogeneous-finiteparam} to satisfy \eqref{ineq} in any finite characteristic field.

So, let the characteristic be $p$, where $p > 2$. In this case, the polynomial $f$ and the parameter $s$ of Section \ref{subsubsection: sparse f construction} remain the same with the $d_i$'s chosen as specified in Section \ref{section-inhomogeneous-finiteparam}. The overall argument in both directions of the reduction is highly similar to the characteristic $0$ case, with the main differences being the choice of the $d_i$'s and that in the proofs of the observations, lemmas and propositions in Sections \ref{subsubsection: sparse f construction}, \ref{subsubsection: sparse fwd direction} and \ref{subsubsection: sparse bwd direction} wherever Observation \ref{lem_binom} is used, then it is invoked for the finite characteristic case. Thus, Observations \ref{Obs: deg separated polys} and \ref{Obs: ETSparse poly degree} hold without any change while Observation \ref{Obs: ETsparse poly sparsity } holds by using Observation \ref{lem_binom} for the finite characteristic case. 

\paragraph{The forward direction.} If $\vecu \in \{0,1\}^n$ is such that $\psi(\vecu) = 1$ and $f$, as described in \eqref{Definition: 3SAT poly}, is the polynomial corresponding to $\psi$, then Proposition \ref{Proposition: fwd direction} shows that for the transform $A$ of \eqref{action}, defined using $\vecu$, $\cal S(f(A\vecz)) \leq s$ holds. 

\paragraph{The reverse direction.} If $A \in \GL(|\vecz|, \F)$ is such that $\cal S(f(A\vecz)) \leq s$, then the analysis of Section \ref{subsubsection: sparse bwd direction} continues to hold in this case with little changes. Thus, Lemma \ref{Lemma: fix x0} shows that $A(x_0) = x_0$ without loss of generality. Then, Lemmas \ref{var_spars_lem} and \ref{Lemma: equality condition Qi} analyse the sparsity of $Q_i(A\vecz)$, where $i \in [n]$ and $Q_i$ is as defined in Section \ref{subsubsection: sparse f construction}. Together, Lemmas \ref{Lemma: fix x0}, \ref{var_spars_lem} and \ref{Lemma: equality condition Qi} show that $A$ is a permuted scaled version of the transform of \eqref{action}. Finally, Proposition \ref{Proposition: bwd direction} shows that a satisfying assignment for $\psi$ can be extracted from $A$.

\subsubsection*{\underline{Construction for characteristic $2$ fields}} \label{section-inhomogeneous-char2}
Over characteristic $2$ fields, the polynomial $y_i + x_i$ is the same as $y_i - x_i$. Due to this, the definition of $Q_i$ and that of $R_k$ in Section \ref{subsubsection: sparse f construction} need to be changed. Moreover, the sparsifying transform will also be slightly different. Formally, let $\psi$, $\vecx, x_0, \vecy$ and $\vecz$ be as denoted in Section \ref{subsubsection: sparse f construction}.  Let $d_1,d_2,d_3,d_4 \in \N$. The construction of Section \ref{subsubsection: sparse f construction} is modified as follows:
\begin{itemize}
    \item For all $i \in [n]$, define $Q_i(\vecz)$ as:
\begin{equation*} 
\begin{split}
    Q_i(\vecz) &:= Q_{i,1}(\vecz) + Q_{i,2}(\vecz) + Q_{i,3}(\vecz), \text{ where}\\
    Q_{i,1}(\vecz) &:= x_0^{(3i-2)d_1}x_i^{d_2}\,,\, Q_{i,2}(\vecz) := x_0^{(3i-1)d_1}(y_i+x_i)^{d_3} \text{ and } Q_{i,3}(\vecz) := x_0^{3id_1}y_i^{d_3}.
\end{split}
\end{equation*}
    \item For the $k\ith$ clause, $k \in [m]$, define $R_{k} := x_0^{(3n+k)d_1}\prod_{j \in C_k}(y_j+ a_{k,j}x_j)^{d_4}$.
\end{itemize}
Define $s:= 1 + n(d_3+3) + m(d_4+1)^2$ as before. Set the $d_i$'s as specified in Section \ref{section-inhomogeneous-finiteparam} with the conditions of \eqref{ineq} imposed. Define $f$ as:
\begin{equation} \label{char 2 poly}
f(\vecz) := x_0^{d_1} + \sum_{i=1}^{n}Q_i(\vecz) + \sum_{k=1}^{m}R_k(\vecz).
\end{equation}
Observations \ref{Obs: deg separated polys} and \ref{Obs: ETSparse poly degree} hold with little change. Observation \ref{Obs: poly sparsity char 2 case} analyses the sparsity and support of $f$.

\begin{observation} \label{Obs: poly sparsity char 2 case}
    $\cal S(f) \leq 1 + n(d_3 + 3) + m(d_4+1)^3$ and $4 \leq \supp(f) \leq 7$. 
\end{observation}
\begin{example} \label{Remark: char two not natural}
 Over characteristic $2$ fields, the sparsity of the polynomial output by the reduction depends on the number of variables which are complemented within a clause. Hence, for the same number of variables $n$ and the same number of clauses $m$, the output polynomial corresponding to two different $\psi$'s may have different sparsity. Thus, the reduction is not natural over characteristic $2$ fields.    
\end{example}

\paragraph{The forward direction.} Let $\vecu \in \{0,1\}^n$ be such that $\psi(\vecu) = 1$ and $f$, as described in \eqref{char 2 poly}, be the polynomial corresponding to $\psi$. Proposition \ref{Proposition: fwd direction char 2} shows how $\vecu$ can be used to construct a sparsifying transform. The proof of Proposition \ref{Proposition: fwd direction char 2} is very similar to that of Proposition \ref{Proposition: fwd direction}.

\begin{proposition} \label{Proposition: fwd direction char 2}
$\cal S(f(A\vecz)) \leq s$ where $A \in \GL(|\vecz|,\F)$ is as follows:
    \begin{equation} \label{action char 2} 
            A: x_0 \mapsto x_0,~ x_i \mapsto x_i,~ y_i \mapsto y_i + (1-u_i)x_i \ \ \forall i \in [n].
    \end{equation}
\end{proposition}

\paragraph{The reverse direction.} Let $A \in \GL(|\vecz|,\F)$ be such that $\cal S(f(A\vecz)) \leq s$. The analysis of Section \ref{subsubsection: sparse bwd direction} holds with some changes. Formally, Lemma \ref{Lemma: fix x0} holds without any change in its proof. Thus, $A(x_0) = x_0$ without loss of generality. Lemma \ref{Qi sparsity analysis} analyses $\cal S(Q_i(A\vecz))$, $i \in [n]$, and its proof is similar to that of Lemma \ref{var_spars_lem}. 
\begin{lemma} \label{Qi sparsity analysis}
    For any invertible $A$ and $i \in [n]$:
        \begin{equation*}
            \cal S(Q_i(A\vecz)) = \cal S(Q_{i,1}(A\vecz))+\cal S(Q_{i,2}(A\vecz))+\cal S(Q_{i,3}(A\vecz))\geq d_3 + 3,
        \end{equation*}
    where $Q_i$, $Q_{i,1}$, $Q_{i,2}$ and $Q_{i,3}$ are as defined in this subsection. Equality holds if and only if under $A$
        \[x_i \mapsto X_i  \text{ and } y_i \mapsto Y_i +(1 - u_i)X_i \] 
    for some scaled $X_i, Y_i \in \vecz$ and $u_i\in \{0,1\}$. Further, if $\cal S(Q_i(A\vecz)) \neq d_3 + 3$, then $\cal S(Q_i(A\vecz)) \geq 2d_3 + 3$.
\end{lemma}
Lemma \ref{Lemma: equality condition Qi} also holds with the same proof as before. Lemmas \ref{Lemma: fix x0}, \ref{Qi sparsity analysis} and \ref{Lemma: equality condition Qi} together show that $A$ is a permuted scaled version of the transform described in \eqref{action char 2}. Proposition \ref{Proposition: bwd direction char 2} then holds and can be proven similarly to Proposition \ref{Proposition: bwd direction}. 

\begin{proposition} \label{Proposition: bwd direction char 2}
    With $A$ as described in \eqref{action char 2}, $\vecu = (u_1, \dots ,u_n)$ is a satisfying assignment for $\psi$. 
\end{proposition}
\subsubsection{Setting of parameters in the non-homogeneous case} \label{section-inhomogeneous-finiteparam}
Let the characteristic be $p > 0$. If $p > d_1$, where the value of $d_1$ is as set in \eqref{di-inhomogeneous-char0} for characteristic $0$ fields, then the $d_i$'s are chosen to be the same as in \eqref{di-inhomogeneous-char0}. Otherwise, $p$ must be $O(nm^3)$. When $p = O(nm^3)$, we choose $d_1, d_2, d_3$ and $d_4$, to be of form $p^j-1$, $j \in \N$, while satisfying the inequalities of \eqref{ineq} along with $d_1 > d_2 > d_3 > d_4$. This is done so that Observation \ref{lem_binom} can be used for characteristic $p$ fields with $p = O(nm^3)$. It is possible to choose $d_i$'s in this way because, for any $k \in \N$, there is exactly one number of form $p^j-1$, $j \in \N$, in $[k,pk]$. The bounds on $d_1,d_2,d_3$ and $d_4$ are as follows:
\begin{equation*}
    \begin{split}
        d_4 &\leq pm, \ d_3 \leq pm(d_4+1)^2+p=O(p^3m^3) \implies s = O(nm^3p^3),\\    
        d_2 &= pd_3 + (p-1)= O(p^4m^3), \ d_1 = \max(r,pd_2 + p-1)\\  
    \end{split}
\end{equation*}
where $r \in [s,ps]$ is of form $p^j-1$, $j \in \N$. Thus, $r \leq ps = O(nm^3p^4)$, while $pd_2 + p-1 = O(p^5m^3)$. As $p=O(nm^3)$, 
        \begin{align*}
            d_1 = O(n^5m^{18}), \ d_2=O(n^4m^{15}), \ d_3 = O(n^3m^{12}), \ d_4=O(nm^4) \text{ and } s= O(n^4m^{12}).
        \end{align*}
\subsubsection{The homogeneous case} \label{section-homogeneous-finite}
Like in the non-homogeneous case, we first show how the construction of Section \ref{section-ETSparse-homogenous} can be used to prove the reduction over fields where the characteristic is greater than $2$ and then give a modification of this construction to prove the reduction over characteristic $2$ fields. Note that the degrees are chosen in Section \ref{section-homogeneous-finiteparam} to satisfy \eqref{ineq_homogeneous} in any finite characteristic field.

So, let the characteristic be $p$, where $p > 2$. We consider the polynomial $f$ and parameter $s$ as defined in Section \ref{section-ETSparse-homogenous}. Note that the degree of $f$ is $d_1 + d_2 + (3n+m+1)(d_3+1)$. We choose $d_i$'s in Section \ref{section-homogeneous-finiteparam} such that $d_3,d_4$ and $d_5$ are of form $p^k-1$ for some $k \in \N$, while $d_1$ and $d_2$ are of form $p^l(p^t-1)$, for some $t,l \in \N$. For this choice of the $d_i$'s, Observations \ref{Obs: deg separated polys homogeneous} and \ref{Obs: sparsity poly homogeneous} continue to hold for $f$. While the forward direction is proved similarly to the characteristic $0$ case, the reverse direction requires some change. More precisely, Lemma \ref{Lemma: y0,x0 fixed homogeneous}, which was proven earlier using Claim \ref{lem_div}, requires a different proof. This is because Claim \ref{lem_div} holds for fields with characteristic $0$ or $p$, with $p$ being ``large enough''. If $p > d_1 + d_2 + (3n+m+1)(d_3+1)$ for $d_i$'s as chosen in \eqref{di-homogeneous-char0}, then Claim \ref{lem_div} holds and so does Lemma \ref{Lemma: y0,x0 fixed homogeneous} along with the rest of the argument in the reverse direction of Section \ref{section-ETSparse-homogenous}. Thus, we consider the case when $p \leq d_1 + d_2 + (3n+m+1)(d_3+1) = O((n+m)m^3)$ and prove Lemma \ref{Lemma: y0,x0 fixed homogeneous} by a different argument. Then, the rest of the argument in the reverse direction of Section \ref{section-ETSparse-homogenous} continues to hold in the same way as before.

\paragraph{The forward direction.} Let $\vecu \in \{0,1\}^n$ be such that $\psi(\vecu) = 1$ and $f$, as described in Section \ref{section-ETSparse-homogenous}, be the polynomial corresponding to $\psi$. Proposition \ref{Proposition: fwd direction homogeneous}, with the same proof as before, shows how $\vecu$ can be used to construct a sparsifying transform.

\paragraph{The reverse direction.} 
Let $A \in \GL(|\vecz|,\F)$ such that $\cal S(f(A\vecz)) \leq s$. We prove Lemma \ref{Lemma: y0,x0 fixed homogeneous} (refer Section \ref{proof-x0,y0-fixed-homogeneous-finitechar} for its proof), which shows that $A(x_0)$ and $A(y_0)$ have only one variable each, by showing that for an appropriate choice of $d_1$ and $d_2$ (see Section \ref{section-homogeneous-finiteparam}), and the characteristic being finite, the following holds 
\[\cal S(f(A\vecz)) = \cal S(A(x_0^{d_1}))\cal S(A(y_0^{d_2}))\cal S(g(A\vecz)).\]
Here $g(\vecz)$ is a polynomial of degree $(3n+m+1)(d_3+1)$. Using Observation \ref{Lemma: linear form sparsity} and the choice of $d_i$'s, $A(x_0)$ and $A(y_0)$ are shown to be single variables. With Lemma \ref{Lemma: y0,x0 fixed homogeneous} proven, Lemma \ref{Lemma: y0 separated}, with the same proof as before, shows that all the summands in $f(A\vecz)$ are degree separated from one another with respect to $x_0$. From Lemmas \ref{Lemma: y0,x0 fixed homogeneous} and \ref{Lemma: y0 separated}, it follows that
\[\cal S(f(A\vecz)) = \cal S(x_0^{d_1}y_0^{d_2 + (3n+m+1)(d_3+1)}) + \sum_{i=1}^{n}\cal S(Q_i(A\vecz)) + \sum_{k=1}^{m}\cal S(R_k(A\vecz)).\]
Then, Lemmas \ref{Lemma: Qi(Az) analysis} and \ref{Lemma: equality condition Qi homogeneous} hold just as in the characteristic $0$ case because $d_3$, $d_4$ and $d_5$ are of form $p^k-1$. Together, Lemmas \ref{Lemma: y0,x0 fixed homogeneous}, \ref{Lemma: Qi(Az) analysis} and \ref{Lemma: equality condition Qi homogeneous} show that $A$ is a permuted scaled version of the transform of \eqref{action homogeneous case}. Proposition \ref{Proposition: bwd direction homogeneous} shows how to derive a satisfying assignment for $\psi$ from $A$.

\subsubsection*{\underline{Construction for characteristic $2$ fields}} \label{section-homogeneous-char2}
Similar to the non-homogeneous case, we modify the construction of Section \ref{section-ETSparse-homogenous} to make the reduction work over characteristic $2$ fields. Consider the following polynomials:
\begin{itemize}
    \item For all $i \in [n]$, define $Q_i(\vecz) := Q_{i,1}(\vecz) + Q_{i,2}(\vecz) + Q_{i,3}(\vecz) $ as:
\begin{equation*}
 \begin{split}
 Q_{i,1}(\vecz) &:= x_0^{d_1 + (3i-2)(d_3+1)}y_0^{d_2 + (3n+m-3i+3)(d_3+1)-d_3} x_i^{d_3}, \\ Q_{i,2}(\vecz) &:=x_0^{d_1 + (3i-1)(d_3+1)}y_0^{d_2 + (3n+m-3i+2)(d_3+1)-d_4}(y_i+x_i)^{d_4}, \\ Q_{i,3}(\vecz) &:=x_0^{d_1 + 3i(d_3+1)}y_0^{d_2 + (3n+m-3i+1)(d_3+1)-d_4}y_i^{d_4}.     
 \end{split}
\end{equation*}
    \item For the $k\ith$ clause, $k \in [m]$, define 
    \[R_k(\vecz) := x_0^{d_1 + (3n+k)(d_3+1)}y_0^{d_2 +(m-k+1)(d_3+1)-3d_5}\prod_{j\in C_k}(y_j+ a_{k,j}x_j)^{d_5}.\]
\end{itemize}

Define $s:= 1 + n(d_4+3) + m(d_5+1)^2$ as before. Set the $d_i$'s as specified in Section \ref{section-homogeneous-finiteparam} to satisfy the conditions in \eqref{ineq_homogeneous}. Define $f$ as:
\begin{equation} \label{Homogeneous char 2 poly}
  f(\vecz) :=  x_0^{d_1}y_0^{d_2 + (3n+m+1)(d_3+1)} + \sum_{i=1}^{n}Q_i(\vecz) + \sum_{k=1}^{m}R_k(\vecz).  
\end{equation}
Observation \ref{Obs: deg separated polys homogeneous} holds with little change. Observation \ref{Obs: poly sparsity char 2 case homogeneous} analyses $\cal S(f)$ and $\supp(f)$.

\begin{observation} \label{Obs: poly sparsity char 2 case homogeneous}
    $\cal S(f) \leq 1 + n(d_4 + 3) + m(d_5+1)^3$ and $5 \leq \supp(f) \leq 8$. 
\end{observation}

\begin{example} \label{Remark: char two not natural homogeneous}
 Like the non-homogeneous case, over characteristic $2$ fields, the sparsity of the homogeneous polynomial output by the reduction depends on the number of variables which are complemented within a clause. Hence, for the same number of variables $n$ and same number of clauses $m$, the output polynomial corresponding to two different $\psi$'s may have different sparsity. Hence, the reduction is not natural over characteristic $2$ fields.    
\end{example}

\paragraph{The forward direction.} Let $\vecu \in \{0,1\}^n$ be such that $\psi(\vecu) = 1$ and $f$, as described in \eqref{Homogeneous char 2 poly}, be the polynomial corresponding to $\psi$. Proposition \ref{Proposition: fwd direction char 2 homogeneous} shows how $\vecu$ can be used to construct a sparsifying transform. The proof of Proposition \ref{Proposition: fwd direction char 2 homogeneous} is similar to that of Proposition \ref{Proposition: fwd direction homogeneous}.

\begin{proposition} \label{Proposition: fwd direction char 2 homogeneous}
$\cal S(f(A\vecz)) \leq s$ where $A \in \GL(|\vecz|,\F)$ is as follows:
    \begin{equation} \label{action char 2 homogeneous} 
            A: x_0 \mapsto x_0,~ y_0 \mapsto y_0,~ x_i \mapsto x_i,~ y_i \mapsto y_i + (1-u_i)x_i \ \ \forall i \in [n].
    \end{equation}
\end{proposition}

\paragraph{The reverse direction.} Let $A \in \GL(|\vecz|,\F)$ be such that $\cal S(f(A\vecz)) \leq s$. Lemmas \ref{Lemma: y0,x0 fixed homogeneous} and \ref{Lemma: y0 separated} hold with little change in the arguments presented in the finite characteristic case. Thus, $A(x_0) = x_0$ and $A(y_0) = y_0$ without loss of generality. Lemma \ref{Lemma: Qi(Az) sparsity analysis char 2} analyses $\cal S(Q_i(A\vecz))$, $i \in [n]$, and its proof is similar to that of Lemma \ref{Lemma: Qi(Az) analysis}. 
\begin{lemma} \label{Lemma: Qi(Az) sparsity analysis char 2}
    For any invertible $A$ and $i \in [n]$:
        \begin{equation*}
            \cal S(Q_i(A\vecz)) = \cal S(Q_{i,1}(A\vecz))+\cal S(Q_{i,2}(A\vecz))+\cal S(Q_{i,3}(A\vecz))\geq d_4 + 3,
        \end{equation*}
    where $Q_i$, $Q_{i,1}$, $Q_{i,2}$ and $Q_{i,3}$ are as defined in this subsection. Equality holds if and only if under $A$
        \[x_i \mapsto X_i  \text{ and } y_i \mapsto Y_i +(1 - u_i)X_i \] 
    for some scaled $X_i, Y_i \in \vecz$ and $u_i\in \{0,1\}$. Further, if $\cal S(Q_i(A\vecz)) \neq d_4 + 3$, then $\cal S(Q_i(A\vecz)) \geq 2d_4 + 3$.
\end{lemma}
Lemma \ref{Lemma: equality condition Qi homogeneous} also holds with the same proof as before. Lemmas \ref{Lemma: y0,x0 fixed homogeneous}, \ref{Lemma: Qi(Az) sparsity analysis char 2} and \ref{Lemma: equality condition Qi homogeneous} together show that $A$ is a permuted scaled version of the transform described in \eqref{action char 2 homogeneous}. Proposition \ref{Proposition: bwd direction char 2 homogeneous} then holds and can be proven similarly to Proposition \ref{Proposition: bwd direction homogeneous}.
\begin{proposition} \label{Proposition: bwd direction char 2 homogeneous}
    With $A$ as described in \eqref{action char 2 homogeneous}, $\vecu = (u_1, \dots ,u_n)$ is a satisfying assignment for $\psi$. 
\end{proposition}
\subsubsection{Setting of parameters in the homogeneous case} \label{section-homogeneous-finiteparam}
Let the characteristic be $p > 0$. If $p > d_1 + d_2 + (3n+m+1)(d_3+1)$, where $d_1,d_2$ and $d_3$ are as chosen in \eqref{di-homogeneous-char0} for the characteristic $0$ case, then the same setting of $d_i$'s holds. Otherwise, $p = O((n+m)m^3)$. Similar to the non-homogeneous case, we choose $d_3$, $d_4$ and $d_5$ so that $d_3 > d_4 > d_5$ and they are of form $p^k-1$, $k \in \N$. We have the following bounds: 
\begin{equation} \label{d3,d4,d5 def}
           \begin{split}
           d_5 &\leq pm \implies d_5 = O(pm) = O((n+m)m^4), \\
           d_4 &\leq pm(d_5+1)^2 + p \implies d_4 = O(p^3m^3) = O((n+m)^3m^{12}), \\ 
           \therefore s &= O(p^3nm^3) = O((n+m)^3nm^{12}), \\ 
           d_3 &= pd_4+p-1 \implies d_3 = O(p^4m^3) = O((n+m)^4m^{15}). 
           \end{split}
\end{equation}
This choice of $d_3$, $d_4$ and $d_5$ ensures they are $(mn)^{O(1)}$, satisfy \eqref{ineq_homogeneous} and that Observation \ref{lem_binom} can be used for characteristic $p$ fields with $p = O((n+m)m^3)$. Now, let $k_1 := \lfloor \log_p(s) \rfloor + 1, k_2 := \lfloor \log_p((3n+m+1)(d_3+1)) \rfloor + 1$, then set 
\begin{equation} \label{d2 def}
d_2 = \sum_{i=k_2}^{k_1+k_2-1} (p-1)p^i = p^{k_1+k_2} - p^{k_2}.
\end{equation}
For this choice, $d_2 > s$ and $d_2 > (3n+m+1)(d_3+1)$. Lastly, let $k_3 := \lfloor \log_p(d_2+(3n+m+1)(d_3+1)) \rfloor + 1$, then set
\begin{equation} \label{d1 def}
d_1 = \sum_{i=k_3}^{k_3+k_1-1} (p-1)p^i = p^{k_1+k_3} - p^{k_3} > d_2+(3n+m+1)(d_3+1).
\end{equation}
For this choice of $d_1$ and $d_2$ it holds that, 
\[d_2 = O(p^2s(3n+m+1)(d_3+1)) = O(p^9nm^6(n+m)) = O((n+m)^{10}nm^{33})\] 
\[d_1 = O(p^2s(d_2+(3n+m+1)(d_3+1))) = O(p^{14}n^2m^9(n+m)) = O((n+m)^{15}n^2m^{51}).\]

%% file: Gap_ETSparse.tex
\section{$\NP$-hardness of $\gapet$} \label{Subsection: gap reduction}
In this section, we prove Theorem \ref{Theorem: gap ETSparse NP-hard}. We first prove parts $1$ and $2$ of Theorem \ref{Theorem: gap ETSparse NP-hard} over characteristic $0$ fields without considering translations in Sections \ref{sec:gap non-homog} and \ref{sec:gap homog}, respectively. Section \ref{sec:gap finite char} extends both parts to finite characteristic fields. Appendix \ref{Section: Gap ETSparse proofs} contains the proofs of the lemmas in this section. Appendix \ref{Section: Affine gap result} proves part $1$ of Theorem \ref{Theorem: gap ETSparse NP-hard} while considering translations.\footnote{For part $2$ of Theorem \ref{Theorem: gap ETSparse NP-hard}, translations are not considered; see footnote \ref{footnote: homogeneous no translate} for an explanation.}

\paragraph{Proof sketch.} For a $\CNF$ $\psi$, we carefully analyze the sparsity of the corresponding polynomial $f$ as defined in Section \ref{Subsection: Reduction ETSparse}, for the non-homogeneous and the homogeneous case. For unsatisfiable $\psi$'s, we do a slightly deeper analysis on $\cal S(f(A\vecz))$, for all $A \in \GL(|\vecz|,\F)$, to show a lower bound. For satisfiable $\psi$'s, $\cal S(f(A\vecz))$ has already been upper bounded for an appropriate $A \in \GL(|\vecz|,\F)$. The degree parameters are also chosen differently so that the gap between the lower bound and the upper bound is significant. Comparing the sparsities for satisfiable and unsatisfiable $\psi$'s proves Theorem \ref{Theorem: gap ETSparse NP-hard}. Note that throughout this section, we assume $\epsilon \in (0,1/3)$ to be an arbitrary constant.

\subsection{Analyzing the gap: the non-homogeneous case} \label{sec:gap non-homog}

For a $\CNF$ $\psi$, consider the polynomial $f$ as defined in \eqref{Definition: 3SAT poly}. Choose the $d_i$'s to satisfy the following while also being $(mn)^{O(1)}$:
\begin{equation}\label{Eqn: gap ineq}
    d_4 \geq \max(4mn,(mn)^{O(1/\epsilon)}), \ 
    d_3 = m(d_4+1)^2 + 1, \ 
    d_2 = d_3^2 + 1, \ 
    d_1 = d_2 + 1. 
\end{equation}
Note under these constraints, the conditions in \eqref{ineq} are also satisfied. Let $s := \cal S(f)$. Then, $s = 1+n(2d_3+3) + m(d_4+1)^3$ by Observation \ref{Obs: ETsparse poly sparsity }. For $\psi \in \bar{\SAT}$, Lemma \ref{Lemma: unsat gap} shows lower bounds on $\cal S(f(A\vecz))$ for all $A \in \GL(|\vecz|,\F)$. In the lemma, Item \ref{case 1 unsat gap} is essentially Lemma \ref{Lemma: fix x0}, Items \ref{case 2 unsat gap} and \ref{case 3 unsat gap} are a slightly deeper analysis of that in Lemmas \ref{var_spars_lem} and \ref{Lemma: equality condition Qi}, and Item \ref{case 4 unsat gap} is the analysis in Proposition \ref{Proposition: bwd direction}. Thus, Lemma \ref{Lemma: unsat gap} encapsulates the analysis of Section \ref{subsubsection: sparse bwd direction}.  For $\psi \in \SAT$, by Proposition \ref{Proposition: fwd direction}, there exists $A \in \GL(|\vecz|,\F)$ such that $\cal S(f(A\vecz)) \leq s_0$, where $s_0 = 1+n(d_3+3) + m(d_4+1)^2$. Proposition \ref{Proposition: Sparsity gap} shows $\gapet$ is $\NP$-hard by comparing the sparsity for satisfiable and unsatisfiable $\psi$'s, and uses Lemma \ref{Lemma: unsat gap} and the conditions in \eqref{Eqn: gap ineq}. 

\begin{lemma} \label{Lemma: unsat gap}
   Let $\psi \in \bar{\SAT}$, $f$ be as defined in \eqref{Definition: 3SAT poly} corresponding to $\psi$ and $A \in \GL(|\vecz|,\F)$.
   \begin{enumerate}
       \item \label{case 1 unsat gap} If $A(x_0)$ is a linear form in at least $2$ variables, $\cal S(f(A\vecz)) \geq d_1+1$.

       \item \label{case 2 unsat gap} If $A$ is not as in item \ref{case 1 unsat gap} and $A(x_j)$ is a linear form in at least $2$ variables for some $j \in [n]$, then $\cal S(f(A\vecz)) \geq d_2+1$.

       \item \label{case 3 unsat gap} If $A$ is not as in items \ref{case 1 unsat gap} and \ref{case 2 unsat gap} and for some $j \in [n]$, $A(y_j + x_j)$ or $A(y_j - x_j)$ is a linear form in at least $3$ variables, then $\cal S(f(A\vecz)) \geq \frac{d_3^2 + 3d_3 + 2}{2}$. 

       \item \label{case 4 unsat gap} If $A$ is not of the form described in the previous three items, then 
       $\cal S(f(A\vecz)) \geq (d_4 + 1)^3.$
   \end{enumerate}
\end{lemma}

\begin{proposition} \label{Proposition: Sparsity gap}
 Let $\textnormal{char}(\F) = 0$. If the input in $\gapet$ is an $s$-sparse polynomial, then $\gapet$ is $\NP$-hard for $\alpha = s^{1/3-\epsilon}$.
\end{proposition}

\begin{proof}
If $\psi \in \SAT$, then $\cal S(f(A\vecz)) \leq s_0$ where $A$ is as described in \eqref{action}. If $\psi \in \bar{\SAT}$, then it follows from Lemma \ref{Lemma: unsat gap} that for any $A \in \GL(|\vecz|,\F)$: 
\[\cal S(f(A\vecz)) \geq \min\bigg(d_1 + 1,d_2 + 1,\frac{d_3^2+3d_3 + 2}{2},(d_4+1)^3\bigg).\]
The constraints imposed in \eqref{Eqn: gap ineq} ensure that $(d_4+1)^3$ is the minimum. As $d_3 = m(d_4+1)^2 + 1$, therefore $s_0 = 1 + n(d_3+3) + m(d_4+1)^2 \leq 3nd_3 = 3mn(d_4+1)^2 + 3n \leq 4mn(d_4+1)^2$. Thus, the gap in the sparsities of the YES instances and the NO instances is 
\[  \frac{(d_4+1)^3}{s_0} \geq \frac{(d_4+1)^3}{4mn(d_4+1)^2} = \frac{d_4+1}{4mn}. \]
Also, as $d_4 \geq 4mn$, ${\cal S(f)} = s \leq 2m(d_4+1)^3 \implies d_4+1 \geq (\frac{s}{2m})^{1/3}$. Then, the gap is
\[  \frac{(d_4+1)^3}{s_0} \geq \frac{d_4+1}{4mn} \geq \frac{s^{1/3}}{2^{1/3}4m^{4/3}n}. \]
Finally, note that $s \geq d_4^3$. Thus, for $d_4^{3\epsilon} \geq (mn)^{O(1)}$ large enough, 
\[s^{\epsilon} \geq d_4^{3\epsilon} \geq {2^{1/3}4m^{4/3}n} \implies \frac{s^{1/3}}{2^{1/3}4m^{4/3}n} \geq s^{1/3 - \epsilon}.\]
Hence, the gap is at least $s^{1/3 - \epsilon}$. Therefore, $\SAT$ reduces to $\gapet$ for $\alpha = s^{1/3 - \epsilon}$.
\end{proof} 

\subsection{Analyzing the gap: the homogeneous case} \label{sec:gap homog}
Consider the polynomial $f$ as defined in \eqref{Homogeneous poly} for $\psi$. Choose the $d_i$'s to satisfy the following constraints while also being $(mn)^{O(1)}$.
\begin{equation}\label{Eqn: gap homogeneous ineq}
\begin{split}
    d_5 &\geq \max((4mn,(mn)^{O(1/\epsilon)}), \ d_4 = m(d_5+1)^2 + 1, \ d_3 = d_4^2 + 1, \\
    d_2 &= \max((3n+m+1)(d_3+1),s) + 1, \\ 
    d_1 &= d_2 + (3n+m+1)(d_3+1) + 1. 
\end{split}
\end{equation}
Under these constraints, the conditions in \eqref{ineq_homogeneous} are also satisfied. Let $s := \cal S(f)$. Then, $s = 1+n(2d_4+3) + m(d_5+1)^3$ by Observation \ref{Obs: sparsity poly homogeneous}. For $\psi \in \bar{\SAT}$, Lemma \ref{Lemma: unsat homogeneous gap}, proved using Claim \ref{lem_div}, shows lower bounds on $\cal S(f(A\vecz))$, for all $A \in \GL(|\vecz|,\F)$. In the lemma, Items \ref{case 1 unsat homogeneous gap} and \ref{case 2 unsat homogeneous gap} are essentially Lemma \ref{Lemma: y0,x0 fixed homogeneous}, Items \ref{case 3 unsat homogeneous gap} and \ref{case 4 unsat homogeneous gap} are a deeper analysis of that in Lemmas \ref{Lemma: Qi(Az) analysis} and \ref{Lemma: equality condition Qi homogeneous}, and Item \ref{case 5 unsat homogeneous gap} is the analysis in Proposition \ref{Proposition: bwd direction homogeneous}. For $\psi \in \SAT$, by Proposition \ref{Proposition: fwd direction homogeneous}, there exists $A \in \GL(|\vecz|, \F)$ such that $\cal S(f(A\vecz)) \leq s_0$, where $s_0 = 1 + n(d_4 +3) + m(d_5+1)^2$. Proposition \ref{Proposition: Sparsity gap homogeneous} proves $\gapet$ is $\NP$-hard using Lemma \ref{Lemma: unsat homogeneous gap} and the conditions in \eqref{Eqn: gap homogeneous ineq}. 
\begin{lemma} \label{Lemma: unsat homogeneous gap}
   Let $\psi \in \bar{\SAT}$, $f(\vecz)$ be the polynomial as defined in \eqref{Homogeneous poly} corresponding to $\psi$ and $A \in \GL(|\vecz|,\F)$.
   \begin{enumerate}
       \item \label{case 1 unsat homogeneous gap} If $A(x_0)$ is a linear form in at least $2$ variables, $\cal S(f(A\vecz)) \geq d_1+1$. 
       
       \item \label{case 2 unsat homogeneous gap} If $A$ is not as in item \ref{case 1 unsat homogeneous gap} and $A(y_0)$ is a linear form in at least $2$ variables, $\cal S(f(A\vecz)) \geq d_2+1$.

       \item \label{case 3 unsat homogeneous gap} If $A$ is not as in items \ref{case 1 unsat homogeneous gap} and \ref{case 2 unsat homogeneous gap}, and $A(x_j)$ is a linear form in at least $2$ variables for some $j \in [n]$, then $\cal S(f(A\vecz)) \geq d_3+1$.

       \item \label{case 4 unsat homogeneous gap} If $A$ is not as in items \ref{case 1 unsat homogeneous gap}, \ref{case 2 unsat homogeneous gap} and \ref{case 3 unsat homogeneous gap}, and for some $j \in [n]$, $A(y_j + x_j)$ or $A(y_j - x_j)$ is a linear form in at least $3$ variables, then $\cal S(f(A\vecz)) \geq \frac{d_4^2 + 3d_4 + 2}{2}$. 

       \item \label{case 5 unsat homogeneous gap} If $A$ is not of the form described in the previous four items, then 
       $\cal S(f(A\vecz)) \geq (d_5 + 1)^3.$
   \end{enumerate}
\end{lemma}
\begin{proposition} \label{Proposition: Sparsity gap homogeneous}
Let $\textnormal{char}(\F) = 0$. If the input in $\gapet$ is an $s$-sparse homogeneous polynomial, then $\gapet$ is $\NP$-hard for $\alpha = s^{1/3-\epsilon}$. 
\end{proposition}

\begin{proof}
If $\psi \in \SAT$, then $\cal S(f(A\vecz)) \leq s_0$ where $A$ is as described in \eqref{action homogeneous case}. If $\psi \in \bar{\SAT}$, then it follows from Lemma \ref{Lemma: unsat homogeneous gap} that for any $A \in \GL(|\vecz|,\F)$: 
\[\cal S(f(A\vecz)) \geq \min\bigg(d_1 + 1,d_2 + 1, d_3 + 1, \frac{d_4^2+3d_4 + 2}{2},(d_5+1)^3\bigg).\]
The constraints imposed in \eqref{Eqn: gap homogeneous ineq} ensure that $(d_5+1)^3$ is the minimum. As $d_4 = m(d_5+1)^2 + 1$, therefore $s_0 = 1 + n(d_4+3) + m(d_5+1)^2 \leq 3nd_5 = 3mn(d_5+1)^2 + 3n \leq 4mn(d_5+1)^2$. Thus, the gap in the sparsities of the YES instances and the NO instances is 
\[  \frac{(d_5+1)^3}{s_0} \geq \frac{(d_5+1)^3}{4mn(d_5+1)^2} = \frac{d_5+1}{4mn}. \]
Also, as $d_5 \geq 4mn$, ${\cal S(f)} = s \leq 2m(d_5+1)^3 \implies d_5+1 \geq (\frac{s}{2m})^{1/3}$. Then, the gap is
\[  \frac{(d_5+1)^3}{s_0} \geq \frac{d_5+1}{4mn} \geq \frac{s^{1/3}}{2^{1/3}4m^{4/3}n}. \]
Finally, note that $s \geq d_5^3$. Thus, for $d_5^{3\epsilon} \geq (mn)^{O(1)}$ large enough, 
\[s^{\epsilon} \geq d_5^{3\epsilon} \geq {2^{1/3}4m^{4/3}n} \implies \frac{s^{1/3}}{2^{1/3}4m^{4/3}n} \geq s^{1/3 - \epsilon}.\]
Hence, the gap is at least $s^{1/3 - \epsilon}$. Therefore, $\SAT$ reduces to $\gapet$ for $\alpha = s^{1/3 - \epsilon}$.
\end{proof}

\begin{example}
    As mentioned in the second remark near the end of Section \ref{section-ETSparse-homogenous}, a simpler construction of a homogeneous $f$ with four degree parameters is possible. This simpler construction can be used to show the $\NP$-hardness of $\gapet$. The reason we use the construction currently defined in Section \ref{section-ETSparse-homogenous} is that it allows us to prove part $2$ of Theorem \ref{Theorem: gap ETSparse NP-hard} for homogeneous polynomials over finite characteristic fields and characteristic $0$ fields without using separate constructions.
\end{example}
\subsection{Extension to finite characteristic fields} \label{sec:gap finite char}
In this section, we will show how the construction of Sections \ref{subsubsection: sparse f construction} and \ref{section-ETSparse-homogenous}, with some changes for appropriate cases, can be used to show the $\NP$-hardness of $\gapet$ over finite characteristic fields for the non-homogeneous case (in the following section) and homogeneous case (in Section \ref{Subsection: Gap result homogeneous char p}), respectively.

\subsubsection{The non-homogeneous case} \label{Subsection: Gap result char p}
Let the characteristic be $p$, where $p > 2$. If $p > d_1$, where $d_1$ is as chosen in Section \ref{sec:gap non-homog}, then the argument of that section holds. Hence, it is assumed that $p \leq d_1 =  (mn)^{O(1)}$. Consider the polynomial $f$ defined in \eqref{Definition: 3SAT poly}. Let $\epsilon \in (0,1/3)$ be an arbitrary constant. Choose the $d_i$'s to be of form $p^k-1$ for some $k \in \N$ to satisfy the following inequalities, along with those of \eqref{ineq}:
\begin{equation}\label{Eqn: gap ineq prime}
    d_4 \geq \max(3pmn,(mn)^{O(1/\epsilon)}), \  d_3 > m(d_4+1)^2, \ d_2 > (d_3+1)^2, \ d_1 > d_2. 
\end{equation}
Note that we can get $d_3 = O(pm(d_4+1)^2)$, $d_2 = O(p(d_3+1)^2)$ and $d_1 = O(pd_2)$. Let $s :=\cal S(f)$. Then, $s = 1 + n(2d_3+3) + m(d_4+1)^3$ by Observation \ref{Obs: ETsparse poly sparsity }. For $\psi \in \bar{\SAT}$, Lemma \ref{Lemma: unsat gap char p} shows lower bounds on  $\cal S(f(A\vecz))$ for all $A \in \GL(|\vecz|,\F)$. Like Lemma \ref{Lemma: unsat gap}, Lemma \ref{Lemma: unsat gap char p} is a slightly deeper analysis of that in the reverse direction of Section \ref{section-inhomogeneous-finite}. For $\psi \in \SAT$, by Proposition \ref{Proposition: fwd direction} there exists $A \in \GL(|\vecz|,\F)$ such that $\cal S(f(A\vecz)) \leq s_0$, where $s_0 = 1+n(d_3+3) + m(d_4+1)^2$. Proposition \ref{Proposition: Sparsity gap prime} shows the $\NP$-hardness of $\gapet$ using Lemma \ref{Lemma: unsat gap char p} and the inequalities in \eqref{Eqn: gap ineq prime}.
\begin{lemma} \label{Lemma: unsat gap char p}
   Let $\psi \in \bar{\SAT}$, $f$, as defined in \eqref{Definition: 3SAT poly}, be the polynomial corresponding to $\psi$ and $A \in \GL(|\vecz|,\F)$.
   \begin{enumerate}
       \item \label{case 1 unsat gap char p} If $A(x_0)$ is a linear form in at least $2$ variables, $\cal S(f(A\vecz)) \geq d_1 + 1$.

       \item \label{case 2 unsat gap char p} If $A$ is not as in item \ref{case 1 unsat gap char p} and for some $j \in [n]$, $A(x_j)$ is a linear form in at least $2$ variables, then $\cal S(f(A\vecz)) \geq d_2 + 1$.

       \item \label{case 3 unsat gap char p} If $A$ is not as in item \ref{case 1 unsat gap char p} and \ref{case 2 unsat gap char p} and for some $j \in [n]$, $A(y_j + x_j)$ or $A(y_j - x_j)$ is a linear form in at least $3$ variables, $\cal S(f(A\vecz)) \geq (d_3+1)^{1.63}$. 

       \item \label{case 4 unsat gap char p} If $A$ is not of the form described in the previous three cases, then 
       $\cal S(f(A\vecz)) \geq (d_4 + 1)^3.$
   \end{enumerate}
\end{lemma}
\begin{proposition} \label{Proposition: Sparsity gap prime}
Let $\textnormal{char}(\F) = p > 2$. If the input in $\gapet$ is an $s$-sparse polynomial, then $\gapet$ is $\NP$-hard for $\alpha = s^{1/3-\epsilon}$.
\end{proposition}
\begin{proof}
The proof is similar to that of Proposition \ref{Proposition: Sparsity gap}. If $\psi$ is satisfiable, then $\cal S(f(A\vecz)) \leq s_{0}$ where $A$ is as described in \eqref{action} and $s_{0} = 1 + n(d_3+3) + m(d_4+1)^2$. For unsatisfiable $\psi$, it follows from Lemma \ref{Lemma: unsat gap char p} that for any $A \in \GL(|\vecz|,\F)$: 
\[\cal S(f(A\vecz)) \geq \min(d_1 + 1,d_2 + 1,(d_3+1)^{1.63},(d_4+1)^3).\]
As $d_3 > m(d_4+1)^2$, therefore $s_0 = 1 + n(d_3+3) + m(d_4+1)^2 \leq 3nd_3 \leq 3pmn(d_4+1)^2$.  The conditions imposed in \eqref{Eqn: gap ineq prime} ensure that $d_1 + 1 > d_2 + 1 > (d_3+1)^{1+\log_p((p+1)/2)} > (d_4+1)^3 > s_{0}$. Thus, the gap in the sparsities of the YES instances and NO instances is 
\[  \frac{(d_4+1)^3}{s_{0}} \geq \frac{(d_4+1)^3}{3pmn(d_4+1)^2} = \frac{d_4+1}{3pmn}. \]
Also, note as $d_4 \geq 3pmn$, therefore $s \leq 2m(d_4+1)^3 \implies d_4+1 \geq (\frac{s}{2m})^{1/3}$. Then, the gap is
\[  \frac{(d_4+1)^3}{s_{0}} \geq \frac{d_4+1}{3pmn} \geq \frac{s^{1/3}}{p2^{1/3}3m^{4/3}n}. \]
Finally, note that $s \geq d_4^3$. Thus, for $d_4^{3\epsilon} \geq (mn)^{O(1)}$ large enough, 
\[s^{\epsilon} \geq d_4^{3\epsilon} \geq {p2^{1/3}3m^{4/3}n} \implies \frac{s^{1/3}}{p2^{1/3}3m^{4/3}n} \geq s^{1/3 - \epsilon}.\]
Hence, the gap is at least $s^{1/3 - \epsilon}$. Therefore, $\SAT$ reduces to $\gapet$ for $\alpha = s^{1/3 - \epsilon}$.
\end{proof}

\subsubsection*{\underline{Analysis over characteristic $2$ fields}} \label{Subsection: Gap result char 2}
 In this case, we consider the construction of Section \ref{section-inhomogeneous-finite}. For a $\CNF$ $\psi$, let $f$ be the corresponding polynomial as defined in \eqref{char 2 poly}. Let $s := \cal S(f)$. Over characteristic $2$ fields, the value of $s$ depends on the number of variables complemented in a clause. To show the hardness of $\gapet$, we require $s \geq d_4^3$ (see the proof of Proposition \ref{Proposition: Sparsity gap char 2}). This can be achieved if there exists a clause with all variables complemented. Hence, we assume, without loss of generality, that there is such a clause in $\psi$.\footnote{To have some clause, say the first one, contain only complemented variables, every uncomplemented variable $x$ in the clause can be replaced by $\neg x$ followed by complementing each occurrence of $x$ in the remaining clauses. \label{footnote: complemented clause}} The $d_i$'s are chosen in the same way as in Section \ref{Subsection: Gap result char p} to satisfy \eqref{Eqn: gap ineq prime} with $p$ set to $2$. In particular, they satisfy the following inequalities.
\begin{equation}\label{Eqn: gap ineq 2}
    d_4 \geq \max(6mn,(mn)^{O(1/\epsilon)}), \  d_3 > m(d_4+1)^2, \ d_2 > (d_3+1)^2, \ d_1 > d_2. 
\end{equation}
By Observation \ref{Obs: poly sparsity char 2 case} and the assumption on $\psi$, it holds that 
\[1 + n(d_3 + 3) + (d_4 + 1)^3 \leq s \leq 1 + n(d_3 + 3) + m(d_4 + 1)^3.\]
For $\psi \in \bar{\SAT}$, Lemma \ref{Lemma: unsat gap char 2} shows lower bounds on $\cal S(f(A\vecz))$ where $A \in \GL(|\vecz|,\F)$. For $\psi \in \SAT$, by Proposition \ref{Proposition: fwd direction char 2} there exists $A \in \GL(|\vecz|,\F)$ such that $\cal S(f(A\vecz)) \leq s_0$, where $s_0 = 1 + n(d_3 + 3) + m(d_4+1)^2$. Proposition \ref{Proposition: Sparsity gap char 2} shows the $\NP$-hardness of $\gapet$ using Lemma \ref{Lemma: unsat gap char 2} and the inequalities in \eqref{Eqn: gap ineq 2}.

\begin{lemma} \label{Lemma: unsat gap char 2}
   Let $\psi \in \bar{\SAT}$, $f$, as defined in \eqref{char 2 poly}, be the polynomial corresponding to $\psi$ and $A \in \GL(|\vecz|,\F)$.
   \begin{enumerate}
       \item \label{case 1 unsat gap char 2} If $A(x_0)$ is a linear form in at least $2$ variables, $\cal S(f(A\vecz)) \geq d_1 + 1$.

       \item \label{case 2 unsat gap char 2} If $A$ is not as in item \ref{case 1 unsat gap char p} and for some $j \in [n]$, $A(x_j)$ is a linear form in at least $2$ variables, then $\cal S(f(A\vecz)) \geq d_2 + 1$.

       \item If $A$ is not as in item \ref{case 1 unsat gap char 2} and \ref{case 2 unsat gap char 2} and for some $j \in [n]$, $A(y_j + x_j)$ or $A(y_j)$ is a linear form in at least $3$ variables, $\cal S(f(A\vecz)) \geq (d_3+1)^{1.58}$. 

       \item If $A$ is not of the form described in the previous three cases, then 
       $\cal S(f(A\vecz)) \geq (d_4 + 1)^3.$
   \end{enumerate}
\end{lemma}
\begin{proposition} \label{Proposition: Sparsity gap char 2}
Let $\textnormal{char}(\F) = 2$. If the input in $\gapet$ is an $s$-sparse polynomial, then $\gapet$ is $\NP$-hard for $\alpha = s^{1/3-\epsilon}$.
\end{proposition}
\begin{proof}  
The proof is similar to that of Proposition \ref{Proposition: Sparsity gap prime}. If $\psi$ is satisfiable, then $\cal S(f(A\vecz)) \leq s_{0}$ where $A$ is as described in \eqref{action} and $s_{0} = 1 + n(d_3+3) + m(d_4+1)^2$. For unsatisfiable $\psi$, it follows from Lemma \ref{Lemma: unsat gap char p} that for any $A \in \GL(|\vecz|,\F)$: 
\[\cal S(f(A\vecz)) \geq \min(d_1 + 1,d_2 + 1,(d_3+1)^{1.58},(d_4+1)^3).\]
As $d_3 > m(d_4+1)^2$, therefore $s_0 = 1 + n(d_3+3) + m(d_4+1)^2 \leq 3nd_3 \leq 6mn(d_4+1)^2$.  The conditions imposed in \eqref{Eqn: gap ineq 2} ensure that $d_1 + 1 > d_2 + 1 > (d_3+1)^{1.58} > (d_4+1)^3 > s_{0}$. Consequently, the gap in the sparsities of the YES instances and NO instances is 
\[  \frac{(d_4+1)^3}{s_{0}} \geq \frac{(d_4+1)^3}{6mn(d_4+1)^2} = \frac{d_4+1}{6mn}. \]
Also, note as $d_4 \geq 6mn$, therefore $s \leq 2m(d_4+1)^3 \implies d_4+1 \geq (\frac{s}{2m})^{1/3}$. Then, the gap is
\[  \frac{(d_4+1)^3}{s_{0}} \geq \frac{d_4+1}{6mn} \geq \frac{s^{1/3}}{2^{4/3}3m^{4/3}n}. \]
Finally, note that $s \geq d_4^3$. Thus, for $d_4^{3\epsilon} \geq (mn)^{O(1)}$ large enough, 
\[s^{\epsilon} \geq d_4^{3\epsilon} \geq {2^{4/3}3m^{4/3}n} \implies \frac{s^{1/3}}{2^{4/3}3m^{4/3}n} \geq s^{1/3 - \epsilon}.\]
Hence, the gap is at least $s^{1/3 - \epsilon}$. Therefore, $\SAT$ reduces to $\gapet$ for $\alpha = s^{1/3 - \epsilon}$.
\end{proof}

\subsubsection{The homogeneous case} \label{Subsection: Gap result homogeneous char p}
Let the characteristic be $p$, where $p > 2$. If $p > d_1$, where $d_1$ is as chosen in Section \ref{sec:gap homog}, then the argument of that section holds. Hence, assume $p \leq d_1 =  O((mn)^{O(1)})$. Consider the polynomial $f$ defined in \eqref{Homogeneous poly}. Choose $d_3$, $d_4$ and $d_5$ to be of form $p^k-1$ for some $k \in \N$ while also satisfying: 
\begin{equation}\label{Eqn: gap ineq char p}
\begin{split}
    d_5 &\geq \max(3pmn,(mn)^{O(1/\epsilon)}), \ d_4 > m(d_5+1)^2, \ d_3 > (d_4+1)^2. 
\end{split}
\end{equation}
Thus, $d_4 = O(pm(d_5+1)^2)$ and $d_3 = O(p(d_4+1)^2)$. Now, let $k_1 := \lfloor \log_p(d_3+2) \rfloor + 1, k_2 := \lfloor \log_p((3n+m+1)(d_3+1)) \rfloor + 1$, then set 
\begin{equation}
d_2 := \sum_{i=k_2}^{k_1+k_2-1} (p-1)p^i = p^{k_1+k_2} - p^{k_2}.
\end{equation}
Lastly, let $k_3 := \lfloor \log_p(d_2+(3n+m+1)(d_3+1)) \rfloor + 1$, then set
\begin{equation}
d_1 := \sum_{i=k_3}^{k_3+k_1-1} (p-1)p^i = p^{k_1+k_3} - p^{k_3}.
\end{equation}
For this choice of the $d_i$'s, the conditions in \eqref{ineq_homogeneous} are satisfied. Let $s := \cal S(f)$. By Observation \ref{Obs: sparsity poly homogeneous}, $s = 1 + n(2d_4+3) + m(d_5+1)^3$. For $\psi \in \bar{\SAT}$, Lemma \ref{Lemma: unsat gap homogeneous char p} shows lower bounds on $\cal S(f(A\vecz))$ for all $A \in \GL(|\vecz|,\F)$. Like Lemma \ref{Lemma: unsat homogeneous gap}, Lemma \ref{Lemma: unsat gap homogeneous char p} is a slightly deeper analysis of that in the reverse direction of Section \ref{section-homogeneous-finite} For $\psi \in \SAT$, by Proposition \ref{Proposition: fwd direction homogeneous} there exists $A \in \GL(|\vecz|,\F)$ such that $\cal S(f(A\vecz)) \leq s_0$, where $s_0 = 1+n(d_4+3) + m(d_5+1)^2$. Proposition \ref{Proposition: Sparsity gap homogeneous prime} shows the $\NP$-hardness of $\gapet$ using Lemma \ref{Lemma: unsat gap homogeneous char p} and the setting of the $d_i$'s in this section. 

\begin{lemma} \label{Lemma: unsat gap homogeneous char p}
   Let $\psi \in \bar{\SAT}$, $f$, as defined in \eqref{Homogeneous poly}, be the polynomial corresponding to $\psi$ and $A \in \GL(|\vecz|,\F)$.
   \begin{enumerate}
       \item \label{case 1 unsat gap homogeneous char p} If $A(x_0)$ or $A(y_0)$ is a linear form in at least $2$ variables, $\cal S(f(A\vecz)) \geq d_3+2$.

       \item \label{case 2 unsat gap homogeneous char p} If $A$ is not as in item \ref{case 1 unsat gap homogeneous char p} and $A(x_j)$ is a linear form in at least $2$ variables for some $j \in [n]$, then $\cal S(f(A\vecz)) \geq d_3+1$.

       \item If $A$ is not as in items \ref{case 1 unsat gap homogeneous char p} and \ref{case 2 unsat gap homogeneous char p} and for some $j \in [n]$, $A(y_j + x_j)$ or $A(y_j - x_j)$ is a linear form in at least $3$ variables, $\cal S(f(A\vecz)) \geq (d_4+1)^{1.63}$. 

       \item If $A$ is not of the form described in the previous three cases, then 
       $\cal S(f(A\vecz)) \geq (d_5 + 1)^3.$
   \end{enumerate}
\end{lemma}
\begin{proposition} \label{Proposition: Sparsity gap homogeneous prime}
    Let $\textnormal{char}(\F) = p > 2$. If the input in $\gapet$ is an $s$-sparse homogeneous polynomial, then $\gapet$ is $\NP$-hard for $\alpha = s^{1/3-\epsilon}$.
\end{proposition}
\begin{proof}
The proof is similar to that of Proposition \ref{Proposition: Sparsity gap homogeneous}. If $\psi \in \SAT$, then $\cal S(f(A\vecz)) \leq s_0$ where $A$ is as described in \eqref{action homogeneous case}. If $\psi \in \bar{\SAT}$, then it follows from Lemma \ref{Lemma: unsat gap char p} that for any $A \in \GL(|\vecz|,\F)$: 
\[\cal S(f(A\vecz)) \geq \min\bigg(d_3 + 2, d_3 + 1, (d_4+1)^{1.63}, (d_5+1)^3\bigg).\]
The constraints imposed in \eqref{Eqn: gap ineq} ensure that $(d_5+1)^3$ is the minimum. As $d_4 > m(d_5+1)^2$, therefore $s_0 = 1 + n(d_4+3) + m(d_5+1)^2 < 3nd_4 \leq 3pmn(d_5+1)^2$. Thus, the gap in the sparsities of the YES instances and the NO instances is 
\[  \frac{(d_5+1)^3}{s_0} \geq \frac{(d_5+1)^3}{3pmn(d_5+1)^2} = \frac{d_5+1}{3pmn}. \]
Also, as $d_5 \geq 3pmn$, ${\cal S(f)} = s \leq 2m(d_5+1)^3 \implies d_5+1 \geq (\frac{s}{2m})^{1/3}$. Then, the gap is
\[  \frac{(d_5+1)^3}{s_0} \geq \frac{d_5+1}{3pmn} \geq \frac{s^{1/3}}{p2^{1/3}3m^{4/3}n}. \]
Finally, note that $s \geq d_5^3$. Thus, for $d_5^{3\epsilon} \geq (mn)^{O(1)}$ large enough, 
\[s^{\epsilon} \geq d_5^{3\epsilon} \geq {p2^{1/3}3m^{4/3}n} \implies \frac{s^{1/3}}{p2^{1/3}3m^{4/3}n} \geq s^{1/3 - \epsilon}.\]
Hence, the gap is at least $s^{1/3 - \epsilon}$. Therefore, $\SAT$ reduces to $\gapet$ for $\alpha = s^{1/3 - \epsilon}$.
\end{proof}

\subsubsection*{\underline{Analysis over characteristic $2$ fields}} \label{Subsection: Gap result homogeneous char 2}
For characteristic $2$ fields, consider the polynomial $f$ as defined in \eqref{Homogeneous char 2 poly}. Let $s := \cal S(f)$. Like in the characteristic $2$ construction for the non-homogeneous case, the value of $s$ depends on the number of variables complemented within a clause. To show the $\NP$-hardness of $\gapet$, $s \geq d_5^3$ is required (see the proof of Proposition \ref{Proposition: Sparsity gap homogeneous char 2}), which can be achieved if there is at least one clause where all variables are complemented. Therefore, we assume that such a clause exists (see footnote \ref{footnote: complemented clause}). Set the $d_i$'s as specified in the beginning of Section \ref{Subsection: Gap result homogeneous char p} with $p = 2$. For this setting of the $d_i$'s, the constraints in \eqref{ineq_homogeneous} are satisfied. By Observation \ref{Obs: poly sparsity char 2 case homogeneous} and the assumption on $\psi$, it holds that
\[1 + n(d_4 + 3) + (d_5 + 1)^3 \leq s \leq 1 + n(d_4 + 3) + m(d_5 + 1)^3.\]For $\psi \in \bar{\SAT}$, Lemma \ref{Lemma: unsat gap homogeneous char 2} shows lower bounds on $\cal S(f(A\vecz))$ for all $A \in \GL(|\vecz|,\F)$. For $\psi \in \SAT$, by Proposition \ref{Proposition: fwd direction char 2 homogeneous} there exists $A \in \GL(|\vecz|,\F)$ such that $\cal S(f(A\vecz)) \leq s_0$, where $s_0 = 1+n(d_4+3) + m(d_5+1)^2$. Proposition \ref{Proposition: Sparsity gap homogeneous char 2} shows the $\NP$-hardness of $\gapet$ using Lemma \ref{Lemma: unsat gap homogeneous char 2} and the setting of the $d_i$'s in this section. 

\begin{lemma} \label{Lemma: unsat gap homogeneous char 2}
   Let $\psi \in \bar{\SAT}$, $f$, as defined in \eqref{Homogeneous char 2 poly}, be the polynomial corresponding to $\psi$ and $A \in \GL(|\vecz|,\F)$.
   \begin{enumerate}
       \item \label{case 1 unsat gap homogeneous char 2} If $A(x_0)$ or $A(y_0)$ is a linear form in at least $2$ variables, $\cal S(f(A\vecz)) \geq d_3+2$.

       \item \label{case 2 unsat gap homogeneous char 2} If $A$ is not as in item \ref{case 1 unsat gap homogeneous char 2} and $A(x_j)$ is a linear form in at least $2$ variables for some $j \in [n]$, then $\cal S(f(A\vecz)) \geq d_3+1$.

       \item If $A$ is not as in items \ref{case 1 unsat gap homogeneous char 2} and \ref{case 2 unsat gap homogeneous char 2} and for some $j \in [n]$, $A(y_j + x_j)$ or $A(y_j - x_j)$ is a linear form in at least $3$ variables, $\cal S(f(A\vecz)) \geq (d_4+1)^{1.58}$. 

       \item If $A$ is not of the form described in the previous three cases, then 
       $\cal S(f(A\vecz)) \geq (d_5 + 1)^3.$
   \end{enumerate}
\end{lemma}
\begin{proposition} \label{Proposition: Sparsity gap homogeneous char 2}
    Let $\textnormal{char}(\F) = 2$. If the input in $\gapet$ is an $s$-sparse homogeneous polynomial, then $\gapet$ is $\NP$-hard for $\alpha = s^{1/3-\epsilon}$.
\end{proposition}
\begin{proof}
The proof is similar to that of Proposition \ref{Proposition: Sparsity gap homogeneous prime}. If $\psi \in \SAT$, then $\cal S(f(A\vecz)) \leq s_0$ where $A$ is as described in \eqref{action char 2 homogeneous}. If $\psi \in \bar{\SAT}$, then it follows from Lemma \ref{Lemma: unsat gap char p} that for any $A \in \GL(|\vecz|,\F)$: 
\[\cal S(f(A\vecz)) \geq \min\bigg(d_3 + 2, d_3 + 1, (d_4+1)^{1.58}, (d_5+1)^3\bigg).\]
The constraints imposed in \eqref{Eqn: gap ineq} ensure that $(d_5+1)^3$ is the minimum. As $d_4 > m(d_5+1)^2$, therefore $s_0 = 1 + n(d_4+3) + m(d_5+1)^2 < 3nd_4 \leq 6mn(d_5+1)^2$. Thus, the gap in the sparsities of the YES instances and the NO instances is 
\[  \frac{(d_5+1)^3}{s_0} \geq \frac{(d_5+1)^3}{6mn(d_5+1)^2} = \frac{d_5+1}{6mn}. \]
Also, as $d_5 \geq 6mn$, $s \leq 2m(d_5+1)^3 \implies d_5+1 \geq (\frac{s}{2m})^{1/3}$. Then, the gap is
\[  \frac{(d_5+1)^3}{s_0} \geq \frac{d_5+1}{6mn} \geq \frac{s^{1/3}}{2^{4/3}3m^{4/3}n}. \]
Finally, note that $s \geq d_5^3$. Thus, for $d_5^{3\epsilon} \geq (mn)^{O(1)}$ large enough, 
\[s^{\epsilon} \geq d_5^{3\epsilon} \geq {2^{4/3}3m^{4/3}n} \implies \frac{s^{1/3}}{2^{4/3}3m^{4/3}n} \geq s^{1/3 - \epsilon}.\]
Hence, the gap is at least $s^{1/3 - \epsilon}$. Therefore, $\SAT$ reduces to $\gapet$ for $\alpha = s^{1/3 - \epsilon}$.
\end{proof}

%% file: ETCSP.tex
\section{$\NP$-hardness of $\etcsp$} \label{Section: ETsupport NP-hard}
In this section, we prove Theorem \ref{Theorem: ETsupport NP-hard}. All lemmas and observations are proved in Appendix \ref{Section: ETsupport proofs}. As stated at the end of Section \ref{Section: Results}, we strengthen the construction in \cite{BDSS24} to prove Theorem \ref{Theorem: ETsupport NP-hard} for $\sigma \geq 5$. At the end of Section \ref{secConstructPoly}, we describe the differences between the current construction and that in \cite{BDSS24}.
\vspace{-4mm}
\paragraph{Proof sketch.}We map $\psi$, a $\CNF$, to a polynomial $f$, which is the sum of polynomials which are degree separated or degree separated with respect to some variable. At least one polynomial in $f$ is of support $\sigma + 1$ and the rest of support $\sigma$ (where $\sigma$ is a constant). As the summands are degree separated or degree separated with respect to some variable, $\supp(f) = \sigma + 1$ and for any invertible linear transform $A$, $\supp(A(f))$ is equal to the maximum support size among the transformed summands. Claim \ref{FullSupport} is used to show that $\psi \in \SAT$ iff there exists an invertible linear transform $A$, such that $\supp(A(f)) \leq \sigma$. Thus, the reduction also holds for $\prometcsp$. For characteristic $p$ fields, we assume $p > \sigma + 1$ so that Claim \ref{FullSupport} holds.

\subsection{Construction of $f$} \label{secConstructPoly}
Let $\sigma \geq 5$ be an even integer constant and $\psi$ be as denoted in Section \ref{subsubsection: sparse f construction}. For odd $\sigma$, we describe the changes in the construction/argument at the appropriate points. As $\sigma$ is a constant, let $n \geq \sigma + 4$. The proofs of Lemmas \ref{P1Sep} and \ref{P2Sep} use $n \geq \sigma + 4$. To ensure $\supp(f) = \sigma + 1$, we assume that all the variables in the first clause are complemented (see footnote \ref{footnote: complemented clause}) and that no clause is repeated. Let $\vecx := \{x_1,\dots,x_n\}$, $\vecy := \{y_{1},\dots,y_{n}\}$ and $\vecz := \{z_{1},\dots, z_{\sigma-5}\}$ and $\vecw := \vecx \sqcup \vecy \sqcup \vecz$. Note that when $\sigma = 5$, $\vecz$ is an empty set and hence we do not consider the $\vecz$ variables in such a case. By $(w_1 \cdots w_l)^{\star}$, where $w_i \in \vecw$ and $l,\star \in \N$, we denote a power of $w_1\cdots w_l$. Consider the polynomials:

\begin{itemize}
    \item First, introduce $\binom{n + \sigma - 5}{\sigma}$ many monomials defined by the set 
    \[P := \{(w_{1}\cdots w_{{\sigma}})^\star \ | \ w_{1}, \dots, w_{{\sigma}} \in \vecz \sqcup \vecx \ \text{ and are pairwise distinct} \}.\]
 
    \item Then, introduce $\binom{n}{\frac{\sigma}{2}}$ many monomials defined by the set 
    \[Q := \{((x_{i_1}y_{i_1})\cdots (x_{i_{\frac{\sigma}{2}}} y_{i_{\frac{\sigma}{2}}}))^\star \ | \ i_1,\dots,i_{\frac{\sigma}{2}} \in [n] \ \text{ and are pairwise distinct} \}.\]
    \textbf{Note:} For odd $\sigma$, the monomials are of form $((x_{i_1}y_{i_1})\cdots (x_{i_{\frac{\sigma-1}{2}}} y_{i_{\frac{\sigma-1}{2}}})x_{i_\frac{\sigma+1}{2}})^\star$. Thus $|Q| = \binom{n}{\frac{\sigma+1}{2}}\frac{\sigma+1}{2}$. 

    \item Let $R := \{R_k(\vecw) \ | \ k \in [m]\}$, where $R_k(\vecw)$ is defined corresponding to the $k\ith$ clause as:
    \[R_{k}(\vecw) := (\prod_{j \in C_k} (y_{j} - a_{k,j}x_{j} )^{2+a_{k,j}}) (z_{1} \cdots z_{\sigma-5} ). \]
\end{itemize}
Define $f(\vecw) := \sum_{g \in P} g(\vecw) + \sum_{h \in Q} h(\vecw) + \sum_{k=1}^{m} R_k(\vecw)$. The powers, denoted by $\star$, must be such that the following hold:
\begin{enumerate} \label{Conditions: deg-csp case}
    \item \label{pow-degsep-CSP} All the polynomials in $P \sqcup Q$ are degree separated from one another and from the $R_k$'s, with the powers being at least $\sigma + 1$. 
    \item \label{pow-finitechar-CSP} Over characteristic $p$ fields, where $p > 0$, all the powers are less than $p$ or are of the form $p^k -1$ for some $k \in \N$.
\end{enumerate} 
In Section \ref{etcsp-param}, we choose the powers to satisfy the above conditions over any field. The following observations hold by the definition of the polynomials and the above conditions.
\begin{observation} \label{Obs: deg sep clause poly}
For $k \in [m]$, the degree of $R_k$ is at most $\sigma + 4$ and $\supp (R_k(\vecw)) \leq \sigma + 1$ with equality in both cases when $a_{k,j} = 1$ for all $ j \in C_k$. For $k, l \in [m]$ with $l \neq k$, any monomial of $R_k$ is degree separated from every monomial of $R_l$ with respect to some variable in $\vecx \sqcup \vecy$. 
\end{observation}
\begin{observation} \label{Obs: low support properties}
$\cal S(f(\vecw)) = O(n^\sigma + m)$ and $\supp(f(\vecw)) = \sigma + 1$.
\end{observation}
\begin{observation} \label{Obs: support of orbit poly}
 For any $A \in \GL(|\vecw|,\F)$, it holds that
\[\supp (f(A\vecw)) = \max\left(\max_{g \in P, h \in Q}\big(\supp (g(A\vecw)),\supp (h(A\vecw))\big),\supp \left(\sum_{k=1}^{m}R_k(A\vecw)\right)\right).\]
\end{observation}
\noindent \textbf{Note:} There are three differences between the current construction and that in \cite{BDSS24}:
\begin{enumerate}
    \item In \cite{BDSS24}, the clause polynomial $R_k$ was defined as:
    \[R_k(\vecw) := (\prod_{j \in C_k}(y_j - a_{k,j}x_j)^2 )(z_1z_2\dots z_{\sigma -5})^{\star}.\]
    The monomial $(z_1z_2\dots z_{\sigma-5})^{\star}$ was used to make the $R_k$'s degree separated from one another and from the polynomials in $P \sqcup Q$ (by assigning appropriate values to the powers $\star$), and to ensure that some $R_k$ has support $\sigma + 1$. The former condition necessitated imposing $\sigma \geq 6$. In the current construction, we leverage the fact that a variable can be complemented or uncomplemented in a clause to make the $R_k$'s degree separated from one another with respect to some variable in $\vecx \sqcup \vecy$ (see Observation \ref{Obs: deg sep clause poly}) regardless of the presence or absence of the monomial $z_1z_2\dots z_{\sigma-5}$, which improves $\sigma \geq 6$ to $\sigma \geq 5$. We still need $z_1z_2 \dots z_{\sigma-5}$ in the $R_k$'s to ensure some $R_k$ has support $\sigma + 1$ when $\sigma \geq 6$.

    \item In \cite{BDSS24}, Condition \ref{pow-degsep-CSP} stated that all polynomials in $P \sqcup Q \sqcup R$ are degree separated with the powers being at least $\sigma + 1$, which implied that $\supp(f(A\vecw))$, for any $A \in \GL(|\vecw|,\F)$ can be analyzed by analyzing the support of individual polynomials. In contrast, the condition is now less restrictive as the nature of the polynomials in $P \sqcup Q$ is such that if these polynomials are degree separated amongst themselves and from the $R_k's$, and if $A \in \GL(|\vecw|,\F)$ is such that $\supp(f(A\vecw)) \leq \sigma$, then $A$ has a certain structure as implied by Observation \ref{Obs: support of orbit poly} and Lemmas \ref{P1Sep} and \ref{P2Sep}. The structured nature of $A$ implies that $R_k(A\vecw)$ looks similar to $R_k(\vecw)$ and then an argument similar to the proof of Observation \ref{Obs: deg sep clause poly} shows that it suffices to analyze the support of individual $R_k(A\vecw)$'s which leads to the recovery of a satisfying assignment (see the proof of Proposition \ref{Proposition: ETsupport reverse direction}).

    \item In \cite{BDSS24}, the values assigned to the powers, $\star$, were dependent on $m$ (the number of clauses), while in this version the powers are independent of $m$.
\end{enumerate}

\subsection{The forward direction} \label{Section: supp fwd direction}
Proposition \ref{Proposition: supp fwd direction} shows how a satisfying assignment for $\psi$ implies the existence of an invertible $A$, such that $\supp(f(A\vecw)) = \sigma$, by constructing $A$ from the satisfying assignment.

\begin{proposition} \label{Proposition: supp fwd direction}
Let $\psi \in \SAT$ with $(u_1,\dots,u_n) \in \{0,1\}^n$ a satisfying assignment. Then, $\supp(f(A\vecw)) = \sigma$, where the transform $A$ is defined as
\begin{equation} \label{actioncsp}
        \begin{split}
            A: z_j \mapsto z_j, \  x_i \mapsto x_i, \ y_i \mapsto y_i + (1-u_i)x_i \ \ i \in [n], j \in [\sigma - 5].
        \end{split}
    \end{equation}
\end{proposition}
\begin{proof}
As Condition \ref{pow-degsep-CSP} and Observation \ref{Obs: support of orbit poly} are satisfied, it suffices to analyse the action of $A$ on individual polynomials of $P \sqcup Q$ and then analyse $\supp(\sum_{k=1}^{m}R_k(A\vecw))$. Clearly, for $g(\vecw) \in P$, $g(A\vecw) = g(\vecw)$. Let $h(\vecw) \in Q$. Then $h(\vecw)$ is of form $((x_{t_1}y_{t_1})\cdots (x_{t_{\frac{\sigma}{2}}} y_{t_{\frac{\sigma}{2}}}))^\star$, $t_j \in [n]$, and $A$ acts on $h$ as:
    \[A: ((x_{t_1}y_{t_1})\cdots (x_{t_{\frac{\sigma}{2}}} y_{t_{\frac{\sigma}{2}}}))^\star \mapsto 
     ((x_{t_1})(y_{t_1} + (1-u_{t_1}) x_{t_1})\cdots (x_{t_{\frac{\sigma}{2}}})( y_{t_{\frac{\sigma}{2}}} + (1- u_{t_{\frac{\sigma}{2}}})x_{t_{\frac{\sigma}{2}}}))^\star.\]
Note $|\cup_{t=1}^{\sigma} \var(\ell_t)| = \sigma$, where $\ell_t = A(w)$ and $w \in \var(h(\vecw))$. Thus, $\supp(h(A\vecw)) \leq \sigma$. When $\sigma$ is odd, a similar argument holds for the modified construction of $Q$. For $k \in [m]$, $A$ acts on $R_k(\vecw)$ as:
    \[A: (\prod_{j \in C_k}(y_{j} - a_{k,j} x_j )^{2+a_{k,j}}) \cdot (z_{1} \cdots z_{\sigma-5} ) \mapsto (\prod_{j \in C_k}(y_{j} + (1 - a_{k,j} - u_{j}) x_j )^{2+a_{k,j}} \cdot (z_{1} \cdots z_{\sigma-5} ).\]
If $a_{k,j} \neq u_{j}$, then $a_{k,j} = 1 - u_{j}$. Since $\psi$ is satisfiable, therefore for all $k \in [m]$, $a_{k,j} \neq u_{j}$ for some $j \in C_k$. Hence, $\supp(R_k(A\vecw)) \leq (\sigma-5) + 5 = \sigma$ for all $k \in [m]$. This implies $\supp(\sum_{k=1}^{m}R_k(A\vecw))  \leq \sigma$. Thus, $\supp(f(A\vecw)) = \sigma$.
\end{proof}
\subsection{The reverse direction} \label{Subsection: csp bwd direction}
Now, we show that if $\supp(f(A\vecw)) \leq \sigma$ for $A \in \GL(|\vecw|,\F)$, then a satisfying assignment can be recovered for $\psi$. Lemmas \ref{P1Sep} and \ref{P2Sep}, proved using Claim \ref{FullSupport}, together show that $A$ is as:
\[A: z_j \mapsto z_j, \ x_i \mapsto x_i, \ y_i \mapsto y_i + c_ix_i \ \ \ c_i \in \F, \ j \in [\sigma-5], i \in [n]\]
without loss of generality.\footnote{as permutation and non-zero scaling of variables do not affect the support.} Proposition \ref{Proposition: ETsupport reverse direction} derives a satisfying assignment for $\psi$ from $A$.

\begin{lemma} \label{P1Sep}
    If $\supp(f(A\vecw)) \leq \sigma$, then $\forall w \in \vecz \sqcup \vecx$, $A(w) = W$, for scaled variable $W \in \vecw$.
\end{lemma}
\begin{lemma} \label{P2Sep}
    If $\supp(f(A\vecw)) \leq \sigma$, then $A(x_i) = X_i$ and $A(y_i) = Y_i + c_i X_i$, for scaled variables $Y_i,X_i \in \vecw$ and $c_i \in \F$.
\end{lemma}
\begin{proposition} \label{Proposition: ETsupport reverse direction}
    A satisfying assignment $\vecu$ for $\psi$ can be extracted from $A$.
\end{proposition}
\begin{proof}
    The action of $A$ on $R_{k}$, where $k \in [m]$, is:
    \[(\prod_{j \in C_k}(y_{j} - a_{k,j} x_j )^{2+a_{k,j}}) \cdot (z_{1} \cdots z_{\sigma-5} ) \mapsto (\prod_{j \in C_k}(y_{j} + (c_j - a_{k,j}) x_j )^{2+a_{k,j}}) \cdot (z_{1} \cdots z_{\sigma-5} ).\]
    An argument similar to the proof of Observation \ref{Obs: deg sep clause poly} shows that for any $k, l \in [m]$ with $k \neq l$, any monomial of $R_k(A\vecw)$ is degree separated from every monomial of $R_l(A\vecw)$ with respect to some variable in $\vecx \sqcup \vecy$. This implies $\supp(\sum_{k=1}^{m}R_k(A\vecw)) = \max_{k \in [m]}(\supp (R_k(A\vecw))$. Since $\supp (f(A\vecw)) \leq \sigma$, therefore $\supp (R_k(A\vecw)) \leq \sigma$ for all $k \in [m]$. Using the fact that the characteristic of the field is not equal to $2$ or $3$ and the action of $A$ of $R_k(\vecw)$, it can be seen that  $\supp(R_k(A\vecw)) \leq \sigma$ iff for some $j \in C_k$, $c_j = a_{k,j}$. Hence, for each $R_k(\vecw)$, there exists $j \in C_k$ such that $c_j \in \{0,1\}$. Construct $\vecu \in \{0,1\}^n$ by setting $u_j := 1 - c_{j}$, for appropriate $j \in C_k$ and the remaining $u_i$'s to arbitrary values in $\{0,1\}$. From the definition of $\vecu$, it follows that for the $k\ith$ clause, there exists $j \in C_k$ such that $u_{j} \neq a_{k,j}$. As $k$ is arbitrary, all clauses are satisfied.
\end{proof}

%% file: SparseShift.tex
\section{$\NP$-hardness of $\sparseshift$ and $\gapsparseshift$} \label{section-shifteqvhard}
In Section \ref{section-sparseshiftNPH}, we adapt the reduction and analysis of Section \ref{Subsection: Reduction ETSparse} to prove part one of Theorem \ref{Theorem: ETSparse_shift NP-hard}. In Section \ref{section-gapsparseshiftNPH}, we prove part two of Theorem \ref{Theorem: ETSparse_shift NP-hard} by an analysis similar to that in Section \ref{Subsection: gap reduction}. Note that for uniformity of presentation, we show the reduction and analysis over fields; the same also holds over any integral domain. In Section \ref{Sec:SETsparse prev work compare}, we briefly describe the reduction in \cite{ChillaraGS23} and then compare our results with theirs. The proofs of the observations and lemmas in this section can be found in Appendix \ref{Section: sparse-shift proofs}.
\subsection{Hardness of $\sparseshift$} \label{section-sparseshiftNPH}

\paragraph{Proof sketch.}The reduction is similar to that in Section \ref{Subsection: Reduction ETSparse}. Each variable and clause of a $\CNF$ $\psi$ is mapped to a distinct polynomial and these polynomials are degree separated from one another with respect to a fixed variable. Summing these polynomials gives the polynomial $f$. The degrees are chosen such that $f$ is shift-equivalent to an $s$-sparse polynomial (for a suitable sparsity parameter $s$) if and only if $\psi \in \SAT$. We first describe the reduction for fields $\F$ of characteristic not equal to $2$ and then do so for characteristic $2$ fields.

Formally, let $\psi$, $\vecx$ and $x_0$ be as denoted in Section \ref{Subsection: Reduction ETSparse}. Let $\vecz := \vecx \sqcup \{x_0\}$. For $\vecb \in \F^{|\vecz|}$, $\vecb_{|w}$ denotes the component of $\vecb$ corresponding to the variable $w \in \vecz$.  Let $d_1$, $d_2$, $d_3 \in \N$. Consider the following polynomials:

\begin{itemize}
    \item Corresponding to $x_i$, where $i\in [n]$, define $Q_i(\vecz)$ as: 
\begin{equation*} 
\begin{split}
    Q_i(\vecz) &:= Q_{i,1}(\vecz) + Q_{i,2}(\vecz), \  Q_{i,1}(\vecz) := x_0^{(2i-1)(d_2+1)}(x_i + 1)^{d_2} \text{ and } \\  Q_{i,2}(\vecz) &:= x_0^{2i(d_2+1)}(x_i - 1)^{d_2}. 
\end{split}
\end{equation*}

\item For the $k\ith$ clause, $k \in [m]$,  define $R_k(\vecz) := x_0^{k}\prod_{j\in C_k}(x_j + (-1)^{a_{k,j}})^{d_3}.$ 
\end{itemize}
Define $s:= 1 + n(d_2+2) + m(d_3+1)^2$. Impose the following conditions on the $d_i$'s:
\begin{equation} \label{Sparse-shift-cond}
    d_1 \geq 4n(d_2+1)+2d_2 + 2, \
    d_2 \geq m(d_3+1)^2 + 1, \
    d_3 \geq m.
\end{equation}
Finally, define $f(\vecz)$ as:
\begin{equation} \label{Sparse-shift-poly}
  f(\vecz) := x_0^{d_1} + \sum_{i=1}^{n} Q_i(\vecz) + \sum_{k=1}^{m} R_k(\vecz).  
\end{equation}

The $d_i$'s are chosen in Section \ref{Choice of degrees-shift} such that they are $(mn)^{O(1)}$ while also satisfying the inequalities of \eqref{Sparse-shift-cond}. Observations \ref{Obs: deg separated poly-shift}, \ref{Obs: shift-poly degree} and \ref{Obs: shift-poly sparsity} hold under the conditions of \eqref{Sparse-shift-cond}.
\begin{observation} \label{Obs: deg separated poly-shift}
For all $i \in [n]$, $k \in [m]$, the polynomials $x_0^{d_1}$, $Q_{i,1}(\vecz)$, $Q_{i,2}(\vecz)$ and $R_k(\vecz)$ are degree separated from one another with respect to $x_0$. Also, $Q_i(\vecz)$ is degree separated from other $Q_j(\vecz)$'s with respect to $x_0$, for $i,j \in [n]$ and $i \neq j$. Similarly, $R_k(\vecz)$ is degree separated from $R_l(\vecz)$ with respect to $x_0$ for $k,l \in [m]$ and $k \neq l$.
\end{observation}
\begin{observation} \label{Obs: shift-poly degree}
The degree of $f$ is $d_1 = (mn)^{O(1)}$. Also, $\frac{d_1}{2} > s$.
\end{observation}
\begin{observation} \label{Obs: shift-poly sparsity}
    $\cal S(f(\vecz)) = 1 + n(2d_2+2) + m(d_3+1)^3$ and $\supp(f) = 4$.
\end{observation}
\paragraph{The forward direction.} If $\vecu \in \{0,1\}^n$ is such that $\psi(\vecu) = 1$ and $f$, as described in \eqref{Sparse-shift-poly}, is the polynomial corresponding to $\psi$, then Proposition \ref{Proposition: fwd direction-shift} shows that there exists a $\vecb \in \F^{|\vecz|}$, defined using $\vecu$, such that $\cal S(f(\vecz + \vecb)) \leq s$ holds.

\begin{proposition} \label{Proposition: fwd direction-shift}
    Let $\vecu \in \{0,1\}^n$, such that $\psi(\vecu) = 1$. Then $\cal S(f(\vecz + \vecb)) \leq s$, where $\vecb \in \{-1,0,1\}^{|\vecz|}$ is defined as follows:
    \begin{equation} \label{Eqn: sparse shift vector}
        \vecb_{|x_0} := 0, \vecb_{|x_i} := (-1)^{u_i}, i \in [n]
    \end{equation}
\end{proposition}
\begin{proof}
It follows from the fact that $\vecb_{|x_0}$ is $0$, the definition of $f$ in \eqref{Sparse-shift-poly} and the choice of the $d_i$'s that $x_0^{d_1}$ and the remaining summands of $f(\vecz + \vecb)$ are degree separated from one another with respect to $x_0$ (The proof is similar to that of Lemma \ref{Lemma: degree separation-shift}). Thus, by Observation \ref{lem_deg_sep_sum} 
\[\cal S(f(\vecz + \vecb)) = \cal S(x_0^{d_1}) + \sum_{i=1}^{n} \cal S(Q_i(\vecz + \vecb)) + \sum_{k=1}^{m} \cal S(R_k(\vecz + \vecb)).\]
Thus, it suffices to analyze the sparsity of $Q_i(\vecz + \vecb)$'s and $R_k(\vecz + \vecb)$'s. We now analyze $\cal S(Q_i(\vecz + \vecb))$ for $i \in [n]$. If $u_i = 0$, then
    \begin{align*}
     Q_{i,1}(\vecz + \vecb) = x_0^{(2i-1)(d_2+1)}(x_i+2)^{d_2}, \ Q_{i,2}(\vecz + \vecb) =  x_0^{2i(d_2+1)} x_i^{d_2}.
    \end{align*} 
    If $u_i = 1$, then 
    \begin{align*}
     Q_{i,1}(\vecz + \vecb) = x_0^{(2i-1)(d_2+1)}x_i^{d_2}, \ Q_{i,2}(\vecz + \vecb) =  x_0^{2i(d_2+1)} (x_i-2)^{d_2}.
    \end{align*}  
By the binomial theorem, the choice of the $d_2$ and Lucas's theorem (for the finite characteristic case), if $u_i = 0$, then $\cal S(Q_{i,1}(\vecz + \vecb)) = d_2 + 1$ and $\cal S(Q_{i,2}(\vecz + \vecb)) = 1$ and, if $u_i = 1$ then $\cal S(Q_{i,1}(\vecz + \vecb)) = 1$ and $\cal S(Q_{i,2}(\vecz + \vecb)) = d_2 + 1$. In either case, 
\[\cal S(Q_i(\vecz + \vecb)) = \cal S(Q_{i,1}(\vecz + \vecb)) + \cal S(Q_{i,2}(\vecz + \vecb)) = d_2+2.\] 
For the $k\ith$ clause, $k \in [m]$,the corresponding polynomial $R_k(\vecz + \vecb)$ is as:
    \begin{equation*}
   R_k(\vecz + \vecb) = x_0^{k}\prod_{j \in C_k}(x_j+ (-1)^{a_{k,j}}+(-1)^{u_j})^{d_3}.
\end{equation*}
As the multiplicands in $R_k(\vecz + \vecb)$ do not share any variables, $\cal S(R_{k}(\vecz + \vecb))$ is the product of the sparsity of the multiplicands. Since $\psi(\vecu) = 1$, therefore in the $k\ith$ clause there exists $j \in C_k$ such that $a_{k,j} \neq u_j$. For that $j$, $(x_j+(-1)^{a_{k,j}} + (-1)^{u_j})^{d_3} = x_j^{d_3}$. As at least one literal is true in every clause under $\vecu$, $\cal S(R_{k}(\vecz + \vecb)) \leq (d_3+1)^2$ by the choice of $d_3$ and Lucas's theorem (for the finite characteristic case). Thus,
\[\cal S(f(\vecz + \vecb)) = \cal S(x_0^{d_1}) + \sum_{i=1}^{n} \cal S(Q_i(\vecz + \vecb)) + \sum_{k=1}^{m} \cal S(R_k(\vecz + \vecb)) \leq 1 + n(d_2+2) + m(d_3+1)^2 = s.\] 
\end{proof}

\paragraph{The reverse direction.} We leverage the constraints in \eqref{Sparse-shift-cond} to show that if $\cal S(f(\vecz + \vecb)) \leq s$ for some $\vecb \in \F^{|\vecz|}$, then $\vecb$ is as described in \eqref{Eqn: sparse shift vector}. Lemma \ref{Lemma: x0 fixed-shift} shows that $\vecb_{|x_0} = 0$. Then, Lemma \ref{Lemma: degree separation-shift} shows that the summands of $f(\vecz + \vecb)$ are degree separated with respect to $x_0$. 

\begin{lemma} \label{Lemma: x0 fixed-shift}
     If $\cal S(f(\vecz + \vecb)) \leq s$, then $\vecb_{|x_0} = 0$.
\end{lemma}
\begin{lemma} \label{Lemma: degree separation-shift}
    For all $i \in [n]$, $k \in [m]$, $x_0^{d_1}$, $Q_{i,1}(\vecz + \vecb)$, $Q_{i,2}(\vecz + \vecb)$ and $R_k(\vecz + \vecb)$ are degree separated from one another with respect to $x_0$. Also, $Q_i(\vecz + \vecb)$ is degree separated with respect to $x_0$ from other $Q_j(\vecz + \vecb)$'s, for $i,j \in [n]$ and $i \neq j$. Similarly, $R_k(\vecz + \vecb)$ is degree separated with respect to $x_0$ from $R_l(\vecz + \vecb)$ for $k,l \in [m]$ and $k \neq l$.
\end{lemma}
\vspace{-7mm}
\[ \therefore  \cal S(f(\vecz + \vecb)) = \cal S(x_0^{d_1}) + \sum_{i=1}^{n} \cal S(Q_i(\vecz + \vecb)) + \sum_{k=1}^{m} \cal S(R_k(\vecz + \vecb)) \text{ by Lemma \ref{Lemma: degree separation-shift}.}\]
Lemma \ref{Qi sparsity analysis: shift} analyses the sparsity of $Q_i(\vecz+\vecb)$. The proof of Lemma \ref{Lemma: equality condition Qi-shift} uses the condition $d_2 \geq m(d_3+1)^2 + 1$. 
\begin{lemma} \label{Qi sparsity analysis: shift}
    For any $\vecb \in \F^{|\vecz|}$ and $i \in [n]$:
        \begin{equation*}
            \cal S(Q_i(\vecz+\vecb)) = \cal S(Q_{i,1}(\vecz+\vecb))+\cal S(Q_{i,2}(\vecz+\vecb)) \geq d_2 + 2,
        \end{equation*}
    where $Q_i$, $Q_{i,1}$ and $Q_{i,2}$ are as defined in this section. Equality holds if and only if $\vecb_{|x_i} \in \{-1,1\}$. Further, if $\cal S(Q_i(\vecz+\vecb)) \neq d_2 + 2$, then $\cal S(Q_i(\vecz+\vecb)) = 2d_2 + 2$.
\end{lemma}
\begin{lemma} \label{Lemma: equality condition Qi-shift}
Under the given $\vecb$, $\cal S(Q_i(\vecz + \vecb)) = d_2 + 2$ holds for all $i \in [n]$.
\end{lemma}
Lemmas \ref{Lemma: x0 fixed-shift}, \ref{Qi sparsity analysis: shift} and \ref{Lemma: equality condition Qi-shift} together show that $\vecb$ is as described in \eqref{Eqn: sparse shift vector}. Then, Proposition \ref{Proposition: bwd direction-shift} uses $d_3 \geq m$ to show how a satisfying assignment for $\psi$ can be recovered from $\vecb$.
\begin{proposition} \label{Proposition: bwd direction-shift}
   With $\vecb$ as described in \eqref{Eqn: sparse shift vector}, $\vecu = (u_1, \dots , u_n)$ is a satisfying assignment for $\psi$.  
\end{proposition}
\begin{proof}
Suppose not; then there exists $k \in [m]$ such that the $k\ith$ clause, $\lor_{j \in C_k}(x_j\oplus a_{k,j})$, in $\psi$ is unsatisfied. Since this clause is unsatisfied, $u_j = a_{k,j}$ for all $j \in C_k$. Thus, $R_k(\vecz + \vecb) = x_0^{k} \prod_{j \in C_k}(x_j\pm 2)^{d_3}$ and $\cal S(R_k(\vecz + \vecb)) = (d_3+1)^3\geq (m+1)(d_3+1)^2$, where the equality holds by the fact that $R_k(\vecz + \vecb)$ is a product of linear forms not sharing variables, the binomial theorem, the choice of $d_3$ and Lucas's Theorem (for finite characteristic fields), and the inequality holds as $d_3 \geq m$. By the definition of $f$ and $s$, Observations \ref{Lemma: degree separation-shift} and \ref{lem_deg_sep_sum}, it holds that
\begin{equation*}
    \begin{split}
            \cal S(f(\vecz + \vecb)) & \geq  \cal S((x_0)^{d_1}) + \sum_{i = 1}^{n} \cal S(Q_i(\vecz+\vecb)) +\cal S(R_k(\vecz+\vecb)) \\
            &\geq  1 + n(d_2 + 2) + m(d_3+1)^2 + (d_3+1)^2 = s + (d_3+1)^2 > s,
    \end{split}
\end{equation*}
a contradiction. Thus, $\vecu$ is a satisfying assignment for $\psi$.
\end{proof}

\subsubsection*{\underline{Construction for characteristic $2$ fields}}
Since over characteristic $2$ fields, $x_i - 1$ and $x_i + 1$ are the same polynomials, we need to modify the previous construction. Moreover, the sparsifying vector will also be slightly different. Formally, let $\psi$, $\vecx, x_0,$ and $\vecz$ be as denoted earlier. Let $d_1$, $d_2$, $d_3 \in \N$. Consider the following polynomials:
\begin{itemize}
    \item Corresponding to $x_i$, where $i\in [n]$, define $Q_i(\vecz)$ as: 
\begin{equation*} 
\begin{split}
    Q_i(\vecz) &:= Q_{i,1}(\vecz) + Q_{i,2}(\vecz), \  Q_{i,1}(\vecz) := x_0^{(2i-1)(d_2+1)}x_i^{d_2}, Q_{i,2}(\vecz) := x_0^{2i(d_2+1)}(x_i + 1)^{d_2}. 
\end{split}
\end{equation*}
\item For the $k\ith$ clause, $k \in [m]$,  define $R_k(\vecz) := x_0^{k}\prod_{j\in C_k}(x_j + a_{k,j})^{d_3}.$ 
\end{itemize}
Define $s:= 1 + n(d_2+2) + m(d_3+1)^2$. Choose $d_i$'s as specified in Section \ref{Choice of degrees-shift} so that they are $(mn)^{O(1)}$ and satisfy the conditions of \eqref{Sparse-shift-cond}. Finally, define $f(\vecz)$ as:
\begin{equation} \label{affine variable poly char 2-shift}
  f(\vecz) := x_0^{d_1} + \sum_{i=1}^{n} Q_i(\vecz) + \sum_{k=1}^{m} R_k(\vecz).  
\end{equation}
Observations \ref{Obs: deg separated poly-shift}, \ref{Obs: shift-poly degree} hold with little change. Observation \ref{Obs: poly sparsity char 2-shift} analyses the sparsity and support of $f$ and has a proof similar to that of Observations \ref{Obs: translated poly sparsity}.

\begin{observation} \label{Obs: poly sparsity char 2-shift}
    $\cal S(f(\vecz)) \leq 1 + n(d_2+2) + m(d_3+1)^3$ and $\supp(f) = 4$.
\end{observation}

\begin{example} \label{Remark: char two not natural shift}
The sparsity of the polynomial output by the reduction over characteristic $2$ fields depends on the number of variables which are complemented within a clause. Hence, for the same number of variables $n$ and the same number of clauses $m$, the output polynomial corresponding to two different $\psi$'s may have different sparsity. Thus, the reduction is not natural over characteristic $2$ fields.
\end{example}
\paragraph{The forward direction.} Let $\vecu \in \{0,1\}^n$ be such that $\psi(\vecu) = 1$ and $f$, as described in \eqref{affine variable poly char 2-shift}, be the polynomial corresponding to $\psi$. Proposition \ref{Proposition: fwd direction char 2-shift} shows how $\vecu$ can be used to construct a sparsifying vector $\vecb$ and has a proof similar to that of Proposition \ref{Proposition: fwd direction-shift}.
\begin{proposition} \label{Proposition: fwd direction char 2-shift}
    Let $\vecu \in \{0,1\}^n$, such that $\psi(\vecu) = 1$. Then $\cal S(f(\vecz + \vecb)) \leq s$, where $\vecb \in \{0,1\}^{|\vecz|}$ is defined as follows:
    \begin{equation} \label{Eqn: sparse shift vector char 2}
        \vecb_{|x_0} := 0, \vecb_{|x_i} := 1 - u_i, i \in [n]
    \end{equation}
\end{proposition}

\paragraph{The reverse direction.} Let $\vecb \in \F^{|\vecz|}$ be such that $\cal S(f(\vecz + \vecb)) \leq s$. The analysis of the reverse direction in the previous section holds with some changes. Formally, Lemma \ref{Lemma: x0 fixed-shift} holds without any change in its proof, while Lemma \ref{Lemma: degree separation-shift} holds with some change in its statement and proof. Thus, $\vecb_{|x_0} = 0$. Lemma \ref{Qi sparsity analysis: shift char 2} analyses $\cal S(Q_i(\vecz + \vecb))$, where $i \in [n]$, and its proof is similar to that of Lemma \ref{Qi sparsity analysis: shift}. 
\begin{lemma} \label{Qi sparsity analysis: shift char 2}
    For any $\vecb \in \F^{|\vecz|}$ and $i \in [n]$:
        \begin{equation*}
            \cal S(Q_i(\vecz+\vecb)) = \cal S(Q_{i,1}(\vecz+\vecb))+\cal S(Q_{i,2}(\vecz+\vecb)) = d_2 + 2,
        \end{equation*}
    where $Q_i$, $Q_{i,1}$ and $Q_{i,2}$ are as defined in this section.
\end{lemma}
\noindent Lemmas \ref{Lemma: x0 fixed-shift} and \ref{Qi sparsity analysis: shift char 2} together show that $\vecb$ is as described in \eqref{Eqn: sparse shift vector char 2}. Then, Proposition \ref{Proposition: bwd direction char 2-shift} shows how a satisfying assignment for $\psi$ can be recovered from $\vecb$ and can be proved similarly as Proposition \ref{Proposition: bwd direction-shift}.

\begin{proposition} \label{Proposition: bwd direction char 2-shift}
   With $\vecb$ as described in \eqref{Eqn: sparse shift vector char 2}, $\vecu = (u_1, \dots , u_n)$ is a satisfying assignment for $\psi$.  
\end{proposition}

\subsubsection{Choice of degrees} \label{Choice of degrees-shift}
\paragraph{For characteristic $0$ fields.} In this case, the inequalities in \eqref{Sparse-shift-cond} can be converted to equalities. Thus
\begin{equation*}
    \begin{split}
        d_3 &= m, \ d_2 = m(m+1)^2+1 = O(m^3) \implies s = O(nm^3), \\
        d_1 &= 4n(d_2+1) + 2d_2 + 2 = O(nm^3).
    \end{split}
\end{equation*}
\paragraph{For finite characteristic fields.} Let $\textnormal{char}(\F) = p > 0$. If $p > d_1$, where the value of $d_1$ is as in the $\textnormal{char}(\F) = 0$ case, then the $d_i$'s are set as per the $\textnormal{char}(\F) = 0$ case. Otherwise  $p = O(nm^3)$. In such a case, we choose the $d_i$'s to be of form $p^k - 1$ for some $k \in \N$. Then
\begin{equation*}
    \begin{split} 
        d_3 &\leq pm, \ d_2 \leq pm(d_3+1)^2+p = O(p^3m^3) \implies s = O(nm^3p^3), \\ 
        d_1 &\leq p(4n(d_2+1) + 2d_2 + 2) = O(nm^3p^4).
    \end{split}
\end{equation*}
As $p = O(nm^3)$, therefore
\[        
d_1 = O(n^5m^{15}), \ d_2 = O(n^3m^{12}),\  d_3= O(nm^4) \text{ and }  s = O(n^4m^{12}).
\]
\subsection{Hardness of $\gapsparseshift$} \label{section-gapsparseshiftNPH}
\paragraph{Proof sketch.} Like in Section \ref{Subsection: gap reduction}, for a $\CNF$ $\psi$, we carefully analyze the sparsity of the corresponding polynomial $f$ as defined in Section \ref{section-sparseshiftNPH}. For unsatisfiable $\psi$'s, we do a slightly deeper analysis on $\cal S(f(\vecz + \vecb))$, for all $\vecb \in \F^{|\vecz|}$, to show a lower bound. For satisfiable $\psi$'s, $\cal S(f(\vecz + \vecb))$ has already been upper bounded for an appropriate $\vecb \in \F^{|\vecz|}$. The degree parameters are also chosen differently so that the gap between the lower bound and the upper bound is significant. Comparing the sparsities for satisfiable and unsatisfiable $\psi$'s proves part two of Theorem \ref{Theorem: ETSparse_shift NP-hard}. Throughout this section, we assume $\epsilon \in (0,\frac{1}{3})$ to be an arbitrary constant.
\subsubsection{For characteristic $0$ fields} \label{subsection-gap-char0}
For a $\CNF$ $\psi$, consider the polynomial $f$ as defined in \eqref{Sparse-shift-poly}. Choose the $d_i$'s to satisfy the following while also being $(mn)^{O(1)}$:
\begin{equation}\label{Eqn: gap ineq-shift}
    d_3 \geq \max(4mn,(mn)^{O(1/\epsilon)}), \ 
    d_2 = m(d_3+1)^2 + 1, \ 
    d_1 = d_2^2 + 4n(d_2+1) + 2d_2 + 2. 
\end{equation}
Under these constraints, the conditions in \eqref{Sparse-shift-cond} are also satisfied. Let $s := \cal S(f)$. Then, $s = 1+n(2d_2+2) + m(d_3+1)^3$ by Observation \ref{Obs: shift-poly sparsity}. For $\psi \in \bar{\SAT}$, Lemma \ref{Lemma: unsat gap-shift} shows lower bounds on $\cal S(f(\vecz + \vecb))$ for all $\vecb \in \F^{|\vecz|}$. In the lemma, Item \ref{case 1 unsat gap-shift} is essentially Lemma \ref{Lemma: x0 fixed-shift}, and Item \ref{case 2 unsat gap-shift} is the analysis in Proposition \ref{Proposition: bwd direction-shift}. Thus, Lemma \ref{Lemma: unsat gap-shift} encapsulates the analysis of the reverse direction in Section \ref{section-sparseshiftNPH}.  For $\psi \in \SAT$, by Proposition \ref{Proposition: fwd direction-shift}, there exists $\vecb \in \F^{|\vecz|}$ such that $\cal S(f(\vecz + \vecb)) \leq s_0$, where $s_0 = 1+n(d_2+2) + m(d_3+1)^2$. Proposition \ref{Proposition: Sparsity gap-shift} shows $\gapsparseshift$ is $\NP$-hard by comparing the sparsity for satisfiable and unsatisfiable $\psi$'s, and uses Lemma \ref{Lemma: unsat gap-shift} and the conditions in \eqref{Eqn: gap ineq-shift}. 

\begin{lemma} \label{Lemma: unsat gap-shift}
   Let $\psi \in \bar{\SAT}$, $f$ be as defined in \eqref{Sparse-shift-poly} corresponding to $\psi$ and $\vecb \in \F^{|\vecz|}$.
   \begin{enumerate}
       \item \label{case 1 unsat gap-shift} If $\vecb_{|x_0} \neq 0$, then $\cal S(f(\vecz + \vecb)) \geq \frac{d_1}{2}$.

       \item \label{case 2 unsat gap-shift} If $\vecb_{|x_0} = 0$, then $\cal S(f(\vecz+\vecb)) \geq (d_3+1)^3$.
   \end{enumerate}
\end{lemma}
\begin{proposition} \label{Proposition: Sparsity gap-shift}
 Let $\textnormal{char}(\F) = 0$. If the input in $\gapsparseshift$ is an $s$-sparse polynomial, then $\gapsparseshift$ is $\NP$-hard for $\alpha = s^{1/3-\epsilon}$.
\end{proposition}
\begin{proof}
If $\psi \in \SAT$, then $\cal S(f(\vecz+\vecb)) \leq s_0$, where $\vecb$ is as described in \eqref{Eqn: sparse shift vector}. If $\psi \in \bar{\SAT}$, then it follows from Lemma \ref{Lemma: unsat gap-shift} that for any $\vecb \in \F^{|\vecz|}$: 
\[\cal S(f(\vecz + \vecb)) \geq \min\left(\frac{d_1}{2},(d_3+1)^3\right).\]
The constraints imposed in \eqref{Eqn: gap ineq-shift} ensure that $(d_3+1)^3$ is the minimum. As $d_2 = m(d_3+1)^2 + 1$, therefore $s_0 = 1 + n(d_2+2) + m(d_3+1)^2 \leq 3nd_2 = 3mn(d_3+1)^2 + 3n \leq 4mn(d_3+1)^2$. Thus, the gap in the sparsities of the YES instances and the NO instances is 
\[  \frac{(d_3+1)^3}{s_0} \geq \frac{(d_3+1)^3}{4mn(d_3+1)^2} = \frac{d_3+1}{4mn}. \]
Also, as $d_3 \geq 4mn$, ${\cal S(f)} = s \leq 2m(d_3+1)^3 \implies d_3+1 \geq (\frac{s}{2m})^{1/3}$. Then, the gap is
\[  \frac{(d_3+1)^3}{s_0} \geq \frac{d_3+1}{4mn} \geq \frac{s^{1/3}}{2^{1/3}4m^{4/3}n}. \]
Finally, note that $s \geq d_3^3$. Thus, for $d_3^{3\epsilon} \geq (mn)^{O(1)}$ large enough, 
\[s^{\epsilon} \geq d_3^{3\epsilon} \geq {2^{1/3}4m^{4/3}n} \implies \frac{s^{1/3}}{2^{1/3}4m^{4/3}n} \geq s^{1/3 - \epsilon}.\]
Hence, the gap is at least $s^{1/3 - \epsilon}$. Therefore, $\SAT$ reduces to $\gapsparseshift$ for $\alpha = s^{1/3 - \epsilon}$.
\end{proof} 
\subsubsection{Extension to finite characteristic fields} \label{subsection-gap-charp}
Let the characteristic be $p > 0$. If $p > d_1$, where $d_1$ is as chosen in Section \ref{subsection-gap-char0}, then the argument of that section continues to hold. Otherwise, $p \leq d_1 = (mn)^{O(1)}$. In this case, choose the $d_i$'s to be of form $p^k - 1$ so that they are $(mn)^{O(1)}$ and satisfy the following conditions:
\begin{equation}\label{Eqn: gap ineq-shift-prime}
    d_3 \geq \max(3pmn,(mn)^{O(1/\epsilon)}), \ 
    d_2 > m(d_3+1)^2, \ 
    d_1 > d_2^2 + 4n(d_2+1) + 2d_2 + 2. 
\end{equation}
We can get $d_1 = O(p(d_2^2 + 4n(d_2+1) + 2d_2 + 2))$ and $d_2 = O(pm(d_3+1)^2)$. For this choice of the $d_i$'s, the conditions in \eqref{Sparse-shift-cond} are also satisfied. Let $s := \cal S(f)$. Then, $s = 1+n(2d_2+2) + m(d_3+1)^3$ by Observation \ref{Obs: shift-poly sparsity}. Note that Lemma \ref{Lemma: unsat gap-shift} also holds over finite characteristic fields due to the choice of the $d_i$'s and Lucas's Theorem. Thus, For $\psi \in \bar{\SAT}$, Lemma \ref{Lemma: unsat gap-shift} shows lower bounds on $\cal S(f(\vecz + \vecb))$ for all $\vecb \in \F^{|\vecz|}$. In the lemma, Item \ref{case 1 unsat gap-shift} is essentially Lemma \ref{Lemma: x0 fixed-shift}, and Item \ref{case 2 unsat gap-shift} is the analysis in Proposition \ref{Proposition: bwd direction-shift}. Thus, Lemma \ref{Lemma: unsat gap-shift} encapsulates the analysis of the reverse direction in Section \ref{section-sparseshiftNPH}.  For $\psi \in \SAT$, by Proposition \ref{Proposition: fwd direction-shift}, there exists $\vecb \in \F^{|\vecz|}$ such that $\cal S(f(\vecz + \vecb)) \leq s_0$, where $s_0 = 1+n(d_2+2) + m(d_3+1)^2$. Proposition \ref{Proposition: Sparsity gap-shift-prime} shows $\gapsparseshift$ is $\NP$-hard by comparing the sparsity for satisfiable and unsatisfiable $\psi$'s, and uses Lemma \ref{Lemma: unsat gap-shift} and the conditions in \eqref{Eqn: gap ineq-shift-prime}.
\begin{proposition} \label{Proposition: Sparsity gap-shift-prime}
 Let $\textnormal{char}(\F) = p > 2$. If the input in $\gapsparseshift$ is an $s$-sparse polynomial, then $\gapsparseshift$ is $\NP$-hard for $\alpha = s^{1/3-\epsilon}$.
\end{proposition}

\begin{proof}
The proof is similar to that of Proposition \ref{Proposition: Sparsity gap-shift}. If $\psi \in \SAT$, then $\cal S(f(\vecz + \vecb)) \leq s_0$ where $\vecb$ is as described in \eqref{Eqn: sparse shift vector} and $s_0 = 1 + n(d_2+2) + m(d_3+1)^2$. If $\psi \in \bar{\SAT}$, it follows from Lemma \ref{Lemma: unsat gap-shift} that for any $\vecb \in \F^{|\vecz|}$: 
\[\cal S(f(\vecz + \vecb)) \geq \min\bigg(\frac{d_1}{2},(d_3+1)^3\bigg).\]
The conditions imposed in \eqref{Eqn: gap ineq-shift-prime} ensure that $(d_3+1)^3 > s_0$ and $(d_3+1)^3$ is the minimum. As $d_2 > m(d_3+1)^2$, therefore $s_0 = 1 + n(d_2+2) + m(d_3+1)^2 \leq 3nd_2 \leq 3pmn(d_3+1)^2$. Consequently, the gap in the sparsities of the YES instances and NO instances is 
\[  \frac{(d_3+1)^3}{s_0} \geq \frac{(d_3+1)^3}{3pmn(d_3+1)^2} = \frac{d_3+1}{3pmn}. \]
Also, note that as $d_3 \geq 3pmn$, therefore $s \leq 2m(d_3+1)^3 \implies d_3+1 \geq (\frac{s}{2m})^{1/3}$. Then, the gap is
\[  \frac{(d_3+1)^3}{s_0} \geq \frac{d_3+1}{3pmn} \geq \frac{s^{1/3}}{p2^{1/3}3m^{4/3}n}. \]
Finally, note that $s \geq d_3^3$. Thus, for $d_3^{3\epsilon} \geq (mn)^{O(1)}$ large enough, 
\[s^{\epsilon} \geq d_3^{3\epsilon} \geq {p2^{1/3}3m^{4/3}n} \implies \frac{s^{1/3}}{p2^{1/3}3m^{4/3}n} \geq s^{1/3 - \epsilon}.\]
Hence, the gap is at least $s^{1/3 - \epsilon}$. Therefore, $\SAT$ reduces to $\gapsparseshift$ for $\alpha = s^{1/3 - \epsilon}$.
\end{proof}

\subsubsection*{\underline{Extension to characteristic $2$ fields}} 
Let the characteristic be $2$. Consider the polynomial as defined in \eqref{affine variable poly char 2-shift}. Let $s := \cal S(f)$. Now, $s$ depends on the number of variables complemented in a clause. To prove the hardness of $\gapsparseshift$, $s \geq d_3^3$ is required (see the proof of Proposition \ref{Proposition: Sparsity gap-shift-char2}), and this can be achieved if there is at least one clause where all the variables are complemented. Thus, assume, without loss of generality, that such a clause exists (see footnote \ref{footnote: complemented clause}). Choose the $d_i$'s to satisfy \eqref{Eqn: gap ineq-shift-prime} with $p$ set to $2$. By Observation \ref{Obs: poly sparsity char 2-shift} and the assumption on $\psi$, it holds that 
\[1 + n(d_2 + 2) + (d_3 + 1)^3 \leq s \leq 1 + n(d_2 + 2) + m(d_3 + 1)^3.\]
For $\psi \in \bar{\SAT}$, Lemma \ref{Lemma: unsat gap-shift}, with some changes in its statement and proof, shows lower bounds on $\cal S(f(\vecz + \vecb))$ for all $\vecb \in \F^{|\vecz|}$. In the lemma, Item \ref{case 1 unsat gap-shift} is essentially Lemma \ref{Lemma: x0 fixed-shift}, and Item \ref{case 2 unsat gap-shift} is the analysis in Proposition \ref{Proposition: bwd direction char 2-shift}. Thus, Lemma \ref{Lemma: unsat gap-shift} encapsulates the analysis of the reverse direction in Section \ref{section-sparseshiftNPH}. For $\psi \in \SAT$, by Proposition \ref{Proposition: fwd direction char 2-shift}, there exists $\vecb \in \F^{|\vecz|}$ such that $\cal S(f(\vecz + \vecb)) \leq s_0$, where $s_0 = 1+n(d_2+2) + m(d_3+1)^2$. Proposition \ref{Proposition: Sparsity gap-shift-char2} shows $\gapsparseshift$ is $\NP$-hard by comparing the sparsity for satisfiable and unsatisfiable $\psi$'s, and uses Lemma \ref{Lemma: unsat gap-shift} and the conditions in \eqref{Eqn: gap ineq-shift-prime}.

\begin{proposition} \label{Proposition: Sparsity gap-shift-char2}
 Let $\textnormal{char}(\F) = 2$. If the input in $\gapsparseshift$ is an $s$-sparse polynomial, then $\gapsparseshift$ is $\NP$-hard for $\alpha = s^{1/3-\epsilon}$.
\end{proposition}
\begin{proof}

The proof is similar to that of Proposition \ref{Proposition: Sparsity gap-shift}. If $\psi \in \SAT$, then $\cal S(f(\vecz + \vecb)) \leq s_0$ where $\vecb$ is as described in \eqref{Eqn: sparse shift vector char 2} and $s_0 = 1 + n(d_2+2) + m(d_3+1)^2$. If $\psi \in \bar{\SAT}$, it follows from Lemma \ref{Lemma: unsat gap-shift} that for any $\vecb \in \F^{|\vecz|}$: 
\[\cal S(f(\vecz + \vecb)) \geq \min\left(\frac{d_1}{2},(d_3+1)^3\right).\]
The conditions imposed in \eqref{Eqn: gap ineq-shift-prime} ensure that $(d_3+1)^3 > s_0$ and $(d_3+1)^3$ is the minimum. As $d_2 > m(d_3+1)^2$, therefore $s_0 = 1 + n(d_2+2) + m(d_3+1)^2 \leq 3nd_2 \leq 6mn(d_3+1)^2$. Consequently, the gap in the sparsities of the YES instances and NO instances is 
\[  \frac{(d_3+1)^3}{s_0} \geq \frac{(d_3+1)^3}{6mn(d_3+1)^2} = \frac{d_3+1}{6mn}. \]
Also, note that as $d_3 \geq 6mn$, therefore $s \leq 2m(d_3+1)^3 \implies d_3+1 \geq (\frac{s}{2m})^{1/3}$. Then, the gap is
\[  \frac{(d_3+1)^3}{s_0} \geq \frac{d_3+1}{6mn} \geq \frac{s^{1/3}}{2^{4/3}3m^{4/3}n}. \]
Finally, note that $s \geq d_3^3$. Thus, for $d_3^{3\epsilon} \geq (mn)^{O(1)}$ large enough, 
\[s^{\epsilon} \geq d_3^{3\epsilon} \geq {2^{4/3}3m^{4/3}n} \implies \frac{s^{1/3}}{2^{4/3}3m^{4/3}n} \geq s^{1/3 - \epsilon}.\]
Hence, the gap is at least $s^{1/3 - \epsilon}$. Therefore, $\SAT$ reduces to $\gapsparseshift$ for $\alpha = s^{1/3 - \epsilon}$.    
\end{proof}
\subsection{Comparison with the work of \cite{ChillaraGS23}} \label{Sec:SETsparse prev work compare}
In this section, we discuss the results by \cite{ChillaraGS23} on $s\text{-to-}(s-1)\sparseshift$, the variant of $\sparseshift$ as described before Theorem \ref{Theorem: ETSparse_shift NP-hard}, and on $\gapsparseshift$ alongwith a brief overview of the proofs, and compare their results with ours.

\paragraph{Result on $s\text{-to-}(s-1)\sparseshift$.} The authors of \cite{ChillaraGS23} showed that polynomial solvability over $R$, an integral domain that is \emph{not} a field, reduces to $s\text{-to-}(s-1)\sparseshift$, where the input polynomial can be given in the sparse representation or as a circuit. Thus, their reduction also implies $\sparseshift$ (considered over $R$) is at least as hard as polynomial solvability. Their reduction proceeds in two stages. In the first stage, they take as input a system of polynomial equations $T$ (given as circuits or in the sparse representation) and construct another system of equations $\tilde{T}$, where each equation is an affine form or a degree $2$ polynomial with just two monomials, such that a solution of $T$ extends to a solution of $\tilde{T}$ and a solution of $\tilde{T}$ implies a solution of $T$. In the second stage, a polynomial $P$ is constructed from the polynomials in $\tilde{T}$ by introducing a set of variables disjoint from that of $\tilde{T}$ and multiplying each polynomial in $\tilde{T}$ with a distinct variable from the set and adding all of them. The polynomial $P$ is of degree $3$ and can be represented by a depth-$4$ circuit. By analysing the sparsity of $P$ under translations, they show that $\tilde{T}$ (and thus $T$) has a solution if and only if the sparsity of $P$ reduces under some translation vector. The construction of $P$ and this analysis requires that $R$ contain an element with no multiplicative inverse. Hence, the reduction does not hold as is over fields. Over $\Z$, the undecidability of solving Diophantine equations \cite{Dav73,Mat70} further implies the undecidability of $s\text{-to-}(s-1)\sparseshift$ over $\Z$. 

In contrast, we reduce from $\SAT$ and show the $\NP$-hardness of $\sparseshift$ over any integral domain \emph{including} fields when the input polynomial in $\sparseshift$ is given in the sparse representation.

\paragraph{Result on $\gapsparseshift$.} The authors of \cite{ChillaraGS23} showed that for any constant $\alpha > 1$, $\gapsparseshift$ is $\NP$-hard over $\R$,$\Q,\F_p$ and $\Z_q$ when the input polynomial is given in the sparse representation and that there exists a constant $\beta > 1$ such that $\beta^{d}\text{-}\mathrm{gap}\text{-}\sparseshift$ is $\NP$-hard over the mentioned algebraic structures when the input polynomial is in $nd$ variables, of degree $d$ with $d \leq n^{O(1)}$ and is  given as a circuit. Note that this also shows $\sparseshift$ is $\NP$-hard over the mentioned algebraic structures. They also showed that $s^{o(1)}\text{-}\mathrm{gap}\text{-}\sparseshift$, where $s$ is the sparsity of the input polynomial given in the sparse representation, is undecidable over $\Z$.

To prove the undecidability of $s^{o(1)}\text{-}\mathrm{gap}\text{-}\sparseshift$ over $\Z$, they show a reduction from polynomial solvability by using the polynomial $P$, with $\sigma := \cal S(P)$, output by the reduction to $\sigma\text{-to-}(\sigma-1)\sparseshift$. They observed that $\sigma\text{-to-}(\sigma-1)\sparseshift$ can be viewed as $\frac{\sigma}{\sigma - 1}\text{-}\mathrm{gap}\text{-}\sparseshift$ where the inputs are a polynomial of sparsity $\sigma$ and an integer $\sigma - 1$. By multiplying $d$ variable-disjoint copies of $P$, where $d$ is a positive integer, they get a polynomial $Q$ with $s := \cal S(Q) = \sigma^d$. By analyzing the sparsity of $Q$ under translations, they show a gap of $(\frac{\sigma}{\sigma-1})^{d} \approx e^{\frac{d}{\sigma-1}} = s^{\frac{1}{(\sigma-1)\log(\sigma)}} = s^{o(1)}$. Thus, they amplify the gap in $\frac{\sigma}{\sigma - 1}\text{-}\mathrm{gap}\text{-}\sparseshift$ to show polynomial solvability reduces to $s^{o(1)}\text{-}\mathrm{gap}\text{-}\sparseshift$, where $s$ is the sparsity of the input polynomial, over any integral domain that is \emph{not} a field. Then, the undecidability of $s^{o(1)}\text{-}\mathrm{gap}\text{-}\sparseshift$ over $\Z$ follows from that of solving Diophantine equations over $\Z$. 

To show the $\NP$-hardness of $\gapsparseshift$ for all constants $\alpha > 1$ over $\R$,$\Q, \F_p$ and $\Z_q$, they first show a reduction from the $(\epsilon,\delta)\text{-}\mathrm{Max}$-$3\mathrm{Lin}$ problem\footnote{In the $(\epsilon,\delta)\text{-}\mathrm{Max}$-$3\mathrm{Lin}$ problem, where $1 > \epsilon > \delta > 0$, a system of linear equations over the underlying algebraic structure is given as input, where each equation is dependent on exactly $3$ variables, and we need to decide if there exists an assignment to the variables such that at least $\epsilon$ fraction of the equations are satisfiable or if for all assignments at most $\delta$ fraction are satisfiable.}, which is $\NP$-hard over the specified algebraic structures for appropriate $\epsilon$ and $\delta$ by the PCP theorem \cite{Has01,FGKP06,GR06,GR07}, to ${\beta\text{-}\mathrm{gap}\text{-}}\sparseshift$ for a small constant $\beta > 1$ by constructing a polynomial $P$ of degree $2$ from the system of linear equations given as input in $(\epsilon,\delta)\text{-}\mathrm{Max}$-$3\mathrm{Lin}$. Let $n$ be the number of variables in $P$, then $P$ can also be represented as a depth-$2$ circuit of size $n^{O(1)}$. They then multiply $d$ variable-disjoint copies of $P$, where $d$ is a parameter, to get an $nd$-variate polynomial $Q$, which has sparsity $\cal S(P)^d$ if $P$ is given in the sparse representation and a circuit of size $Nd+1$ if $P$ is given as a circuit of size $N$. In particular, $Q$ has a depth-$3$ ($\Pi\Sigma\Pi$) circuit of size $n^{O(1)}d$.  By analyzing the sparsity of $Q$ under translations, they amplify the gap $\beta$ to $\beta^d$ and show the $\NP$-hardness of ${\beta^d\text{-}\mathrm{gap}\text{-}}\sparseshift$. For any constant $\alpha > 1$, taking $d = \log_{\beta}(\alpha)$ shows the $\NP$-hardness of $\gapsparseshift$. Their reduction runs in polynomial time as long as $d = O(1)$ when $P$ is given in the sparse representation, and $d \leq N^{O(1)}$ when $P$ is given as a circuit of size $N$. Thus, ${\beta^d\text{-}\mathrm{gap}\text{-}}\sparseshift$ is $\NP$-hard for $d \leq n^{O(1)}$ as $P$ can be given as a depth-$2$ circuit. Note that the reduction from $(\epsilon,\delta)\text{-}\mathrm{Max}$-$3\mathrm{Lin}$ does \emph{not} imply the $\NP$-hardness of $s\text{-to-}(s-1)\sparseshift$ over the specified algebraic structures because for any translation vector $\vecb$, the lower bound on $\cal S(P(\vecx + \vecb))$ is strictly less than $\cal S(P(\vecx))$ and the same holds for $Q$. 

In contrast, we reduce from $\SAT$ and show that ${s^{1/3 -\epsilon}\text{-}\mathrm{gap}\text{-}}\sparseshift$, where $s$ is the sparsity of the input polynomial given in the sparse representation, is $\NP$-hard over any integral domain \emph{including} fields. We establish the gap, which is \emph{super-constant}, by a slightly deeper analysis of the sparsity of the polynomial $f$, constructed in Section \ref{section-sparseshiftNPH}, under translations and comparing the sparsities corresponding to satisfiable and unsatisfiable $\SAT$ formulas. Our reduction does \emph{not} use gap amplification or invoke the PCP theorem. 

%% file: Conclusion.tex
\section{Conclusion} \label{Section: Conclusion}
In this work, we show that ET for sparse polynomials is $\NP$-hard. Particularly, we show the $\NP$-hardness of MCSP for orbits of homogeneous sparse polynomials (a dense subclass of $\hdthree$ circuits). We also define a gap version of ET for sparse polynomials and show it is $\NP$-hard, which implies the $\NP$-hardness of $s^{\frac{1}{3} -\epsilon}$-factor approximation of the sparse-orbit complexity of $s$-sparse polynomials. We also show that ET for constant-support polynomials is $\NP$-hard. Lastly, we show that SET for sparse polynomials is also $\NP$-hard along with the hardness of a gap version of the problem. In all four cases, we reduce $\SAT$ to the respective problems.  We end by listing some problems whose solutions we do not know:

\begin{enumerate}
    \item \textbf{Hardness of $\et$ and $\sparseshift$ for constant degree polynomials:} In the reduction of Theorems \ref{Theorem: ETSparse NP-hard} and \ref{Theorem: ETSparse_shift NP-hard}, can the degree of the output polynomial be made constant? Currently, the degree is polynomial in the number of clauses and variables.

    \item \textbf{Improving the gap in Theorem \ref{Theorem: gap ETSparse NP-hard} and part two of Theorem \ref{Theorem: ETSparse_shift NP-hard}:} Can $\gapet$ and $\gapsparseshift$ be shown $\NP$-hard for $\alpha = s^{1-\epsilon}$, where $s$ is the sparsity of the input polynomial and $\epsilon > 0$ is an arbitrary constant? 

    \item \textbf{Hardness of $\etcsp$ for $\sigma=2$}: Is checking if a given polynomial is in the orbit of a support-$2$ polynomial $\NP$-hard? Theorem \ref{Theorem: ETsupport NP-hard} shows that $\etcsp$ for $\sigma \geq 5$ is $\NP$-hard. 
    \item \textbf{Hardness of MCSP for $\hdthree$ circuits:} Is MCSP for $\hdthree$ circuits $\NP$-hard? 

\end{enumerate}

\section*{Acknowledgments}
\addcontentsline{toc}{section}{\protect\numberline{}Acknowledgments}
We thank the anonymous reviewers for their detailed and constructive feedback, which has helped us improve the presentation of this work. In particular, we thank one of the reviewers for pointing out some inaccuracies in the original proofs of Lemmas \ref{P1Sep} and \ref{P2Sep}; simpler proofs for both lemmas came up in the process of fixing these inaccuracies.

%% file: Appendix.tex
\appendix
\addtocontents{toc}{\protect\setcounter{tocdepth}{1}}
 
\section{Handling translations in Theorems \ref{Theorem: ETSparse NP-hard} and \ref{Theorem: gap ETSparse NP-hard}} \label{Section: Translations} \subsection{Extending Theorem \ref{Theorem: ETSparse NP-hard} for translations}\label{Section: Affine ETsparse}
In this section, we modify the construction of Section \ref{subsubsection: sparse f construction} to prove part $1$ of Theorem \ref{Theorem: ETSparse NP-hard} while considering translations. The idea is to choose the degree parameters for the polynomial $f$ such that if $f(A\vecz + \vecb)$ is $s$-sparse, where $A \in \GL(|\vecz|,\F)$ and $\vecb \in \F^{|\vecz|}$, then $\vecb$ must be $\mathbf{0}$ (the all $0$s vector in $\F^{|\vecz|}$) and $A$ must be as described in \eqref{action}. We first show the reduction over any field of characteristic not equal to $2$ and give a separate construction for characteristic $2$ fields.

Formally, let $\psi$, $\vecx, x_0, \vecy$ and $\vecz$ be as denoted in Section \ref{subsubsection: sparse f construction}. For $\vecb \in \F^{|\vecz|}$, $\vecb_{|w}$ denotes the component of $\vecb$ corresponding to the variable $w \in \vecz$.  Let $d_1$, $d_2$, $d_3$, $d_4 \in \N$. Consider the following polynomials:

\begin{itemize}
    \item Corresponding to $x_i$, where $i\in [n]$, define $Q_i(\vecz)$ as: 
\begin{equation*} 
\begin{split}
    Q_i(\vecz) &:= Q_{i,1}(\vecz) + Q_{i,2}(\vecz) + Q_{i,3}(\vecz), \  Q_{i,1}(\vecz) := x_0^{(3i-2)(d_2+1)}x_i^{d_2}, \\ Q_{i,2}(\vecz) &:= x_0^{(3i-1)(d_2+1)}(y_i+x_i)^{d_3} \text{ and } Q_{i,3}(\vecz) := x_0^{3i(d_2+1)}(y_i-x_i)^{d_3}. 
\end{split}
\end{equation*}

\item For the $k\ith$ clause, $k \in [m]$,  define $R_k(\vecz) := x_0^{k(3d_4+1)}\prod_{j\in C_k}(y_j+ (-1)^{a_{k,j}}x_j)^{d_4}.$ 
\end{itemize}
Define $s:= 1 + n(d_3+3) + m(d_4+1)^2$. Impose the following conditions on the $d_i$'s:
\begin{equation} \label{Affine ETsparse conditions}
    d_1 \geq 6n(d_2 + 1) + 2d_3 + 2, \
    d_2 \geq 2d_3, \
    d_3 \geq m(d_4+1)^2 + 1, \
    d_4 \geq m.
\end{equation}
Finally, define $f(\vecz)$ as:
\begin{equation} \label{affine variable poly}
  f(\vecz) := x_0^{d_1} + \sum_{i=1}^{n} Q_i(\vecz) + \sum_{k=1}^{m} R_k(\vecz).  
\end{equation}

The $d_i$'s are chosen in Section \ref{Choice of degrees} such that they are $(mn)^{O(1)}$ and also satisfy the inequalities of \eqref{Affine ETsparse conditions} over any field. Observations \ref{Obs: deg separated polys translated}, \ref{Obs: translated poly degree} and \ref{Obs: translated poly sparsity} hold under the conditions of \eqref{Affine ETsparse conditions}.
\begin{observation} \label{Obs: deg separated polys translated}
For all $i \in [n]$, $k \in [m]$, the polynomials $x_0^{d_1}$, $Q_{i,1}(\vecz)$, $Q_{i,2}(\vecz)$, $Q_{i,3}(\vecz)$ and $R_k(\vecz)$ are degree separated from one another. Also, $Q_i(\vecz)$ is degree separated from other $Q_j(\vecz)$'s, for $i,j \in [n]$ and $i \neq j$. Similarly, $R_k(\vecz)$ is degree separated from $R_l(\vecz)$ for $k,l \in [m]$ and $k \neq l$.
\end{observation}
\begin{proof}
    Let $i \in [n]$. Note that $Q_{i,1}(\vecz)$ has degree $(3i-2)(d_2+1) + d_2$, $Q_{i,2}(\vecz)$ has degree $(3i-1)(d_2+1) + d_3$ and $Q_{i,3}(\vecz)$ has degree $3i(d_2+1) + d_3$. Clearly, 
    \[3i(d_2+1) + d_3 > (3i-1)(d_2+1) + d_3 > (3i-2)(d_2+1) + d_2.\] 
    Thus, $Q_{i,1}(\vecz)$, $Q_{i,2}(\vecz)$ and $Q_{i,3}(\vecz)$ are degree separated and $Q_i(\vecz)$ is a sum of three degree separated polynomials and has degree $3i(d_2+1) + d_3$.

    Now, let $i \in [n]$ and $k \in [m]$. The lowest degree of any monomial of $Q_i(\vecz)$ is $(3i-2)(d_2+1) + d_2 > 2d_2$, while the degree of any monomial of $R_k(\vecz)$ is $k(3d_4+1) + 3d_4 \leq m(3d_4+1) + 3d_4$. As $d_2 \geq 2d_3$ and $d_3 > m(d_4+1)^2$ from \eqref{Affine ETsparse conditions}, therefore
    \begin{equation*}
        \begin{split}
           2d_2 \geq 4d_3 > 4m(d_4+1)^2 > m(3d_4+1) + 3d_4
        \end{split}
    \end{equation*}
    Thus $Q_i(\vecz)$ is degree separated from $R_k(\vecz)$.
            
    Lastly, let $i,j \in [n]$ where $i < j$ without loss of generality. The highest degree of any monomial of $Q_i(\vecz)$ is $3i(d_2+1) + d_3$ while the lowest degree of any monomial of $Q_j(\vecz)$ is $(3j-2)(d_2+1) + d_2$. Note that
    \[(3j-2)(d_2+1) + d_2 \geq (3i+1)(d_2+1) + d_2 > 3i(d_2+1) + d_3 \]
    because $j \geq i+1$ and $d_2 > d_3$ from \eqref{Affine ETsparse conditions}. Therefore $Q_i(\vecz)$ and $Q_j(\vecz)$ are degree separated. That $R_k(\vecz)$ is degree separated from $R_l(\vecz)$ for $k,l \in [m]$ and $k \neq l$ can be observed from the fact that the degree of $R_k(\vecz)$ is $k(3d_4+1) + 3d_4$. Clearly, $x_0^{d_1}$ is degree separated from the rest of the polynomials because $d_1 \geq 6n(d_2+1) + 2d_3 + 2$ while the highest degree polynomial among $Q_i(\vecz)$'s and $R_k(\vecz)$'s is $Q_n(\vecz)$ of degree $3n(d_2+1) + d_3$. 

\end{proof}
\begin{observation} \label{Obs: translated poly degree}
The degree of $f$ is $d_1 = (mn)^{O(1)}$ with $\frac{d_1}{2} > s$.
\end{observation}
\begin{proof}
By the definition of $f$ in \eqref{affine variable poly} and Observation \ref{Obs: deg separated polys translated}, the degree of $f$ is the maximum degree among $x_0^{d_1}$, $Q_i$'s and $R_k$'s, where $i \in [n]$ and $k \in [m]$. The degree of $Q_i$ is $3i(d_2+1) + d_3$ and that of $R_k$ is $k(3d_4+1) + 3d_4$ with $3i(d_2+1) + d_3  > k(3d_4+1) + 3d_4$. Further, $d_1 \geq 6nd_2 + 6n + 2d_3 + 2$. Hence the degree of $f$ is $d_1$. Finally, it holds by \eqref{Affine ETsparse conditions} that $\frac{d_1}{2} > 3nd_2 > 3nd_3 > s$.
\end{proof}
\begin{observation} \label{Obs: translated poly sparsity}
    $\cal S(f(\vecz)) = 1 + n(2d_3+3) + m(d_4+1)^3$ and $\supp(f) = 7$.
\end{observation}
\noindent The proof of Observation \ref{Obs: translated poly sparsity} is similar to that of Observation \ref{Obs: ETsparse poly sparsity } and uses Observations \ref{Obs: deg separated polys translated}, \ref{lem_deg_sep_sum} and \ref{lem_binom}.

\paragraph{The forward direction.} If $\vecu \in \{0,1\}^n$ is such that $\psi(\vecu) = 1$ and $f$, as described in \eqref{affine variable poly}, is the polynomial corresponding to $\psi$, then Proposition \ref{Proposition: fwd direction}, with some small changes to its statement and proof, shows that for $\vecb = \mathbf{0}$ and the transform $A$ as described in \eqref{action}, $\cal S(f(A\vecz)) \leq s$ holds. 

\paragraph{The reverse direction.} We leverage the constraints in \eqref{Affine ETsparse conditions} to show that if $\cal S(f(A\vecz + \vecb)) \leq s$ for some $A \in \GL(|\vecz|,\F)$ and $\vecb \in \F^{|\vecz|}$, then $A$ is as described in \eqref{action} and $\vecb = \mathbf{0}$. Lemma \ref{Lemma: affine x0 fixed} shows that $A(x_0) = x_0$, without loss of generality, and $\vecb_{|x_0} = 0$. Then, Lemma \ref{Lemma: degree separation affine case} shows that the summands of $f(A\vecz + \vecb)$ are degree separated with respect to $x_0$. 

\begin{lemma} \label{Lemma: affine x0 fixed}
    $A(x_0) = x_0$, without loss of generality, and $\vecb_{|x_0} = 0$.
\end{lemma}
\begin{proof}
    Let $\ell_0 = A(x_0)$ and $b_0 = \vecb_{|x_0}$. If $b_0 = 0$ and $\ell_0$ is a linear form in at least two variables, then by the choice of $d_1$ and Observation \ref{lem_binom}, $\cal S((\ell_0 + b_0)^{d_1}) \geq d_1 + 1 > s$. If $b_0 \neq 0$, then by the binomial theorem
    \[(\ell_0 + b_0)^{d_1}= \sum_{i=0}^{d_1} \binom{d_1}{i} b_0^i \ell_0^{d_1-i}.\]
    The summands in the above expansion are degree separated and $\binom{d_1}{i} \neq 0$ for all $i \in [0,d_1]$ because of the choice of $d_1$ (and Lucas's Theorem if the characteristic is finite). Thus, $(\ell_0 + b_0)^{d_1}$ contains at least one monomial of degree $i$, for all $i \in [0,d_1]$. Since $\frac{d_1}{2} > 3n(d_2+1) + d_3 > s$ and the maximum degree of any polynomial among the $Q_i$'s and $R_k$'s, with $i \in [n]$ and $k \in [m]$, is $3n(d_2+1)+d_3$, therefore at least $s+1$ monomials in $(\ell_0 + b_0)^{d_1}$ are degree separated from $Q_i(A\vecz+\vecb)$'s and $R_k(A\vecz+\vecb)$'s. Hence $\cal S(f(A\vecz + \vecb))  > s$, a contradiction. Thus, $\vecb_{|x_0} = b_0$ must be $0$ and $A(x_0) = \ell_0$ must have only one variable, which can be assumed to be $x_0$ without loss of generality, as permutation and non-zero scaling of variables does not affect the sparsity.
\end{proof}
\begin{lemma} \label{Lemma: degree separation affine case}
    For all $i \in [n]$, $k \in [m]$, $x_0^{d_1}$, $Q_{i,1}(A\vecz + \vecb)$, $Q_{i,2}(A\vecz + \vecb)$, $Q_{i,3}(A\vecz + \vecb)$ and $R_j(A\vecz + \vecb)$ are degree separated from one another with respect to $x_0$. Also, $Q_i(A\vecz + \vecb)$ is degree separated with respect to $x_0$ from other $Q_j(A\vecz + \vecb)$'s, for $i,j \in [n]$ and $i \neq j$. Similarly, $R_k(A\vecz + \vecb)$ is degree separated with respect to $x_0$ from $R_l(A\vecz + \vecb)$ for $k,l \in [m]$ and $k \neq l$.
\end{lemma}
\begin{proof}
    Let $i \in [n]$. For $Q_{i,1}(A\vecz + \vecb)$, $Q_{i,2}(A\vecz + \vecb)$ and $Q_{i,3}(A\vecz + \vecb)$, the respective range of $x_0$-degree of a monomial of respective polynomials is $[(3i-2)(d_2+1),(3i-2)(d_2+1) + d_2]$, $[(3i-1)(d_2+1),(3i-1)(d_2+1) + d_3]$ and $[3i(d_2+1),3i(d_2+1) + d_3]$. As $d_2 > d_3$, it can be observed that these ranges are disjoint, implying
    $Q_{i,1}(A\vecz + \vecb),Q_{i,2}(A\vecz + \vecb)$ and $Q_{i,3}(A\vecz + \vecb)$ are degree separated from one another with respect to $x_0$.

    Let $i,j \in [n]$ and $i < j$ without loss of generality. For $Q_i(A\vecz + \vecb)$ and $Q_j(A\vecz + \vecb)$, the respective range of $x_0$-degree of a monomial of respective polynomials is $[(3i-2)(d_2+1),3i(d_2+1) + d_3]$ and $[(3j-2)(d_2+1),3j(d_2+1) + d_3]$. As $d_2 > d_3$ and $j \geq i+1$, therefore $(3j-2)(d_2+1) > 3i(d_2+1) + d_3$. Hence, $Q_i(A\vecz + \vecb)$ is degree separated from $Q_j(A\vecz + \vecb)$ with respect to $x_0$. 
    
    Now, let $k, l \in [m]$ and $k < l$ without loss of generality. For $R_k(A\vecz + \vecb)$ and $R_l(A\vecz + \vecb)$, the respective range of $x_0$-degree of a monomial of respective polynomials is $[k(3d_4+1),k(3d_4+1) + 3d_4]$ and $[l(3d_4+1),l(3d_4+1) + 3d_4]$. As $l \geq k+1$, therefore $l(3d_4+1) > k(3d_4+1) + 3d_4$. Hence, $R_k(A\vecz + \vecb)$ is degree separated from $R_l(A\vecz + \vecb)$ with respect to $x_0$.

    Lastly, let $i \in [n]$ and $k \in [m]$. The highest $x_0$-degree of a monomial in $R_k(A\vecz+\vecb)$ is $k(3d_4+1) + 3d_4 \leq m(3d_4+1) + 3d_4$, while the lowest $x_0$-degree of any monomial in $Q_i(A\vecz+\vecb)$ is $(3i-2)(d_2+1) > d_2$. Now,
    \[ d_2 \geq 2d_3 > 2m(d_4+1)^2 > m(3d_4+1) + 3d_4 \]
    where the inequalities follow from the conditions in \eqref{Affine ETsparse conditions}. Therefore $Q_i(A\vecz+\vecb)$ is degree separated from $R_k(A\vecz+\vecb)$'s with respect to $x_0$. Clearly, $x_0^{d_1}$ is degree separated with respect to $x_0$ from $Q_i(A\vecz+\vecb)$ and $R_k(A\vecz+\vecb)$ because $d_1 > 3n(d_2 + 1) + d_3$, the highest $x_0$-degree of any monomial among the $Q_i(A\vecz+\vecb)'s$ and $R_k(A\vecz+\vecb)$'s.
\end{proof}
\vspace{-5mm}
\[ \therefore  \cal S(f(A\vecz + \vecb)) = \cal S(x_0^{d_1}) + \sum_{i=1}^{n} \cal S(Q_i(A\vecz + \vecb)) + \sum_{k=1}^{m} \cal S(R_k(A\vecz + \vecb)) \text{ by Lemma \ref{Lemma: degree separation affine case}.}
\]
Lemma \ref{Qi sparsity analysis: translated} analyses the sparsity of $Q_i(A\vecz+\vecb)$. The proof of Lemma \ref{Qi sparsity analysis: translated} is similar to that of Lemma \ref{var_spars_lem} and uses Observations \ref{lem_binom} and \ref{Lemma: affine form sparsity}. 
\begin{lemma} \label{Qi sparsity analysis: translated}
    For any invertible $A$, $\vecb \in \F^{|\vecz|}$ and $i \in [n]$:
        \begin{equation*}
            \cal S(Q_i(A\vecz+\vecb)) = \cal S(Q_{i,1}(A\vecz+\vecb))+\cal S(Q_{i,2}(A\vecz+\vecb))+\cal S(Q_{i,3}(A\vecz+\vecb))\geq d_3 + 3,
        \end{equation*}
    where $Q_i$, $Q_{i,1}$, $Q_{i,2}$ and $Q_{i,3}$ are as defined in this section. Equality holds if and only if $\vecb_{|x_i} = 0, \vecb_{|y_i} = 0$ and under $A$
        \[x_i \mapsto X_i  \text{ and } y_i \mapsto Y_i + (-1)^{u_i}X_i \] 
    for some scaled $X_i, Y_i \in \vecz$ and $u_i\in \{0,1\}$. Further, if $\cal S(Q_i(A\vecz+\vecb)) \neq d_3 + 3$, then $\cal S(Q_i(A\vecz+\vecb)) \geq 2d_3 + 3$.
\end{lemma}
\begin{proof}
    From Lemma \ref{Lemma: degree separation affine case} and Observation \ref{lem_deg_sep_sum}, it follows that $\cal S(Q_i(A\vecz+\vecb)) = \cal S(Q_{i,1}(A\vecz+\vecb))+\cal S(Q_{i,2}(A\vecz+\vecb))+\cal S(Q_{i,3}(A\vecz+\vecb))$. The if direction of the lemma statement is easy to verify. For the only if direction consider the following cases of $A$ and $\vecb$:
\begin{enumerate}
    \item $\cal S(A(x_i) + \vecb_{|x_i}) \geq 2$: If $\vecb_{|x_i} = 0$  then $\cal S(A(x_i)) \geq 2$ and by Observation \ref{lem_binom}, $\cal S(Q_{i,1}(A\vecz + \vecb)) \geq d_2 + 1$. If $\vecb_{|x_i} \neq 0$ then by the choice of $d_2$ and Observation \ref{Lemma: affine form sparsity}, $\cal S(Q_{i,1}(A\vecz + \vecb)) \geq d_2 + 1$. Also, $\cal S(Q_{i,2}(A\vecz + \vecb)) \geq 1$ and $\cal S(Q_{i,3}(A\vecz + \vecb)) \geq 1$. Hence, $\cal S(Q_{i}(A\vecz + \vecb)) \geq d_2+3 \geq 2d_3 + 3$ as $d_2 \geq 2d_3$. For the remaining cases, we consider $\cal S(A(x_i) + \vecb_{|x_i}) = 1$, meaning $A(x_i) = X_i$ for some scaled variable $X_i \in \vecz$ and $\vecb_{|x_i} = 0$. 
    
    \item $\cal S(A(x_i) + \vecb_{|x_i}) = 1$, $\cal S(A(y_i + x_i)+ \vecb_{|y_i} + \vecb_{|x_i})\geq 2$ and $\cal S(A(y_i -x_i) + \vecb_{|y_i} - \vecb_{|x_i})\geq 2$: Like in the previous case, it can be shown using Observation \ref{lem_binom} (if $\vecb_{|y_i} = 0$) or Observation \ref{Lemma: affine form sparsity} (if $\vecb_{|y_i} \neq 0$) that $\cal S(Q_{i,2}(A\vecz + \vecb)) \geq d_3+1$ and $\cal S(Q_{i,3}(A\vecz + \vecb)) \geq d_3+1$ implying $\cal S(Q_{i}(A\vecz + \vecb)) \geq 2d_3 + 3$.
    
    \item $\cal S(A(x_i) + \vecb_{|x_i}) = 1$ with $\cal S(A(y_i + x_i)+ \vecb_{|y_i} + \vecb_{|x_i}) = 1$ or $\cal S(A(y_i -x_i) + \vecb_{|y_i} - \vecb_{|x_i}) = 1$: Because $A$ is invertible, $\cal S(A(y_i-x_i)) \geq 1$ and $\cal S(A(y_i+x_i)) \geq 1$. Further, since $\vecb_{|x_i} = 0$, therefore $\vecb_{|y_i}$ must be $0$. This observation and the invertibility of $A$ imply that exactly one of $\cal S(A(y_i + x_i)) = 1$ or $\cal S(A(y_i - x_i)) = 1$ holds. Without loss of generality, let $\cal S(A(y_i + x_i)) = 1$, which implies $A(y_i) = Y_i - X_i$ for some scaled variable $Y_i \in \vecz$. Then $A(y_i - x_i) = Y_i - 2X_i$. Hence, $\cal S(Q_{i,1}(A\vecz + \vecb)) = 1$, $\cal S(Q_{i,2}(A\vecz + \vecb)) = 1$ and $\cal S(Q_{i,3}(A\vecz + \vecb)) = d_3 + 1$ (by Observation \ref{lem_binom}) implying $\cal S(Q_{i}(A\vecz + \vecb)) = d_3 + 3$.
\end{enumerate}
The first two cases show that if $A$ and $\vecb$ are not as per the lemma statement, then $\cal S(Q_i(A\vecz + \vecb)) \geq 2d_3 + 3$; otherwise, $\cal S(Q_{i}(A\vecz + \vecb)) = d_3 + 3$.
\end{proof}
\noindent Lemma \ref{Lemma: equality condition Qi} holds with some change to its statement and can be proved using Lemmas \ref{Lemma: degree separation affine case} and \ref{Qi sparsity analysis: translated}, and Observation \ref{lem_deg_sep_sum}. Lemmas \ref{Lemma: affine x0 fixed}, \ref{Qi sparsity analysis: translated} and \ref{Lemma: equality condition Qi} together show that $A$ is a permuted scaled version of the transform of \eqref{action} and that $\vecb = \mathbf{0}$. Then, Proposition \ref{Proposition: bwd direction} shows how a satisfying assignment for $\psi$ can be recovered from $A$.

\subsubsection*{\underline{Construction for characteristic $2$ fields}} \label{Section: construction affine char 2}
Since over characteristic $2$ fields $y_i + x_i$ and $y_i-x_i$ are the same polynomials, we need to modify the previous construction. Moreover, the sparsifying transform will also be slightly different. Formally, let $\psi$, $\vecx, x_0, \vecy$ and $\vecz$ be as denoted in Section \ref{subsubsection: sparse f construction}. Let $d_1$, $d_2$, $d_3$, $d_4 \in \N$. Consider the following polynomials:
\begin{itemize}
    \item Corresponding to $x_i$, where $i\in [n]$, define $Q_i(\vecz)$ as: 
\begin{equation*} 
\begin{split}
    Q_i(\vecz) &:= Q_{i,1}(\vecz) + Q_{i,2}(\vecz) + Q_{i,3}(\vecz), \  Q_{i,1}(\vecz) := x_0^{(3i-2)(d_2+1)}x_i^{d_2}, \\ Q_{i,2}(\vecz) &:= x_0^{(3i-1)(d_2+1)}(y_i+x_i)^{d_3} \text{ and } Q_{i,3}(\vecz) := x_0^{3i(d_2+1)}y_i^{d_3}. 
\end{split}
\end{equation*}

\item For the $k\ith$ clause, $k \in [m]$,  define $R_k(\vecz) := x_0^{k(3d_4+1)}\prod_{j\in C_k}(y_j + a_{k,j}x_j)^{d_4}.$ 
\end{itemize}
Define $s:= 1 + n(d_3+3) + m(d_4+1)^2$. Choose $d_i$'s as specified in Section \ref{Choice of degrees} so that they are $(mn)^{O(1)}$ and satisfy the conditions of \eqref{Affine ETsparse conditions}. Finally, define $f(\vecz)$ as:
\begin{equation} \label{affine variable poly char 2}
  f(\vecz) := x_0^{d_1} + \sum_{i=1}^{n} Q_i(\vecz) + \sum_{k=1}^{m} R_k(\vecz).  
\end{equation}
Observations \ref{Obs: deg separated polys translated}, \ref{Obs: translated poly degree} hold with little change. Observation \ref{Obs: translated poly sparsity char 2} analyses the sparsity and support of $f$ and has a proof similar to that of Observations \ref{Obs: translated poly sparsity} and Observation \ref{Obs: poly sparsity char 2 case}.

\begin{observation} \label{Obs: translated poly sparsity char 2}
    $\cal S(f(\vecz)) \leq 1 + n(d_3+3) + m(d_4+1)^3$ and $4 \leq \supp(f) \leq 7$.
\end{observation}

\begin{example} \label{Remark: char two not natural affine}
 Like in Section \ref{section-homogeneous-char2}, the sparsity of the polynomial output by the reduction over characteristic $2$ fields depends on the number of variables which are complemented within a clause. Hence, for the same number of variables $n$ and the same number of clauses $m$, the output polynomial corresponding to two different $\psi$'s may have different sparsity. Thus, the reduction is not natural over characteristic $2$ fields.
\end{example}
\paragraph{The forward direction.} Let $\vecu \in \{0,1\}^n$ be such that $\psi(\vecu) = 1$ and $f$, as described in \eqref{affine variable poly char 2}, be the polynomial corresponding to $\psi$. Proposition \ref{Proposition: fwd direction char 2} shows how $\vecu$ can be used to construct a sparsifying transform $A$ with $\vecb = \mathbf{0}$.

\paragraph{The reverse direction.} Let $A \in \GL(|\vecz|,\F)$ and $\vecb \in \F^{|\vecz|}$ be such that $\cal S(f(A\vecz + \vecb)) \leq s$. The analysis of the reverse direction in the previous section holds with some changes. Formally, Lemma \ref{Lemma: affine x0 fixed} holds without any change in its proof while Lemma \ref{Lemma: degree separation affine case} holds with some change in its statement and proof. Thus, $A(x_0) = x_0$ without loss of generality. Lemma \ref{Qi sparsity analysis: affine char 2} analyses $\cal S(Q_i(A\vecz + \vecb))$, where $i \in [n]$, and its proof is similar to that of Lemma \ref{Qi sparsity analysis: translated}. 

\begin{lemma} \label{Qi sparsity analysis: affine char 2}
    For any $A \in \GL(|\vecz|,\F^{|\vecz|})$, $\vecb \in \F^{|\vecz|}$ and $i \in [n]$:
        \begin{equation*}
            \cal S(Q_i(A\vecz+\vecb)) = \cal S(Q_{i,1}(A\vecz+\vecb))+\cal S(Q_{i,2}(A\vecz+\vecb))+\cal S(Q_{i,3}(A\vecz+\vecb))\geq d_3 + 3,
        \end{equation*}
    where $Q_i$, $Q_{i,1}$, $Q_{i,2}$ and $Q_{i,3}$ are as defined in this section. Equality holds if and only if $\vecb_{|x_i} = 0, \vecb_{|y_i} = 0$ and under $A$
        \[x_i \mapsto X_i  \text{ and } y_i \mapsto Y_i + (1-u_i)X_i \] 
    for some scaled $X_i, Y_i \in \vecz$ and $u_i\in \{0,1\}$. Further, if $\cal S(Q_i(A\vecz+\vecb)) \neq d_3 + 3$, then $\cal S(Q_i(A\vecz+\vecb)) \geq 2d_3 + 3$.
\end{lemma}
\noindent Lemma \ref{Lemma: equality condition Qi} holds with the same proof (for translations) as before. Lemmas \ref{Lemma: affine x0 fixed}, \ref{Qi sparsity analysis: affine char 2} and \ref{Lemma: equality condition Qi} together show that $A$ is a permuted scaled version of the transform of \eqref{action char 2} and that $\vecb = \mathbf{0}$. Then, Proposition \ref{Proposition: bwd direction char 2} shows how a satisfying assignment for $\psi$ can be recovered from $A$.

\subsubsection{Setting of parameters} \label{Choice of degrees}
\paragraph{For characteristic $0$ fields.} In this case, the inequalities in \eqref{Affine ETsparse conditions} can be converted to equalities. Thus
\begin{equation*}
    \begin{split}
        d_4 &= m, \ d_3 = m(m+1)^2+1=O(m^3)  \implies s = O(nm^3), \\
        d_2 &= 2m(m+1)^2+2=O(m^3), \ d_1 = 6nd_2 + 6n + 2d_3 + 2 = O(nm^3).
    \end{split}
\end{equation*}
\paragraph{For finite characteristic fields.} Let the characteristic be $p > 0$. If $p > d_1$, where the value of $d_1$ is as in the characteristic $0$ fields case, then the $d_i$'s are set as per the characteristic $0$ fields case. Otherwise  $p = O(nm^3)$. In such a case, we choose the $d_1 > d_2 > d_3 > d_4$ to be of form $p^k - 1$, for some $k \in \N$, so that the conditions in \eqref{Affine ETsparse conditions} are satisfied and Observation \ref{lem_binom} can be used over characteristic $p$ fields. The lemmas and the observations in the previous section hold for this choice of the $d_i$'s. Now, the following bounds hold on the $d_i$'s: 
\begin{equation*}
    \begin{split}
        d_4 &\leq pm, \ d_3 \leq pm(d_4+1)^2+p = O(p^3m^3) \implies s = O(nm^3p^3), \\
        d_2 &= pd_3 + (p-1)= O(p^4m^3), \ d_1 \leq p(6nd_2+ 6n + 2d_3 + 2) = O(nm^3p^5).
    \end{split}
\end{equation*}
As $p = O(nm^3)$, therefore
\[        
d_1 = O(n^6m^{18}), \ d_2 = O(n^4m^{15}), \ d_3 = O(n^3m^{12}),\  d_4= O(nm^4) \text{ and }  s = O(n^4m^{12}).
\]
\subsection{Extending Theorem \ref{Theorem: gap ETSparse NP-hard} for translations} \label{Section: Affine gap result}
In this section, we prove part $1$ of Theorem \ref{Theorem: gap ETSparse NP-hard} over all fields while considering translations. The proof involves a careful analysis of the sparsity of $f$ as defined corresponding to a $\CNF$ $\psi$. We split the proof into two cases: over characteristic $0$ fields and finite characteristic fields. The reason for the split is that in the analysis, we consider the sparsity of powers of affine forms $h = \ell + c$, where $\cal S(h) \geq 3$. The sparsity of such affine forms depends on the underlying field, as shown in the following analysis using Observation \ref{Lemma: linear form sparsity} and the binomial theorem. Note that we take $\epsilon \in (0,1/3)$ to be an arbitrary constant throughout this section.

\subsubsection{For characteristic $0$ fields} \label{section-gap-affine}
Consider the polynomial $f$ as defined in \eqref{affine variable poly} for a $\CNF$ $\psi$ and choose the $d_i$'s to be $(mn)^{O(1)}$ and also satisfy the following conditions: 
\begin{equation}\label{Eqn: gap ineq affine}
    d_4 \geq \max(4mn,(mn)^{O(1/\epsilon)}), \  d_3 = m(d_4+1)^2 + 1, \ d_2 = d_3^2 + 1, \ d_1 = 6nd_2 + 6n + 2d_3 + 2.
\end{equation}
Note that the constraints in \eqref{Affine ETsparse conditions} are also satisfied under \eqref{Eqn: gap ineq affine}. Let $s: = \cal S(f)$. By Observation \ref{Obs: translated poly sparsity}, $s = 1 + n(2d_3+3)+m(d_4+1)^3$. From Section \ref{Section: Affine ETsparse}, it follows that for satisfiable $\psi$'s, there exists $A \in \GL(|\vecz|,\F)$, $\vecb = \mathbf{0}$ such that $\cal S(f(A\vecz)) \leq s_0$, where $s_0 = 1 + n(d_3 + 3) + m(d_4+1)^2$. For unsatisfiable $\psi$'s, Lemma \ref{Lemma: unsat gap affine} gives lower bounds on $\cal S(f(A\vecz+\vecb))$, where $A \in \GL(|\vecz|,\F)$ and $\vecb \in \F^{|\vecz|}$, and encapsulates the argument of the reverse direction of the reduction in Section \ref{Section: Affine ETsparse} with a slightly deeper analysis. Comparing the sparsities for satisfiable and unsatisfiable $\psi$'s proves part $1$ of Theorem \ref{Theorem: gap ETSparse NP-hard} for translations. Proposition \ref{Proposition: Sparsity gap affine} shows $\gapet$ is $\NP$-hard using Lemma \ref{Lemma: unsat gap affine} and the conditions in \eqref{Eqn: gap ineq affine}.

\begin{lemma} \label{Lemma: unsat gap affine}
   Let $\psi \in \bar{\SAT}$, $f$, as defined in \eqref{affine variable poly}, be the polynomial corresponding to $\psi$, $A \in \GL(|\vecz|,\F)$ and $\vecb \in \F^{|\vecz|}$.
   \begin{enumerate}
       \item \label{case 1 unsat gap affine} If $A(x_0) + \vecb_{|x_0}$ is a non-trivial affine form, then $\cal S(f(A\vecz + \vecb)) \geq \frac{d_1}{2}$.\footnote{an affine form $\ell + c$, where $\ell$ is a linear form and $c$ is a constant, is non-trivial if $\ell$ is a linear form in at least two variables or $c \neq 0$ or both.}

       \item \label{case 2 unsat gap affine} If $A$ and $\vecb$ are not as in item \ref{case 1 unsat gap affine} and for some $j \in [n]$, $A(x_j) + \vecb_{|x_j}$ is a non-trivial affine form, then $\cal S(f(A\vecz + \vecb)) \geq d_2+1$.

       \item If $A$ and $\vecb$ are not as in items \ref{case 1 unsat gap affine} and \ref{case 2 unsat gap affine} and for some $j \in [n]$, $\cal S(A(y_j + x_j) + \vecb_{|y_j}) \geq 3$ or $\cal S(A(y_j - x_j) + \vecb_{|y_j}) \geq 3$, then $\cal S(f(A\vecz + \vecb)) \geq \frac{(d_3 + 1)(d_3 + 2)}{2}$. 

       \item If $A$ and $\vecb$ are not of the form described in the previous three cases, then $\vecb = \mathbf{0}$ and  $\cal S(f(A\vecz)) \geq (d_4 + 1)^3.$
   \end{enumerate}
\end{lemma}
\begin{proof}
\begin{enumerate}
\item The proof of this case follows from the argument in the proof of Lemma \ref{Lemma: affine x0 fixed}. Henceforth, we assume $A(x_0) = x_0$ and $\vecb_{|x_0} = 0$. Then, by Lemma \ref{Lemma: degree separation affine case}, it follows that

\[\cal S(f(A\vecz + \vecb)) = \cal S(x_0^{d_1}) + \sum_{i=1}^{n} \cal S(Q_i(A\vecz + \vecb)) + \sum_{k=1}^{m} \cal S(R_k(A\vecz + \vecb)).\]
\item 
By Lemma \ref{Lemma: degree separation affine case} and Observation \ref{Lemma: affine form sparsity} (applied to $Q_{i,1}(A\vecz + \vecb)$) it follows that 
\[\cal S(f(A\vecz + \vecb)) \geq \cal S(Q_{i,1}(A\vecz + \vecb)) \geq d_2 + 1.\]
\item In this case, for $i \in [0,n]$, $A(x_i)$ is some scaled variable in $\vecz$ and $\vecb_{|x_i} = 0$. If $\vecb_{|y_j} = 0$, then this case is the same as the third case of Lemma \ref{Lemma: unsat gap}. Otherwise, let $A(y_j + x_j) + \vecb_{|y_j} = \ell_j + b_j$, where $\ell_j$ is a linear form in at least two variables. Then, using the binomial theorem and Observations \ref{lem_deg_sep_sum} and \ref{Lemma: linear form sparsity}, it follows that
\[\cal S((\ell_j + b_j)^{d_3}) = \cal S\bigg(\sum_{i=0}^{d_3}\binom{d_3}{i}b_j^{d_3-i}\ell_j^i\bigg) = \sum_{i=0}^{d_3} \cal S\bigg(\binom{d_3}{i}b_j^{d_3-i}\ell_j^i\bigg) \geq \sum_{i=0}^{d_3}(i+1) = \frac{(d_3 + 1)(d_3 + 2)}{2}.\]
Thus, \[\cal S(f(A\vecz + \vecb)) \geq \cal S(Q_{i,2}(A\vecz + \vecb)) \geq  \frac{(d_3 + 1)(d_3 + 2)}{2}.\]
Similarly, if $\cal S(A(y_j - x_j) + \vecb_{|y_j}) \geq 3$, then $\cal S(f(A\vecz + \vecb)) \geq \frac{(d_3 + 1)(d_3 + 2)}{2}$.
\item In this case, for $i \in [0,n]$, $A(x_i) = X_i$ and $\vecb_{|x_i} = 0$. For $i \in [n]$ , $\cal S(A(y_i + x_i) + \vecb_{|y_i}) \leq 2$ and $\cal S(A(y_i - x_i) + \vecb_{|y_i}) \leq 2$. As $A$ is invertible, $\cal S(A(y_i + x_i)) \geq 1$ and $\cal S(A(y_i - x_i)) \geq 1$. If $\cal S(A(y_i + x_i)) \geq 2$ or $\cal S(A(y_i - x_i)) \geq 2$, then $\vecb_{|y_i}$ must be $0$. Further, as $A$ is invertible and $A(x_i) = X_i$, therefore if $\cal S(A(y_i + x_i)) = 1$, then $\cal S(A(y_i - x_i)) \geq 2$ and vice versa also holds, which again implies $\vecb_{|y_i} = 0$. Thus, for all $i \in [n]$, $\vecb_{|y_i} = 0$ implying $\vecb = \mathbf{0}$. This case can then be proved in the same way as the last case of Lemma \ref{Lemma: unsat gap}.
\end{enumerate}
\end{proof}
\begin{proposition} \label{Proposition: Sparsity gap affine}
    Let $\textnormal{char}(\F) = 0$. If the input in $\gapet$ is an $s$-sparse polynomial, then $\gapet$ is $\NP$-hard for $\alpha = s^{1/3-\epsilon}$.
\end{proposition}
\begin{proof}
The proof is similar to that of Proposition \ref{Proposition: Sparsity gap}.  If $\psi \in \SAT$, then $\cal S(f(A\vecz)) \leq s_0$ where $A$ is as described in \eqref{action} and $s_0 = 1 + n(d_3+3) + m(d_4+1)^2$. If $\psi \in \bar{\SAT}$, it follows from Lemma \ref{Lemma: unsat gap affine} that for any $A \in \GL(|\vecz|,\F)$ and $\vecb \in \F^{|\vecz|}$: 
\[\cal S(f(A\vecz + \vecb)) \geq \min\bigg(\frac{d_1}{2},d_2+1,\frac{d_3^2+3d_3+2}{2},(d_4+1)^3\bigg).\]
The conditions imposed in \eqref{Eqn: gap ineq} ensure that $(d_4+1)^3 > s_0$ and $(d_4+1)^3$ is the minimum. As $d_3 = m(d_4+1)^2 + 1$, therefore $s_0 = 1 + n(d_3+3) + m(d_4+1)^2 \leq 3nd_3 \leq 4mn(d_4+1)^2$. Consequently, the gap in the sparsities of the YES instances and NO instances is 
\[  \frac{(d_4+1)^3}{s_0} \geq \frac{(d_4+1)^3}{4mn(d_4+1)^2} = \frac{d_4+1}{4mn}. \]
Also, note that as $d_4 \geq 4mn$, therefore $s \leq 2m(d_4+1)^3 \implies d_4+1 \geq (\frac{s}{2m})^{1/3}$. Then, the gap is
\[  \frac{(d_4+1)^3}{s_0} \geq \frac{d_4+1}{4mn} \geq \frac{s^{1/3}}{2^{1/3}4m^{4/3}n}. \]
Finally, note that $s \geq d_4^3$. Thus, for $d_4^{3\epsilon} \geq (mn)^{O(1)}$ large enough, 
\[s^{\epsilon} \geq d_4^{3\epsilon} \geq {2^{1/3}4m^{4/3}n} \implies \frac{s^{1/3}}{2^{1/3}4m^{4/3}n} \geq s^{1/3 - \epsilon}.\]
Hence, the gap is at least $s^{1/3 - \epsilon}$. Therefore, $\SAT$ reduces to $\gapet$ for $\alpha = s^{1/3 - \epsilon}$.
\end{proof} 
    
\subsubsection{For finite characteristic fields} \label{section-gap-affine-finitechar}
Let the characteristic be $p$, where $p >2$. If $p > d_1$, where $d_1$ is as chosen in Section \ref{section-gap-affine}, then the argument of that section holds. Hence, it is assumed that $p \leq d_1 =  (mn)^{O(1)}$. We again consider the polynomial $f$ as defined in \eqref{affine variable poly} and impose the following constraints on the $d_i$'s.
\begin{equation}\label{Eqn: gap ineq affinprime}
    d_4 \geq \max(3pmn,(mn)^{O(1/\epsilon)}), \ d_3 > m(d_4+1)^2, \ d_2 > (d_3+1)^2, \
    d_1 > 6nd_2 + 6n + 2d_3 + 2.
\end{equation}
Note, we can get $d_3 = O(pm(d_4+1)^2)$, $d_2 = O(p(d_3+1)^2)$ and $d_1 = O(p(6nd_2+6n+2d_3+2))$. The conditions of \eqref{Affine ETsparse conditions} are also satisfied under \eqref{Eqn: gap ineq affinprime}. From Section \ref{Section: Affine ETsparse}, it follows that for satisfiable $\psi$'s, there exists $A \in \GL(|\vecz|,\F)$, $\vecb = \mathbf{0}$ such that $\cal S(f(A\vecz)) \leq s_0$, where $s_0 = 1 + n(d_3 + 3) + m(d_4+1)^2$. For unsatisfiable $\psi$'s, Lemma \ref{Lemma: unsat gap affine char p} gives lower bounds on $\cal S(f(A\vecz+\vecb))$, where $A \in \GL(|\vecz|,\F)$ and $\vecb \in \F^{|\vecz|}$. Proposition \ref{Proposition: Sparsity gap affine char p} shows $\gapet$ is $\NP$-hard using Lemma \ref{Lemma: unsat gap affine char p} and the constraints in \eqref{Eqn: gap ineq affinprime}.

\begin{lemma} \label{Lemma: unsat gap affine char p}
   Let $\psi \in \bar{\SAT}$, $f(\vecz)$ be as defined in \eqref{affine variable poly} corresponding to $\psi$, $A \in \GL(|\vecz|,\F)$ and $\vecb \in \F^{|\vecz|}$.
   \begin{enumerate}
       \item \label{case 1 unsat gap affine char p} If $A(x_0) + \vecb_{|x_0}$ is a non-trivial affine form, then $\cal S(f(A\vecz + \vecb)) \geq \frac{d_1}{2}$.

       \item \label{case 2 unsat gap affine char p} If $A$ and $\vecb$ are not as in item \ref{case 1 unsat gap affine char p} and for some $j \in [n]$, $A(x_j) + \vecb_{|x_j}$ is a non-trivial affine form, then $\cal S(f(A\vecz + \vecb)) \geq d_2+1$.

       \item If $A$ and $\vecb$ are not as in items \ref{case 1 unsat gap affine char p} and \ref{case 2 unsat gap affine char p} and for some $j \in [n]$, $\cal S(A(y_j + x_j) + \vecb_{|y_j}) \geq 3$ or $\cal S(A(y_j - x_j) + \vecb_{|y_j}) \geq 3$, then $\cal S(f(A\vecz + \vecb)) \geq (d_3 + 1)^{1.63}$. 

       \item If $A$ and $\vecb$ are not of the form described in the previous three cases, then $\vecb = \mathbf{0}$ and  $\cal S(f(A\vecz)) \geq (d_4 + 1)^3.$
   \end{enumerate}
\end{lemma}
\begin{proof}
The first, second and fourth cases can be proved similarly to those of Lemma \ref{Lemma: unsat gap affine}. Hence, we consider the third case. Then, for $i \in [0,n]$, $A(x_i)$ is some scaled variable in $\vecz$ and $\vecb_{|x_i} = 0$. If $\vecb_{|y_j} = 0$, then this case is the same as the third case of Lemma \ref{Lemma: unsat gap char p}. If $\vecb_{|y_j} \neq 0$, then, without loss of generality, $A(y_j + x_j) + \vecb_{|y_j} = \ell_j + b_j$, with $\ell_j$ a linear form in at least two variables. Using the binomial theorem, the fact that $d_3 = p^k-1 = \sum_{i=0}^{k-1}(p-1)p^i$, Lucas's theorem and Observations \ref{lem_deg_sep_sum} and \ref{Lemma: linear form sparsity} gives
\[\cal S((\ell_j + b_j)^{d_3}) = \cal S\left(\sum_{i=0}^{d_3}\binom{d_3}{i}b_j^{d_3-i}\ell_j^i\right) = \sum_{i=0}^{d_3}\cal S\left(\binom{d_3}{i}b_j^{d_3-i}\ell_j^i\right)\] \[ \geq \sum_{i=0}^{d_3} \prod_{l=0}^{k-1}\binom{e_{i,l}+2-1}{2-1} = \sum_{i=0}^{d_3} \prod_{l=0}^{k-1}(e_{i,l}+1)\]
where $i = \sum_{l=0}^{k-1}e_{i,l} p^l$ with $e_{i,l} \in [0,p-1]$ ($e_{i,l}$ represents the $l\ith$ ``digit'' in the base-$p$ expansion of $i$). Note that for $i = rp$, where $r \in [0,p^{k-1}-1]$, the value of $e_{t,l}$, where $l \geq 1$ and $t \in [i,i+p-1]$, is the same for all  $t$ while $e_{t,0} = t-i$. Therefore, for such $i$'s, the following holds
\[\sum_{t=i}^{i+p-1}\prod_{l=0}^{k-1}(e_{t,l}+1) = \bigg(\prod_{l=1}^{k-1}(e_{rp,l}+1)\bigg) \cdot \sum_{t=rp}^{rp+p-1}(t-rp+1) = \frac{p(p+1)}{2}\prod_{l=1}^{k-1}(e_{rp,l}+1). \]
Using the above observation and the fact that $d_3 = p^k-1$,
\[\sum_{i=0}^{d_3} \prod_{l=0}^{k-1}(e_{i,l}+1) = \sum_{r=0}^{p^{k-1}-1} \sum_{t=0}^{p-1}\prod_{l=0}^{k-1}(e_{rp+t,l}+1) = \frac{p(p+1)}{2}\sum_{r=0}^{p^{k-1}-1}\prod_{l=1}^{k-1}(e_{rp,l}+1).\]
By repeating the same argument, we get 
\[\sum_{i=0}^{d_3} \prod_{l=0}^{k-1}(e_{i,l}+1) = \bigg(\frac{p(p+1)}{2}\bigg)^k = (d_3+1)^{1+ \log_p((p+1)/2)}.\]
Now, $\log_p((p+1)/2)$ is an increasing function for $p \geq 3$. Thus, $\log_p((p+1)/2) \geq \log_3((3+1)/2) \geq 0.63$. We can then conclude that
\[\cal S(f(A\vecz)) \geq \cal S(Q_{j,2}(A\vecz + \vecb)) \geq (d_3+1)^{1.63}.\]
Similarly, if $\cal S(A(y_j - x_j) + \vecb_{|y_j}) \geq 3$, then $\cal S(f(A\vecz + \vecb)) \geq \cal S(Q_{j,3}(A\vecz + \vecb)) \geq (d_3 + 1)^{1.63}$.
\end{proof}

\begin{proposition} \label{Proposition: Sparsity gap affine char p}
    Let $\textnormal{char}(\F) = p > 2$. If the input in $\gapet$ is an $s$-sparse polynomial, then $\gapet$ is $\NP$-hard for $\alpha = s^{1/3-\epsilon}$.
\end{proposition}
\begin{proof}
The proof is similar to that of Proposition \ref{Proposition: Sparsity gap affine}. For the polynomial $f$ defined in \eqref{affine variable poly}, $s := \cal S(f) = 1 + n(2d_3 + 3) + m(d_4 + 1)^3$. If $\psi \in \SAT$, then $\cal S(f(A\vecz)) \leq s_0$ where $A$ is as described in \eqref{action} and $s_0 = 1 + n(d_3+3) + m(d_4+1)^2$. If $\psi \in \bar{\SAT}$, it follows from Lemma \ref{Lemma: unsat gap affine} that for any $A \in \GL(|\vecz|,\F)$ and $\vecb \in \F^{|\vecz|}$: 
\[\cal S(f(A\vecz + \vecb)) \geq \min\bigg(\frac{d_1}{2},d_2+1,(d_3+1)^{1.63},(d_4+1)^3\bigg).\]
The conditions imposed in \eqref{Eqn: gap ineq affinprime} ensure that $(d_4+1)^3 > s_0$ and $(d_4+1)^3$ is the minimum. As $d_3 > m(d_4+1)^2$, therefore $s_0 = 1 + n(d_3+3) + m(d_4+1)^2 \leq 3nd_3 \leq 3pmn(d_4+1)^2$. Consequently, the gap in the sparsities of the YES instances and NO instances is 
\[  \frac{(d_4+1)^3}{s_0} \geq \frac{(d_4+1)^3}{3pmn(d_4+1)^2} = \frac{d_4+1}{3pmn}. \]
Also, note that as $d_4 \geq 3pmn$, therefore $s \leq 2m(d_4+1)^3 \implies d_4+1 \geq (\frac{s}{2m})^{1/3}$. Then, the gap is
\[  \frac{(d_4+1)^3}{s_0} \geq \frac{d_4+1}{3pmn} \geq \frac{s^{1/3}}{p2^{1/3}3m^{4/3}n}. \]
Finally, note that $s \geq d_4^3$. Thus, for $d_4^{3\epsilon} \geq (mn)^{O(1)}$ large enough, 
\[s^{\epsilon} \geq d_4^{3\epsilon} \geq {p2^{1/3}3m^{4/3}n} \implies \frac{s^{1/3}}{p2^{1/3}3m^{4/3}n} \geq s^{1/3 - \epsilon}.\]
Hence, the gap is at least $s^{1/3 - \epsilon}$. Therefore, $\SAT$ reduces to $\gapet$ for $\alpha = s^{1/3 - \epsilon}$.
\end{proof}
\subsubsection*{\underline{For characteristic $2$ fields}} \label{section-gap-affine-char2}
Let the characteristic be $2$. Consider the polynomial as defined in \eqref{affine variable poly char 2}. Let $s := \cal S(f)$. Now, $s$ depends on the number of variables complemented in a clause. To prove the hardness of $\gapet$, $s \geq d_4^3$ is required (see the proof of Proposition \ref{Proposition: Sparsity gap affine char 2}), and this can be achieved if there is at least one clause where all the variables are complemented. Thus, assume, without loss of generality, that such a clause exists (see footnote \ref{footnote: complemented clause}). Choose the $d_i$'s to satisfy \eqref{Eqn: gap ineq affinprime} with $p$ set to $2$. By Observation \ref{Obs: translated poly sparsity char 2} and the assumption on $\psi$, it holds that  
\[1 + n(d_3 + 3) + (d_4 + 1)^3 \leq s \leq 1 + n(d_3 + 3) + m(d_4 + 1)^3.\]
For $\psi \in \bar{\SAT}$, Lemma \ref{Lemma: unsat gap affine char 2}, which can be proved in the same way as Lemma \ref{Lemma: unsat gap affine char p}, shows lower bounds on $\cal S(f(A\vecz + \vecb))$ where $A \in \GL(|\vecz|,\F)$ and $\vecb \in \F^{|\vecz|}$. For $\psi \in \SAT$, by Proposition \ref{Proposition: fwd direction char 2} there exists $A \in \GL(|\vecz|,\F)$ such that $\cal S(f(A\vecz)) \leq s_0$, where $s_0 = 1 + n(d_3 + 3) + m(d_4+1)^2$. Proposition \ref{Proposition: Sparsity gap affine char 2} shows the $\NP$-hardness of $\gapet$ using Lemma \ref{Lemma: unsat gap affine char 2} and the inequalities in \eqref{Eqn: gap ineq affinprime}.

\begin{lemma} \label{Lemma: unsat gap affine char 2}
   Let $\psi \in \bar{\SAT}$, $f(\vecz)$ be as defined in \eqref{affine variable poly char 2} corresponding to $\psi$, $A \in \GL(|\vecz|,\F)$ and $\vecb \in \F^{|\vecz|}$.
   \begin{enumerate}
       \item \label{case 1 unsat gap affine char 2} If $A(x_0) + \vecb_{|x_0}$ is a non-trivial affine form, then $\cal S(f(A\vecz + \vecb)) \geq \frac{d_1}{2}$.

       \item \label{case 2 unsat gap affine char 2} If $A$ and $\vecb$ are not as in item \ref{case 1 unsat gap affine char p} and for some $j \in [n]$, $A(x_j) + \vecb_{|x_j}$ is a non-trivial affine form, then $\cal S(f(A\vecz + \vecb)) \geq d_2+1$.

       \item If $A$ and $\vecb$ are not as in items \ref{case 1 unsat gap affine char 2} and \ref{case 2 unsat gap affine char 2} and for some $j \in [n]$, $\cal S(A(y_j + x_j) + \vecb_{|y_j}) \geq 3$ or $\cal S(A(y_j) + \vecb_{|y_j}) \geq 3$, then $\cal S(f(A\vecz + \vecb)) \geq (d_3 + 1)^{1.58}$. 

       \item If $A$ and $\vecb$ are not of the form described in the previous three cases, then $\vecb = \mathbf{0}$ and  $\cal S(f(A\vecz)) \geq (d_4 + 1)^3.$
   \end{enumerate}
\end{lemma}
\begin{proposition} \label{Proposition: Sparsity gap affine char 2}
    Let $\textnormal{char}(\F) = 2$. If the input in $\gapet$ is an $s$-sparse polynomial, then $\gapet$ is $\NP$-hard for $\alpha = s^{1/3-\epsilon}$. 
\end{proposition}
\begin{proof}
The proof is similar to that of Proposition \ref{Proposition: Sparsity gap affine char p}. If $\psi \in \SAT$, then $\cal S(f(A\vecz)) \leq s_0$ where $A$ is as described in \eqref{action} and $s_0 = 1 + n(d_3+3) + m(d_4+1)^2$. If $\psi \in \bar{\SAT}$, it follows from Lemma \ref{Lemma: unsat gap affine char 2} that for any $A \in \GL(|\vecz|,\F)$ and $\vecb \in \F^{|\vecz|}$: 
\[\cal S(f(A\vecz + \vecb)) \geq \min\bigg(\frac{d_1}{2},d_2+1,(d_3+1)^{1.58},(d_4+1)^3\bigg).\]
The conditions imposed in \eqref{Eqn: gap ineq affinprime} ensure that $(d_4+1)^3 > s_0$ and $(d_4+1)^3$ is the minimum. As $d_3 > m(d_4+1)^2$, therefore $s_0 = 1 + n(d_3+3) + m(d_4+1)^2 \leq 3nd_3 \leq 6mn(d_4+1)^2$. Consequently, the gap in the sparsities of the YES instances and NO instances is 
\[  \frac{(d_4+1)^3}{s_0} \geq \frac{(d_4+1)^3}{6mn(d_4+1)^2} = \frac{d_4+1}{6mn}. \]
Also, note that as $d_4 \geq 6mn$, therefore $s \leq 2m(d_4+1)^3 \implies d_4+1 \geq (\frac{s}{2m})^{1/3}$. Then, the gap is
\[  \frac{(d_4+1)^3}{s_0} \geq \frac{d_4+1}{6mn} \geq \frac{s^{1/3}}{2^{4/3}3m^{4/3}n}. \]
Finally, note that $s \geq d_4^3$. Thus, for $d_4^{3\epsilon} \geq (mn)^{O(1)}$ large enough, 
\[s^{\epsilon} \geq d_4^{3\epsilon} \geq {2^{4/3}3m^{4/3}n} \implies \frac{s^{1/3}}{2^{4/3}3m^{4/3}n} \geq s^{1/3 - \epsilon}.\]
Hence, the gap is at least $s^{1/3 - \epsilon}$. Therefore, $\SAT$ reduces to $\gapet$ for $\alpha = s^{1/3 - \epsilon}$. 
\end{proof}

\section{Missing proofs from Section \ref{sec:prelim}}\label{Section: Proofs prelim}
\subsection{Proof of Observation \ref{Lemma: equiv degrees same}} 

Under any invertible linear transform applied to the variables of $f$, every monomial of $f$ maps to a linear combination of monomials of the same degree. Thus, no new degree can be added to the set of degrees of $f$ under any invertible linear transform. As $f \sim g$, the set of degrees of $f$ is contained in the set of degrees of $g$, and vice versa, implying the two sets are the same.

\subsection{Proof of Observation \ref{lem_deg_sep_sum}}
    As $f$ and $g$ are degree separated (or degree separated with respect to some variable), each monomial of $f+g$ is a monomial of $f$ or $g$, but not both.

\subsection{Proof of Observation \ref{Lemma: equiv deg sep}}
     As $f_1 \sim f$ and $g_1 \sim g$, therefore $f_1$ and $g_1$ are degree separated by Observation \ref{Lemma: equiv degrees same}. By Observation \ref{lem_deg_sep_sum}, the statement holds.
\subsection{Proof of Observation \ref{Lemma: linear form sparsity}}
    Without loss of generality, let $\ell = \sum_{i=1}^{m} c_i x_i$ where $c_i \in \F \backslash \{0\}$. If $\textnormal{char}(\F) = 0$, then $\cal S(\ell^d) = \binom{d+m-1}{m-1}$ follows from the multinomial theorem and the fact that the number of monomials of degree $d$ in $m$ variables is $\binom{d+m-1}{m-1}$. Suppose $\textnormal{char}(\F) = p$. Then, $d$ is expressible as in the observation statement. It will be shown by   induction on $k$ that,
    \[\cal S(\ell^d) = \prod_{i=0}^{k}\binom{e_i+m-1}{m-1}.\]
    In the base case $k = 0$, $d < p$ and, like the $\textnormal{char}(\F) = 0$ case, it easily follows that $\cal S(\ell^d) = \binom{d+m-1}{m-1}$. Assume the statement for all $j < k$. Suppose $d = e_kp^k + \sum_{i=0}^{k-1} e_ip^i$, where $0 < e_k < p$. Then, using the fact that $(\sum_{j=1}^{m} c_jx_j)^p = \sum_{j=1}^{m} c_j^p x_j^p$ over $\F$, 
    \[\ell^d = \bigg(\sum_{j=1}^{m} c_j x_j\bigg)^{\sum_{i=0}^{k}e_ip^i} = \bigg(\sum_{j=1}^{m} c_j^{p^k} x_j^{p^k}\bigg)^{e_k} \dot \prod_{i=0}^{k-1}\bigg(\sum_{j=1}^{m} c_j^{p^i} x_j^{p^i}\bigg)^{e_i}.\]
    Let $h = \prod_{i=0}^{k-1}\big(\sum_{j=1}^{m} c_j^{p^i} x_j^{p^i}\big)^{e_i}$. Note,
    \[\bigg(\sum_{j=1}^{m} c_j^{p^k} x_j^{p^k}\bigg)^{e_k} = \sum_{\alpha_1 + \dots + \alpha_m = e_k} \binom{e_k}{\alpha_1 \dots \alpha_m} \bigg(\prod_{i=1}^{m} (c_i^{p^k}x_i^{p^k})^{\alpha_i}\bigg).\]
    By the inductive hypothesis, $\cal S(h) = \prod_{i=0}^{k-1}\binom{e_i+m-1}{m-1}$, while $\cal S((\sum_{j=1}^{m} c_j^{p^k} x_j^{p^k})^{e_k}) = \binom{e_k + m-1}{m-1}$, as $e_k < p$. Now,
    \[\ell^d = \sum_{\alpha_1 + \dots + \alpha_m = e_k} \binom{e_k}{\alpha_1 \dots \alpha_m} \bigg(\prod_{i=1}^{m} (c_i^{p^k}x_i^{p^k})^{\alpha_i}\bigg) \cdot h.\]
    The degree of $h < p^k$, while any two monomials in the above expansion are degree separated by at least $p^k$ in at least one variable. Consequently, by Observation \ref{lem_deg_sep_sum}, $\cal S(\ell^d) = \prod_{i=0}^{k}\binom{e_i+m-1}{m-1}$. The above inductive argument is similar to the multinomial version of Lucas's theorem. 

\subsection{Proof of Observation \ref{lem_binom}}
Let $\ell$ be a linear form in exactly $2$ variables. When $\textnormal{char}(\F) = 0$, then by Observation \ref{Lemma: linear form sparsity}, $\cal S(\ell^d) = \binom{d+2-1}{2-1} = d + 1$. When $\textnormal{char}(\F) = p$, then $\cal S(\ell^d) = \prod_{i=0}^{k-1}\binom{e_i+2-1}{2-1} = \prod_{i=0}^{k-1}(e_i + 1)$, where $d = \sum_{i=0}^{k-1}e_i p^i$. It is easy to see the observation holds when $d < p$. When $d = p^k - 1 = \sum_{i=0}^{k-1} (p-1)p^i$, then $\prod_{i=0}^{k-1}(e_i + 1) = p^k = d + 1$. Finally, when $\ell$ is a linear form in  $m \geq 2$ variables, then the observation follows from the fact that $\binom{c+m-1}{m-1} \geq c + 1$ for any $c \in \N$.

\subsection{Proof of Observation \ref{Lemma: affine form sparsity}}
    Using the binomial theorem, 
    \[h^d= \ell^d + c_0^d + \sum_{i=1}^{d-1} \binom d i c_0^i\ell^{d-i}.\]
    As the degree of every monomial in $\ell^{d-i}$ is $d-i$, all the summands in the above expansion are degree separated. From Observations \ref{lem_deg_sep_sum} and \ref{lem_binom}, it holds that $\cal S(h^d) \geq \cal S(\ell^d) + 1$. More precisely, $\cal S(h^d) \geq d+1$, as $\cal S(\ell^{d-i}) \geq 1$ and $\binom{d}{i} \neq 0$ for $d$ as in the observation statement (by Lucas's theorem).

\subsection{Proof of Claim \ref{lem_div}}
     We prove this by induction on $d$. For the base case $d=0$, it is easy to see that the sparsity of any non-zero polynomial is at least $1$. Suppose now the result holds for all $k < d$. Let $\ell = \sum_{i=1}^{n} c_i x_i$ and $f=\ell^d h$. Without loss of generality, assume $f$ is not divisible by any variable, for if it were divisible by some variable $x_i$, then $x_i$ must not divide $\ell$ as $\ell$ contains at least two distinct variables and hence $x_i$ divides $h$, in which case we can replace $f$ and $h$ by $\frac{f}{x_i}$ and $\frac{h}{x_i}$respectively. 
    
    Let $x_j$ be a variable in $\ell$ with a non-zero coefficient and consider $\frac {\partial f}{\partial x_j}$. Now, 
     \[\cal S(f)\geq 1 + \cal S\left(\frac {\partial f}{\partial x_j}\right)\]
     as the derivative map either sends monomials to distinct monomials or eliminates them, and by assumption some monomial in $f$ is not divisible by $x_j$ and will be eliminated. As $f=\ell^d h$, 
     \[\frac{\partial f}{\partial x_j}= c_jd\ell^{d-1} h + \ell^d \frac {\partial h}{\partial x_j}.\] 
    Clearly, $\ell^{d-1}$ divides $\frac{\partial f}{\partial x_j}$. By induction, $\cal S(\frac{\partial f}{\partial x_j}) \geq d$. Hence, $\cal S(f) \geq 1 + \cal S(\frac{\partial f}{\partial x_j}) \geq d + 1$. 

\subsection{Proof of Claim \ref{FullSupport}}
    The claim is first proven for $n=1$. Thus, $g = \ell^d$ where $\ell = \sum_{i=1}^{|\var(\ell)|} c_{i}x_{i}$, $c_i \neq 0$, and $|\var(\ell)| \geq \sigma$, without loss of generality. Note that
     \[\frac{\partial^\sigma g}{\partial x_{1} \cdots \partial x_{\sigma}} = \sigma!\binom{d}{\sigma} c_{1} c_{2} \cdots c_{\sigma} \ell^{d-\sigma}.\]
     Clearly, $\sigma!\binom{d}{\sigma} \neq 0$ when $\textnormal{char}(\F) = 0$. When $\textnormal{char}(\F) = p$ with $p > d$, or $p > \sigma$ and $d = p^k -1$ for some $k \in \N$, this follows by Lucas's Theorem \cite{Lucas78}. So the derivative is non-zero, as $d\geq \sigma$, implying there exists a monomial of support at least $\sigma$ in $g$.
     
     Now, for arbitrary $n$, $g = (\ell_1 \cdots \ell_n)^d$, where $|\cup_{i=1}^{n} \var(\ell_i)| \geq  \sigma$. Observe that
    \begin{equation*}
         \frac{\partial^\sigma g}{\partial x_{i_1} \cdots \partial x_{i_\sigma}} = \qquad \sum_{\mathclap{\substack{j_1 + \cdots + j_n = \sigma \\ j_i \geq 0}}} c_{j_1,\cdots,j_n} \cdot  
         \frac{\partial^\sigma (\ell_1^{j_1} \cdots \ell_n^{j_n})}{\partial x_{i_1} \cdots \partial x_{i_\sigma}} \cdot (\ell_1^{d-j_1} \cdots \ell_n^{d-j_n})
    \end{equation*}
     where $c_{j_1 \cdots j_n} = \binom{d}{j_1} \cdots \binom{d}{j_n}$. Clearly, when $\textnormal{char}(\F) = 0$ or $> d$, $c_{j_1,\cdots,j_n} \neq 0$. When $\textnormal{char}(\F) = p$ and $d = p^k -1$ for some $k \in \N$, then by Lucas's Theorem, all the binomial coefficients are non-zero. Hence $c_{j_1,\cdots,j_n} \not= 0$.

     Observe that the elements of the set $ \mathcal{M} : = \{\ell_1^{d-j_1} \cdots \ell_n^{d-j_n} \ |\ j_1 + \cdots + j_n = \sigma, \ \ j_i \geq 0\}$ are linearly independent as the $\ell_i$'s are linearly independent and $d \geq \sigma$. Also, $\frac{\partial^\sigma (\ell_1^{j_1} \cdots \ell_n^{j_n})}{\partial x_{i_1} \cdots \partial x_{i_\sigma}} \in \mathbb{F}$ as $j_1 + \cdots + j_n = \sigma$. It suffices to show that for some choice of $j_1,\dots,j_n$  and $x_{i_1},\dots,x_{i_\sigma}$ ,where $i_1,\dots,i_\sigma$ are pairwise distinct, $\frac{\partial^\sigma (\ell_1^{j_1} \cdots \ell_n^{j_n})}{\partial x_{i_1} \cdots \partial x_{i_\sigma}} \neq 0$. As elements of $\mathcal{M}$ are linearly independent and  $c_{j_1 \cdots j_n}\not= 0$, this would imply $\frac{\partial^\sigma g}{\partial x_{i_1} \cdots \partial x_{i_\sigma}} \neq 0$, indicating that $\supp(g) \geq \sigma$.
     
     For every $i \geq 1$, define $S_i := \var(\ell_i) \setminus \cup_{j=1}^{i-1} S_j$  if $|\cup_{j=1}^i S_j| < \sigma$ else choose $S_i \subseteq  \var(\ell_i) \setminus \cup_{j=1}^{i-1} S_j$ such that $|\cup_{j=1}^i S_j| = \sigma$; here, $\cup_{j=1}^{i-1} S_j = \emptyset$ for $i=1$. Note that such a collection of sets always exists as $|\cup_{i=1}^n \var(\ell_i)| \geq \sigma$. Say we choose $m \leq n$ such non-empty sets. Let $j_i := |S_i|$ and $S_i := \{x_{i1},\dots,x_{ij_i}\}$.  Hence, 
     \begin{equation*}
         \frac{\partial^\sigma (\ell_1^{j_1} \cdots \ell_m^{j_m})}{(\partial x_{11} \cdots \partial x_{1j_1}) \cdots (\partial x_{m1} \cdots \partial x_{mj_m})} = \prod_{i = 1}^{m} \frac{\partial^{j_i} \ell_i^{j_i}}{\partial x_{i1} \cdots \partial x_{ij_i}}.
     \end{equation*}
     As $\{x_{i1},\dots,x_{ij_i}\} \subseteq \var(\ell_i)$ is not empty and $j_i \leq \sigma < p$ ( in case of finite characteristic fields), by the analysis of the $n=1$ case, $\frac{\partial^{j_i} \ell_i^{j_i}}{\partial x_{i1} \cdots \partial x_{ij_i}} \not=0$. Hence,
     \[\frac{\partial^\sigma (\ell_1^{j_1} \cdots \ell_m^{j_m})}{(\partial x_{11} \cdots \partial x_{1j_1}) \cdots (\partial x_{m1} \cdots \partial x_{mj_m})} \not= 0\]
     and $\supp(g) \geq \sigma$.     

\section{Missing proofs from Section \ref{Subsection: Reduction ETSparse}} \label{Section: ETSparse proofs}
\subsection{Proof of Observation \ref{Obs: deg separated polys}}
    Let $i \in [n]$. Note that $Q_{i,1}(\vecz)$ has degree $(3i-2)d_1 + d_2$, $Q_{i,2}(\vecz)$ has degree $(3i-1)d_1 + d_3$ and $Q_{i,3}(\vecz)$ has degree $3id_1 + d_3$. Clearly $3id_1 + d_3 > (3i-1)d_1 + d_3$. Also, $(3i-1)d_1 + d_3 > (3i-2)d_1 + d_2$ because $d_1 > d_2 > d_2 - d_3$ by the conditions in \eqref{ineq}. Thus, $Q_{i,1}(\vecz)$, $Q_{i,2}(\vecz)$ and $Q_{i,3}(\vecz)$ are degree separated and $Q_i$ is a sum of $3$ degree separated polynomials and has degree $3id_1 + d_3$.

    Now, let $i \in [n]$ and $k \in [m]$. $R_k(\vecz)$ is a polynomial of degree $(3n+k)d_1 + 3d_4$ while the degree of $Q_i(\vecz)$ is $3id_1+d_3$. Note that

    \[(3n+k)d_1 +3d_4 \geq (3n+1)d_1 + 3d_4 > 3nd_1 + d_3 \geq 3id_1 + d_3,\]
    where the second inequality holds because $d_1 > d_3$ by the constraints in \eqref{ineq}. Therefore, $R_k(\vecz)$ and $Q_i(\vecz)$ (hence also $Q_{i,1}$, $Q_{i,2}$ and $Q_{i,3}$) are degree separated from one another. Further, the degree of $R_k(\vecz)$ and that of $Q_i(\vecz)$ are greater than $d_1$ implying $x_0^{d_1}$ is degree separated from $R_k(\vecz)$ and $Q_i(\vecz)$. 
    
    Lastly, let $i,j \in [n]$ where $i < j$ without loss of generality. The highest degree of a monomial in $Q_i(\vecz)$ is $3id_1 + d_3$, while the lowest degree of a monomial in $Q_j(\vecz)$ is $(3j-2)d_1 + d_2$. Now,
    \[(3j-2)d_1 + d_2 \geq (3i+1)d_1 + d_2 > 3id_1 + d_3\]
    as $j \geq i+1$ and $d_1 > d_2 > d_3$ by the conditions in \eqref{ineq}. Thus, $Q_i$ and $Q_j$ are degree separated. That $R_k(\vecz)$ is degree separated from $R_l(\vecz)$ for $k,l \in [n]$ and $k \neq l$ can be observed from the fact that the degree of $R_k(\vecz)$ is $(3n+k)d_1 + 3d_4$.
\subsection{Proof of Observation \ref{Obs: ETSparse poly degree}}
From the definition of $f$ in \eqref{Definition: 3SAT poly}, it follows that the degree of $f$ is the maximum of that of $x_0^{d_1}$, $Q_i$ and $R_k$ , where $i \in [n]$ and $k \in [m]$. As observed in the proof of Observation \ref{Obs: deg separated polys}, degree of $Q_i$ is $3id_1+d_3$, degree of $R_k$ is $(3n+k)d_1+ 3d_4$ and $(3n+k)d_1 + 3d_4 > 3id_1 + d_3 > d_1$. Since $k \leq m$, therefore the highest degree is $(3n+m)d_1 + 3d_4$. So, the degree of $f$ is $(3n+m)d_1 + 3d_4$.

\subsection{Proof of Observation \ref{Obs: ETsparse poly sparsity }}
By Observation \ref{Obs: deg separated polys}, $f$ is a sum of the $n+m+1$ degree separated polynomials $x_0^{d_1}$, $Q_i$ and $R_k$, where $i \in [n]$ and $k \in [m]$. Applying Observation \ref{lem_binom} (for linear forms in two variables over $\textnormal{char}(\F) = 0$ fields) to $Q_{i,2}$ and $Q_{i,3}$ and Observation \ref{lem_deg_sep_sum} to $Q_i$, we get
\[\cal S(Q_i(\vecz)) = \cal S(Q_{i,1}(\vecz)) + \cal S(Q_{i,2}(\vecz)) + \cal S(Q_{i,3}(\vecz)) = 2d_3+3, \ \  \forall i \in [n].\]
By Observation \ref{lem_binom} (for linear forms in two variables over $\textnormal{char}(\F) = 0$ fields) and the assumption that each clause in $\psi$ has $3$ distinct variables, we get that:
\[\cal S(R_k(\vecz)) = \cal S(x_0^{(3n+k)d_1})\prod_{j \in C_k} \cal S((y_j + (-1)^{a_{k,j}}x_j)^{d_4}) = (d_4 + 1)^3 \ \ \forall k \in [m].\]
Finally, applying Observations \ref{Obs: deg separated polys} and \ref{lem_deg_sep_sum} to $f$ gives
\[\cal S(f(\vecz)) = \cal S(x_0^{d_1}) + \sum_{i=1}^{n} \cal S(Q_i(\vecz)) + \sum_{k=1}^{m} \cal S(R_k(\vecz)) = 1 + n(2d_3 + 3) + m(d_4+1)^3.\]
Thus, $\cal S(f(\vecz)) > s$ but also $(mn)^{O(1)}$. 

For the support of $f$, note that $\supp(R_k) = 7$ for all $k \in [m]$ while $\supp(Q_i) = 3$ for all $i \in [n]$ and $\supp(x_0^{d_1}) = 1$. By Observation \ref{Obs: deg separated polys}, $\supp(f) = \supp(R_k) = 7$.

\subsection{Proof of Lemma \ref{Lemma: fix x0}}
    If $A(x_0)$ is a linear form in at least two variables, then it follows from the definition of $f$ in \eqref{Definition: 3SAT poly}, Observation \ref{Obs: deg separated polys}, Observation \ref{lem_binom} applied on $A(x_0^{d_1})$, Observation \ref{Lemma: equiv deg sep}, and the constraint $d_1 \geq s$ in \eqref{ineq} that $\cal S(f(A\vecz) > \cal S(A(x_0^{d_1})) \geq d_1 + 1 > s$, a contradiction. Hence, $A(x_0)$ has only one variable. By multiplying $A$ with a permutation and a scaling matrix, we can assume without loss of generality that $A(x_0) = x_0$. This can be assumed because permutation and non-zero scaling of variables do not affect the sparsity of a polynomial.

\subsection{Proof of Lemma \ref{var_spars_lem}}
It follows from Observations \ref{Obs: deg separated polys} and \ref{Lemma: equiv deg sep} that for any $A \in \GL(|\vecz|,\F)$, $\cal S(Q_i(A\vecz)) = \cal S(Q_{i,1}(A\vecz))+\cal S(Q_{i,2}(A\vecz))+\cal S(Q_{i,3}(A\vecz))$, where $Q_i$ is as described in Section \ref{subsubsection: sparse f construction}. Now, the if direction in the lemma statement is easy to verify. For the only if direction, consider the following cases of $A$:

\begin{enumerate}
    \item $\cal S(A(x_i))\geq 2$:  It follows from Observation \ref{lem_binom} and $d_2 \geq 2d_3$ that $\cal S(Q_{i,1}(A\vecz)) \geq d_2+1 \geq 2d_3 + 1$. Also, $\cal S(Q_{i,2}(A\vecz)) \geq 1$ and $\cal S(Q_{i,3}(A\vecz)) \geq 1$. Thus, $\cal S(Q_{i}(A\vecz)) \geq d_2+3 \geq 2d_3 + 3$.
    
    \item $\cal S(A(x_i)) = 1$, $\cal S(A(y_i + x_i))\geq 2$ and $\cal S(A(y_i -x_i))\geq 2$: It follows from Observation \ref{lem_binom} that $\cal S(Q_{i,2}(A\vecz)) \geq d_3+1$ and $\cal S(Q_{i,3}(A\vecz)) \geq d_3+1$ implying $\cal S(Q_{i}(A\vecz)) \geq 2d_3 + 3$.
    
    \item $\cal S(A(x_i)) = 1$ with  $\cal S(A(y_i + x_i)) = 1$ or $\cal S(A(y_i - x_i)) = 1$: Let $A(x_i) = X_i$ for some scaled variable $X_i \in \vecz$. Because $A$ is invertible exactly one of $\cal S(A(y_i + x_i)) = 1$ or $\cal S(A(y_i - x_i)) = 1$ holds true. Let $\cal S(A(y_i + x_i)) = 1$, without loss of generality. Then, it must be that $A(y_i) = Y_i-X_i$ for some scaled variable $Y_i \in \vecz$. Thus, $A(y_i - x_i) = Y_i - 2X_i$. Hence, $\cal S(Q_{i,1}(A\vecz)) = 1$, $\cal S(Q_{i,2}(A\vecz)) = 1$ and $\cal S(Q_{i,3}(A\vecz)) = d_3 + 1$ (by Observation \ref{lem_binom}) implying $\cal S(Q_{i}(A\vecz)) = d_3 + 3$.
\end{enumerate}
The first two cases show that if $A$ is not as per the lemma statement, then $\cal S(Q_i(A\vecz)) \geq 2d_3 + 3$; otherwise, $\cal S(Q_{i}(A\vecz)) = d_3 + 3$.

\subsection{Proof of Lemma \ref{Lemma: equality condition Qi}}
Suppose $\cal S(Q_j(A\vecz)) \neq d_3 + 3$ for some $j \in [n]$. Then, $\cal S(Q_j(A\vecz)) \geq 2d_3 + 3$ by Lemma \ref{var_spars_lem}. By the definition of $f$ in \eqref{Definition: 3SAT poly}, Observations \ref{Obs: deg separated polys} and \ref{Lemma: equiv deg sep} and the condition $d_3 \geq m(d_4+1)^2 + 1$, we get the following contradiction:
    \begin{equation*}
    \begin{split}
        \cal S(f(A\vecz)) &>  \cal S(A(x_0^{d_1})) + \sum_{i=1, i \neq j}^{n} \cal S(Q_i(A\vecz)) + \cal S(Q_j(A\vecz)) \geq  1+(n-1)(3+d_3)+(3+2d_3) \\
         &=1+n (3+d_3) +d_3 = s- m(d_4+1)^2 +d_3 > s.\end{split}
    \end{equation*}

\subsection{Proof of Observation \ref{Obs: deg separated polys homogeneous}}
The observation follows from the fact that the $x_0$-degree of the summands in $f$, as defined in \eqref{Homogeneous poly}, form an arithmetic progression with common difference $d_3 + 1$ and hence every polynomial in the observation statement has distinct $x_0$-degree. 

\subsection{Proof of Observation \ref{Obs: sparsity poly homogeneous}}
By Observation \ref{Obs: deg separated polys homogeneous}, $f$ is a sum of the $n+m+1$ polynomials $x_0^{d_1}y_0^{d_2 + (3n+m+1)(d_3+1)}$, $Q_i$ and $R_k$, where $i \in [n]$ and $k \in [m]$, which are degree separated with respect to $x_0$. Using arguments similar to the proof of Observation \ref{Obs: ETsparse poly sparsity }, it holds that

\[\cal S(Q_i(\vecz)) = \cal S(Q_{i,1}(\vecz)) + \cal S(Q_{i,2}(\vecz)) + \cal S(Q_{i,3}(\vecz)) = 2d_4+3 \ \  \forall i \in [n],\]
\[\cal S(R_k(\vecz)) = \cal S(x_0^{d_1+(3n+k)(d_3+1)}y_0^{d_2+(m-k+1)(d_3+1) - 3d_5})\prod_{j \in C_k} \cal S((y_j + (-1)^{a_{k,j}}x_j)^{d_5}) = (d_5 + 1)^3 \ \ \forall k \in [m],\]
and 
\[\cal S(f(\vecz)) = \cal S(x_0^{d_1}y_0^{d_2+(3n+m+1)(d_3+1)}) + \sum_{i=1}^{n} \cal S(Q_i(\vecz)) + \sum_{k=1}^{m} \cal S(R_k(\vecz)) = 1 + n(2d_4 + 3) + m(d_5+1)^3.\]
Thus, $\cal S(f(\vecz)) > s$ but also $(mn)^{O(1)}$. For the support of $f$, note that $\supp(R_k) = 8$ for all $k \in [m]$ while $\supp(Q_i) = 4$ for all $i \in [n]$ and $\supp(x_0^{d_1}y_0^{d_2+(3n+m+1)(d_3+1)}) = 2$. Hence, by Observation \ref{Obs: deg separated polys homogeneous}, $\supp(f) = \supp(R_k) = 8$.

\subsection{Proof of Lemma \ref{Lemma: y0,x0 fixed homogeneous}}
Suppose one of $A(x_0)$ or $A(y_0)$ is a linear form in at least two variables. As $A(x_0)^{d_1}$ and $A(y_0)^{d_2}$ divide $f(A\vecz)$, Claim \ref{lem_div} and the conditions of \eqref{ineq_homogeneous} imply $\cal S(f(A\vecz)) > s$, a contradiction. So, $A(x_0)$ and $A(y_0)$ must have only one variable each. Hence, without loss of generality (i.e., after applying scaling and permutation to $A$), $A(x_0) = x_0$ and $A(y_0) = y_0$.

\subsection{Proof of Lemma \ref{Lemma: y0 separated}}
Note that each of the $3n+m+1$ summand polynomials, as mentioned in the lemma statement, is of form $x_0^{d_1 + t(d_3+1)} \cdot y_0^{d_2 + v} \cdot h(\vecz)$, where $t \in [0,3n+m]$ and $v \in \N$. By the construction of $f$, each $t$ corresponds to a unique summand polynomial. For $x_0^{d_1}y_0^{d_2 + (3n+m+1)(d_3+1)}$ the degree of $h$ is $0$. Let $i \in [n]$. For $Q_{i,1}(A\vecz)$, the degree of $h$ is $d_3$ while for $Q_{i,2}(A\vecz)$ and $Q_{i,3}(A\vecz)$ the degree of $h$ is $d_4$. Lastly, for $R_k(A\vecz)$, where $k \in [m]$, the degree of $h$ is $3d_5$. Thus, going over all the summand polynomials, the degree of $h$ is at most the maximum of $d_3, d_4$ and $3d_5$. Since $d_3 > d_4$ and $d_3 > 3d_5$ by the conditions in \eqref{ineq_homogeneous}, therefore the degree of $h$ is at most $d_3$. As shown in Lemma \ref{Lemma: y0,x0 fixed homogeneous}, $A(x_0) = x_0$ and $A(y_0) = y_0$ while there may be monomials in $h(\vecz)$ which have non-zero $x_0$-degree. Thus, the possible range of $x_0$-degree of any monomial of a summand polynomial lies in the range $[d_1 + t(d_3+1), d_1 + t(d_3+1) + d_3]$, which is clearly disjoint for distinct $t$. Therefore, the summand polynomials are degree separated with respect to $x_0$.

\subsection{Proof of Observation \ref{Obs: poly sparsity char 2 case}} \label{proof-obs-polychar2}
    By Observation \ref{Obs: deg separated polys}, $f$ is a sum of the $n+m+1$ degree separated polynomials $x_0^{d_1}$, $Q_i$ and $R_k$, where $i \in [n]$ and $k \in [m]$. Applying Observation \ref{lem_binom} (for linear forms in two variables over finite characteristic fields) to $Q_{i,2}$ and Observation \ref{lem_deg_sep_sum} to $Q_i$, we get
\[\cal S(Q_i(\vecz)) = \cal S(Q_{i,1}(\vecz)) + \cal S(Q_{i,2}(\vecz)) + \cal S(Q_{i,3}(\vecz)) = d_3+3, \ \  \forall i \in [n].\]
    By Observation \ref{lem_binom} (for linear forms in two variables over finite characteristic fields) and the assumption that each clause in $\psi$ has $3$ distinct variables, we get that:
\[\cal S(R_k(\vecz)) = \cal S(x_0^{(3n+k)d_1})\prod_{j \in C_k} \cal S((y_j + a_{k,j}x_j)^{d_4}) \leq (d_4 + 1)^3 \ \ \forall k \in [m].\]
    depending on $a_{k,j} = 0$ or $1$. Finally, applying Observations \ref{Obs: deg separated polys} and \ref{lem_deg_sep_sum} to $f$ gives
\[\cal S(f(\vecz)) = \cal S(x_0^{d_1}) + \sum_{i=1}^{n} \cal S(Q_i(\vecz)) + \sum_{k=1}^{m} \cal S(R_k(\vecz)) \leq 1 + n(d_3 + 3) + m(d_4+1)^3.\]
    For the support of $f$, note that $4 \leq \supp(R_k) \leq 7$ for $k \in [m]$, $\supp(Q_i) = 3$ for all $i \in [n]$ and $\supp(x_0^{d_1}) = 1$. By Observation \ref{Obs: deg separated polys}, $4 \leq \supp(f) \leq 7$.
\subsection{Proof of Lemma \ref{Lemma: y0,x0 fixed homogeneous} over finite characteristic fields} \label{proof-x0,y0-fixed-homogeneous-finitechar}
Note that as $f$ is divisible by $x_0^{d_1}$ and $y_0^{d_2}$, therefore for any $A \in \GL(|\vecz|,\F)$ we can write:
    \[f(A\vecz) = A(x_0^{d_1})A(y_0^{d_2})g(A\vecz).\]
where the degree of $g(A\vecz)$ is $(3n+m+1)(d_3+1)$. Let $A(x_0) = \sum_{l=1}^{|\vecz|} c_lz_l$, where $z_l \in \vecz$ and $c_l \in \F$. The characteristic being finite and the choice of $d_1$ as in \eqref{d1 def} implies:
\[ A(x_0^{d_1}) = \bigg(\sum_{l=1}^{|\vecz|} c_lz_l\bigg)^{\sum_{t=k_3}^{k_1+k_3-1}(p-1)p^t} = \prod_{t=k_3}^{k_1+k_3-1}\bigg(\sum_{l=1}^{|\vecz|} c_l^{p^t}z_l^{p^t}\bigg)^{(p-1)}.\]
Thus, the monomials of $A(x_0^{d_1})$ are of form $\prod_{l=1}^{|\vecz|} z_l^{e_l}$,  where $e_l = \sum_{t=k_3}^{k_1+k_3-1} c_{l,t}p^t$ with $c_{l,t} \in [0,p-1]$, and $\sum_{l=1}^{|\vecz|} e_l = d_1$. Now, for any two monomials of $A(x_0^{d_1})$, there exists a variable $z_l \in \vecz$ such that the difference between the $z_l$-degree of these two monomials is at least $p^{k_3} > d_2 + (3n+m+1)(d_3+1)$, while the degree of $A(y_0^{d_2})g(A\vecz)$ is $d_2 + (3n+m+1)(d_3+1)$. This is because of the way $d_1$ is set in \eqref{d1 def}. Thus, with this observation and Observation \ref{lem_deg_sep_sum}, it holds that
    \[\cal S(f(A\vecz)) = \cal S(A(x_0^{d_1})) \cal S(A(y_0^{d_2})g(A\vecz)).\]
Similarly, because the characteristic is finite and $d_2$ is as chosen in \eqref{d2 def}, for any two monomials of $A(y_0^{d_2})$ there exists a variable $z_l \in \vecz$ such that the difference between the $z_l$-degree of these two monomials is at least $p^{k_2} > (3n+m+1)(d_3+1)$, while the degree of $g(A\vecz)$ is $(3n+m+1)(d_3+1)$. Thus, with this observation and Observation \ref{lem_deg_sep_sum}, it holds that
    \[\cal S(f(A\vecz)) =  \cal S(A(x_0^{d_1})) \cal S(A(y_0^{d_2})g(A\vecz)) = \cal S(A(x_0^{d_1})) \cal S(A(y_0^{d_2})) \cal S(g(A\vecz)).\]
Now, suppose $A(x_0) $ is a linear form in at least $2$ variables. By Observation \ref{Lemma: linear form sparsity} and the definition of $d_1$ in \eqref{d1 def}, it follows that 
    \[\cal S(A(x_0^{d_1})) \geq \prod_{k_3}^{k_1+k_3-1} \binom{p-1+2-1}{2-1} = p^{k_1} > s.\]
This implies $\cal S(f(A\vecz)) \geq \cal S(A(x_0^{d_1})) > s$. Thus, $A(x_0)$ must be some scaled variable in $\vecz$. Therefore, 
    \[\cal S(f(A\vecz)) = \cal S(A(y_0^{d_2})g(A\vecz)).\] 
Similarly, if $\cal S(A(y_0))$ is a linear form in at least $2$ variables, then by Observation \ref{Lemma: linear form sparsity} and the definition of $d_2$ in \eqref{d2 def}
    \[\cal S(A(y_0^{d_2})) \geq \prod_{k_2}^{k_1+k_2-1} \binom{p-1+2-1}{2-1} = p^{k_1} > s.\]
This implies $\cal S(f(A\vecz)) \geq \cal S(A(y_0^{d_2})) > s$. Thus, $A(y_0)$ must also be some scaled variable in $\vecz$. Therefore, $A(x_0) = x_0$ and $A(y_0) = y_0$ without loss of generality by applying an appropriate permutation and scaling transform.
    
\subsection{Proof of Observation \ref{Obs: poly sparsity char 2 case homogeneous}}
    By Observation \ref{Obs: deg separated polys homogeneous}, $f$ is a sum of the $n+m+1$ polynomials $x_0^{d_1}y_0^{d_2 + (3n+m+1)(d_3+1)}$, $Q_i$ and $R_k$, where $i \in [n]$ and $k \in [m]$, which are degree separated with respect to $x_0$. Applying Observation \ref{lem_binom} (for linear forms in two variables over finite characteristic fields) to $Q_{i,2}$ and Observation \ref{lem_deg_sep_sum} to $Q_i$, we get
\[\cal S(Q_i(\vecz)) = \cal S(Q_{i,1}(\vecz)) + \cal S(Q_{i,2}(\vecz)) + \cal S(Q_{i,3}(\vecz)) = d_4+3, \ \  \forall i \in [n].\]
    By Observation \ref{lem_binom} (for linear forms in two variables over finite characteristic fields) and the assumption that each clause in $\psi$ has $3$ distinct variables, we get that:
\[\cal S(R_k(\vecz)) = \cal S(x_0^{d_1 + (3n+k)(d_3+1)}) \cal S(y_0^{d_2 +(m-k+1)(d_3+1)-3d_5})\prod_{j \in C_k} \cal S((y_j + a_{k,j}x_j)^{d_5}) \leq (d_5 + 1)^3 \ \ \forall k \in [m].\]
    depending on whether $a_{k,j}$ is $0$ or $1$. Finally, applying Observations \ref{Obs: deg separated polys homogeneous} and \ref{lem_deg_sep_sum} to $f$ gives
\[\cal S(f(\vecz)) = \cal S(x_0^{d_1}y_0^{d_2 + (3n+m+1)(d_3+1)}) + \sum_{i=1}^{n} \cal S(Q_i(\vecz)) + \sum_{k=1}^{m} \cal S(R_k(\vecz)) \leq 1 + n(d_4 + 3) + m(d_5+1)^3.\]
    For the support of $f$, note that $5 \leq \supp(R_k) \leq 8$ for $k \in [m]$, $\supp(Q_i) = 4$ for all $i \in [n]$ and $\supp(x_0^{d_1}y_0^{d_2 + (3n+m+1)(d_3+1)}) = 2$. By Observation \ref{Obs: deg separated polys homogeneous}, $5 \leq \supp(f) \leq 8$.

\section{Missing proofs from Section \ref{Subsection: gap reduction}}\label{Section: Gap ETSparse proofs}
\subsection{Proof of Lemma \ref{Lemma: unsat gap}}
By Observations \ref{Obs: deg separated polys}, \ref{Lemma: equiv deg sep} and the definition of $f$ as in \eqref{Definition: 3SAT poly}, it follows that:
\[\cal S(f(A\vecz)) = \cal S(A(x_0^{d_1})) + \sum_{i=1}^{n} \cal S(Q_i(A\vecz)) + \sum_{k=1}^{m} \cal S(R_k(A\vecz)).\]
We now analyse $\cal S(f(A\vecz))$ under the transforms listed in the lemma statement. The list covers all possible types of transforms.
\begin{enumerate}
    \item By Observation \ref{lem_binom}, $\cal S(f(A\vecz)) \geq \cal S(A(x_0^{d_1})) \geq d_1 + 1$ follows.

    \item In this case, $A(x_0) = x_0$ without loss of generality, as permutation and non-zero scaling of variables do not influence the sparsity of the polynomial. By Observation \ref{lem_binom}, it follows that $\cal S(f(A\vecz)) \geq$ $\cal S(Q_{j,1}(A\vecz))  \geq d_2 + 1$.

    \item In this case, $A(x_i)$, where $i \in [0,n]$ is some scaled variable in $\vecz$. Without loss of generality, let $A(y_j + x_j)$ be a linear form in at least $3$ variables. By Observation \ref{Lemma: linear form sparsity}, $\cal S(Q_{j,2}(A\vecz)) \geq \binom{d_3+2}{2}$ holds. Therefore,
    \[\cal S(f(A\vecz)) \geq \cal S(Q_{j}(A\vecz)) \geq \cal S(Q_{j,2}(A\vecz)) \geq \binom{d_3+2}{2} = \frac{d_3^2 + 3d_3 + 2}{2}.\]

    \item In this case, $A(x_i) = X_i$, where $i \in [0,n]$ and $X_i \in \vecz$ is some scaled variable. Also, $A(y_i + x_i)$ and $A(y_i - x_i)$ are linear forms in at most two variables for all $i \in [n]$. Thus, $A(y_i) = Y_i + c_iX_i$, where $c_i \in \F$ and $Y_i \in \vecz$ is some scaled variable. As $A$ is invertible, the $Y_i$'s and $X_i$'s are distinct variables. Hence, 
    \[\cal S (R_k(A\vecz)) = \cal S(X_0^{(3n + k)d_1})\prod_{j \in C_k}\cal S((Y_j+ (c_j + (-1)^{a_{k,j}})X_j)^{d_4}). \]
    Since $\psi$ is unsatisfiable, for any such $A$, there exists $k \in [m]$ such that $\cal S(R_k(A\vecz)) \geq (d_4+1)^3$ (by Observation \ref{lem_binom}). Therefore, $\cal S(f(A\vecz)) \geq \cal S(R_k(A\vecz)) \geq (d_4+1)^3 $. 
\end{enumerate}
\subsection{Proof of Lemma \ref{Lemma: unsat homogeneous gap}}
Like in the proof of Lemma \ref{Lemma: unsat gap}, we analyse $\cal S(f(A\vecz))$ with $A$ as listed in the lemma statement. Suppose $A(x_0)$ is a linear form in at least $2$ variables. Since $x_0^{d_1}$ divides $f$, it follows from Claim \ref{lem_div} that $\cal S(f(A\vecz)) \geq d_1 + 1$. Similarly, if $A(x_0)$ is a variable while $A(y_0)$ is a linear form in at least $2$ variables, then since $y_0^{d_2}$ divides $f$, $\cal S(f(A\vecz)) \geq d_2 + 1$ holds by Claim \ref{lem_div}.  For the remaining cases, $A(x_0) = x_0$ and $A(y_0) = y_0$ without loss of generality. It then follows from Lemma \ref{Lemma: y0 separated} that
\[\cal S(f(A\vecz)) = \cal S(A(x_0^{d_1}y_0^{d_2 + (3n+m+1)(d_3+1)})) + \sum_{i=1}^{n} \cal S(Q_i(A\vecz)) + \sum_{k=1}^{m} \cal S(R_k(A\vecz)).\]
The last three cases then can be proved the same way as the last three cases of Lemma \ref{Lemma: unsat gap} in the non-homogeneous case.

\subsection{Proof of Lemma \ref{Lemma: unsat gap char p}}
By Observations \ref{Obs: deg separated polys}, \ref{Lemma: equiv deg sep} and the definition of $f$ as in \eqref{Definition: 3SAT poly}, it follows that:
\[\cal S(f(A\vecz)) = \cal S(A(x_0^{d_1})) + \sum_{i=1}^{n} \cal S(Q_i(A\vecz)) + \sum_{k=1}^{m} \cal S(R_k(A\vecz)).\]
We now analyse $\cal S(f(A\vecz))$ with $A$ as listed in the lemma statement. The analysis for the first, second and fourth cases is similar to that of the respective cases in the proof of Lemma \ref{Lemma: unsat gap}. In the third case, without loss of generality, let $A(y_j + x_j)$, for some $j \in [n]$, be a linear form in at least $3$ variables. By Observation \ref{Lemma: linear form sparsity} and the fact that $d_3 = p^k - 1 = \sum_{i=0}^{k-1} (p-1)p^i$ for some $k \in \N$, 
    \[\cal S(A((y_j + x_j)^{d_3})) \geq \prod_{i=0}^{k-1} \binom{p - 1 + 3 - 1}{3 - 1} = \bigg(\frac{p(p+1)}{2}\bigg)^k = (d_3+1) \bigg(\frac{p+1}{2}\bigg)^k = (d_3+1)^{1+ \log_p(\frac{p+1}{2})}\]
Note, $\log_p(\frac{p+1}{2})$ is an increasing function for $p \geq 3$. Hence, $\log_p(\frac{p+1}{2}) \geq \log_3(4/2) \geq 0.63$. Therefore,
    \[\cal S(f(A\vecz)) \geq \cal S(Q_{j,2}(A\vecz)) \geq (d_3+1)^{1.63}.\]
Similarly, if $A(y_j-x_j)$ is a linear form in at least $3$ variables, then $\cal S(f(A\vecz)) \geq \cal S(Q_{j,3}(A\vecz)) \geq(d_3+1)^{1.63}$.

\subsection{Proof of Lemma \ref{Lemma: unsat gap char 2}}
The proof of this lemma is very similar to that of Lemma \ref{Lemma: unsat gap char p}. In particular, the analysis for the first, second and fourth cases is similar to that of the respective cases in the proof of Lemma \ref{Lemma: unsat gap char p}. So, we consider the third case. Without loss of generality, let $A(y_j + x_j)$ be a linear form in at least $3$ variables for some $j \in [n]$. Then, by Observation \ref{Lemma: linear form sparsity} and the fact that $d_3 = 2^k - 1 = \sum_{i=0}^{k-1} 2^i$ for some $k \in \N$, 
    \[\cal S(A((y_j + x_j)^{d_3})) \geq \prod_{i=0}^{k-1} \binom{1 + 3 - 1}{3 - 1} = 3^k  = (d_3+1)^{\log_2 3} \geq (d_3+1)^{1.58}.\]
Therefore,
    \[\cal S(f(A\vecz)) \geq \cal S(Q_{j,2}(A\vecz)) \geq (d_3+1)^{1.58}.\]
Similarly, if $A(y_j)$ is a linear form in at least $3$ variables, then $\cal S(f(A\vecz)) \geq \cal S(Q_{j,3}(A\vecz)) \geq(d_3+1)^{1.58}$.

\subsection{Proof of Lemma \ref{Lemma: unsat gap homogeneous char p}}
Suppose $A(x_0)$ is a linear form in at least two variables. Applying the argument in the proof of Lemma \ref{Lemma: y0 separated} for the finite characteristic case (refer Section \ref{proof-x0,y0-fixed-homogeneous-finitechar}) shows that 
\[\cal S(f(A\vecz)) \geq \cal S(A(x_0)^{d_1}) \geq p^{k_1+1} \geq d_3+2. \]
Similarly, if $A(x_0)$ is a variable and $A(y_0)$ is a linear form in at least two variables, the same argument shows that 
\[\cal S(f(A\vecz)) \geq \cal S(A(y_0)^{d_2}) \geq p^{k_1+1} \geq d_3+2. \]
For the remaining cases, $A(x_0) = x_0$ and $A(y_0) = y_0$ without loss of generality. It then follows from Lemma \ref{Lemma: y0 separated} that
\[\cal S(f(A\vecz)) = \cal S(A(x_0^{d_1}y_0^{d_2 + (3n+m+1)(d_3+1)})) + \sum_{i=1}^{n} \cal S(Q_i(A\vecz)) + \sum_{k=1}^{m} \cal S(R_k(A\vecz)).\]
The last three cases can then be proved in the same way as the last three cases of Lemma \ref{Lemma: unsat gap char p}.

\subsection{Proof of Lemma \ref{Lemma: unsat gap homogeneous char 2}}
Suppose $A(x_0)$ is a linear form in at least two variables. Applying the argument in the proof of Lemma \ref{Lemma: y0 separated} for the finite characteristic case (refer Section \ref{proof-x0,y0-fixed-homogeneous-finitechar}) shows that 
\[\cal S(f(A\vecz)) \geq \cal S(A(x_0)^{d_1}) \geq 2^{k_1+1} \geq d_3+2. \]
Similarly, if $A(x_0)$ is a variable and $A(y_0)$ is a linear form in at least two variables, the same argument shows that 
\[\cal S(f(A\vecz)) \geq \cal S(A(y_0)^{d_2}) \geq 2^{k_1+1} \geq d_3+2. \]
For the remaining cases, $A(x_0) = x_0$ and $A(y_0) = y_0$ without loss of generality. It then follows from Lemma \ref{Lemma: y0 separated} that
\[\cal S(f(A\vecz)) = \cal S(A(x_0^{d_1}y_0^{d_2 + (3n+m+1)(d_3+1)})) + \sum_{i=1}^{n} \cal S(Q_i(A\vecz)) + \sum_{k=1}^{m} \cal S(R_k(A\vecz)).\]
The last three cases can then be proved in the same way as the last three cases of Lemma \ref{Lemma: unsat gap char 2}.

\section{Missing proofs from Section \ref{Section: ETsupport NP-hard}} \label{Section: ETsupport proofs}
\subsection{Proof of Observation \ref{Obs: deg sep clause poly}}
Note that $R_k(\vecw)$ is a product of variable disjoint polynomials. Hence, the degree of $R_k$ is $\sigma - 5 + \sum_{j \in C_k}(2 + a_{k,j})$ where  
\[\sigma + 1 \leq \sigma - 5 + \sum_{j \in C_k}(2 + a_{k,j}) \leq \sigma + 4\] 
and the upper bound is achieved if $a_{k,j} = 1$ for all $j \in C_k$. Similarly, it can be seen that the variable disjointness property along with the fact that the characteristic of the underlying field is not equal to $2$ or $3$ implies 
\[ \sigma - 2 \leq \supp (R_k(\vecw)) \leq \sigma +1 \] 
where the upper bound is achieved if $a_{k,j} = 1$ for all $j \in C_k$. Now for $k,l \in [m]$, where $k \neq l$, there are two possibilities for the clauses $C_k$ and $C_l$:
\begin{enumerate}
\item $C_k$ and $C_l$ contain the same set of $\vecx$ variables in which case there exists a $j$ such that $a_{k,j} \neq a_{l,j}$. Assume without loss of generality that $a_{k,j} = 1$, therefore $a_{l,j} = 0$. From this and the definition of $R_k$ and $R_l$, it can be seen that 
\[ R_k(\vecw) = (y_j-x_j)^3\tilde{R}_k \cdot (z_1z_2\dots z_{\sigma-5}) , \ R_l(\vecw) = y_j^2\tilde{R}_l \cdot (z_1z_2\dots z_{\sigma-5}) \]
where $\tilde{R}_k$ and  $\tilde{R}_l$ are polynomials which do not contain a variable from $\vecz \sqcup \{x_j,y_j\}$.  Thus, the monomials of $R_k$ are degree separated from those of $R_l$ with respect to $x_j$ or $y_j$.
\item Otherwise $C_k$ and $C_l$ differ in at least one variable. This again implies that there exists a variable $y_j$ such that the monomials of $R_k$ have non-zero $y_j$-degree or non-zero $x_j$-degree while the monomials of $R_l$ have zero $y_j$-degree and zero $x_j$-degree. 
\end{enumerate}
Note that the above argument also holds, with some modification, if all the $R_k$'s are of the form
\[R_k(\vecw) = (\prod_{j \in C_k} (y_j-c_jx_j)^{2+a_{k,j}})(z_1z_2\dots z_{\sigma-5}), \ \ \ c_j \in \F.\]
\subsection{Proof of Observation \ref{Obs: low support properties}}
By Condition \ref{pow-degsep-CSP} in Section \ref{secConstructPoly}, Observations \ref{Obs: deg sep clause poly} and \ref{lem_deg_sep_sum}, it holds that,
\[\cal S(f(\vecw)) = \sum_{g(\vecw) \in P}\cal S(g(\vecw)) + \sum_{h(\vecw) \in Q}\cal S(h(\vecw)) + \sum_{k=1}^{m} \cal S(R_k(\vecw)).\]
As $g(\vecw) \in P$ and $h(\vecw) \in Q$ are monomials, thus $\cal S(g(\vecw)) = 1$ and $\cal S(h(\vecw)) = 1$. For $k \in [m]$, 
\[\cal S(R_k(\vecw)) = (\prod_{j \in C_k} \cal S(y_j - a_{k,j}x_j)^{2+a_{k,j}}) \cdot \cal S(z_{1}z_{2} \dots z_{\sigma-5}) \leq 64\]
because all the polynomials in the product are variable disjoint and $a_{k,j}$ can be non-zero for all $j \in C_k$. Lastly, as $|P| = \binom{n+\sigma-5}{\sigma}$ and $|Q| = \binom{n}{\sigma/2}$ (for odd $\sigma$, $|Q| = \binom{n}{\frac{\sigma+1}{2}}\frac{\sigma+1}{2}$), hence
\[\cal S(f) \leq  \binom{n+\sigma -5}{\sigma} + \binom{n}{\sigma/2} + 64m.\]
Thus, $\cal S(f) = O(n^\sigma + m)$. Now, $\supp(g(\vecw)) = \sigma$ and $\supp(h(\vecw)) = \sigma$, where $g(\vecw) \in P$ and $h(\vecw) \in Q$. By Observation \ref{Obs: deg sep clause poly}, $\supp(R_k(\vecw)) \leq \sigma+1$, for $k \in [m]$. In particular, $\supp(R_1) = \sigma + 1$. By Condition \ref{pow-degsep-CSP} in Section \ref{secConstructPoly} and Observation \ref{Obs: deg sep clause poly}, it holds that $\supp(f(\vecw)) = \supp(R_1(\vecw)) = \sigma+1$.
\subsection{Proof of Observation \ref{Obs: support of orbit poly}}
This observation follows easily from Condition \ref{pow-degsep-CSP} in Section \ref{secConstructPoly}.
\subsection{Proof of Lemma \ref{P1Sep}}
    Let $g(\vecw) = (w_{1}\cdots w_{\sigma})^\star \in P$, where $w_{j} \in \vecz \sqcup \vecx$ and $\star$ represents an integer power which is at least $\sigma + 1$. Now, $g(A\vecw) = (\ell_{1}\cdots \ell_{\sigma})^\star$, where the $\ell_{j}$'s, with $j \in [\sigma]$, are linearly independent linear forms. If $|\cup_{j=1}^{\sigma} \var(\ell_{j})| \geq \sigma + 1$, then by Claim \ref{FullSupport}, $\supp(g(A\vecw)) \geq \sigma + 1$ a contradiction. Thus, $|\cup_{j=1}^{\sigma} \var(\ell_{j})| \leq \sigma$. The linear independence of these $\sigma$ many linear forms implies $|\cup_{j=1}^{\sigma} \var(\ell_{j})| \geq \sigma$. Combining these inequalities gives $|\cup_{j=1}^{\sigma} \var(\ell_{j})| = \sigma$. Thus, there exist variables $W_{1},\dots,W_{\sigma} \in \vecw$ such that $\langle \{\ell_{1},\dots,\ell_{\sigma}\} \rangle = \langle \{ W_{1},\dots,W_{\sigma} \} \rangle$, where $\langle S \rangle$ denotes the vector space spanned by the elements of a set $S$ of polynomials. In particular, $\langle \{\ell_{1},\dots,\ell_{\sigma}\} \rangle = \langle \var(g(A\vecw)) \rangle$.

    Consider the variable $x_1 \in \vecx \sqcup \vecz$ and let $A(x_1) = \ell_1$. Now, $P$ can be seen as a collection of $\sigma$ sized subsets of $\vecz \sqcup \vecx$. Since $n + \sigma - 6 \geq 2(\sigma-1)$ (as implied by $n \geq \sigma + 4$), there exist, without loss of generality, $g_1(\vecw)$ and $g_2(\vecw) \in P$, such that $\var(g_1(\vecw)) \cap  \var(g_2(\vecw)) = \{x_1\}$. Note that $\langle \var(g_1(\vecw)) \rangle \cap \langle \var(g_2(\vecw)) \rangle = \langle \var(g_1(\vecw)) \cap \var(g_2(\vecw)) \rangle$. Thus, 
    \[\dim \langle \var(g_1(\vecw)) \rangle \cap \langle \var(g_2(\vecw)) \rangle= \dim \langle \var(g_1(\vecw)) \cap \var(g_2(\vecw)) \rangle = 1.\] 
    The invertibility of $A$ and the argument of the first paragraph imply there exist $\sigma$-size sets $B_1,B_2 \subseteq \vecw$, with $B_1 = \var(g_1(A\vecw))$ and $B_2 = \var(g_2(A\vecw))$, such that $\langle \var(g_1(\vecw)) \rangle$ and $\langle B_1 \rangle$ are isomorphic and so are $\langle \var(g_2(\vecw)) \rangle$ and $\langle B_2 \rangle$. Similarly, under $A$, $\langle \var(g_1(\vecw)) \rangle \cap \langle \var(g_2(\vecw)) \rangle$ is isomorphic to $\langle B_1 \rangle \cap \langle B_2 \rangle$. Thus, $\dim$ $\langle B_1 \rangle \cap \langle B_2 \rangle = 1$. From the first paragraph, it is also evident that $\ell_1 \in \langle B_1 \rangle \cap \langle B_2 \rangle = \langle B_1 \cap B_2 \rangle$, where the equality follows as $B_1$ and $B_2$ are sets of variables. This implies $A(x_1) = \ell_1 = W_1$, where $W_1$ is some scaled variable in $\vecw$.

\subsection{Proof of Lemma \ref{P2Sep}}
    The proof is similar to that of Lemma \ref{P1Sep}. As $\supp (f(A\vecw)) \leq \sigma$, Lemma \ref{P1Sep} holds. Thus, for all $w \in \vecx \sqcup \vecz$, $A(w) = W$ for some scaled variable $W \in \vecw$. In particular, $A(x_i) = X_i$ where $i \in [n]$ and $X_i \in \vecw$ is some scaled variable. Let $A(y_i) = \ell_i$, where $\ell_i$ contains at least one variable other than those in $\var(A(w))$, for any $ w \in \vecx \sqcup \vecz$. Without loss of generality, consider $h(\vecw) = ((x_{1}y_{1})\cdots (x_{{\frac{\sigma}{2}}} y_{{\frac{\sigma}{2}}}))^\star \in Q$, where $\star$ represents an integer power which is at least $\sigma + 1$. Now, $h(A\vecw) = ((X_{1}\ell_{1})\cdots (X_{{\frac{\sigma}{2}}} \ell_{{\frac{\sigma}{2}}}))^\star$, where $X_{1} \dots X_{{\frac{\sigma}{2}}}, \ell_{1}, \dots \ell_{{\frac{\sigma}{2}}}$ are linearly independent linear forms. If $|\var(\ell_{1}) \cup \dots \cup \var(\ell_{{\frac{\sigma}{2}}}) \cup \{X_{1},\dots,X_{{\frac{\sigma}{2}}}\}| \geq \sigma+1$, then by Claim \ref{FullSupport}, $\supp(h(A\vecw)) \geq \sigma + 1$, a contradiction. Thus, $|\var(\ell_{1}) \cup \dots \cup \var(\ell_{{\frac{\sigma}{2}}}) \cup \{X_{1},\dots,X_{{\frac{\sigma}{2}}}\}| \leq \sigma$. The linear independence of these $\sigma$ many linear forms implies $|\var(\ell_{1}) \cup \dots \cup \var(\ell_{{\frac{\sigma}{2}}}) \cup \{X_{1},\dots,X_{{\frac{\sigma}{2}}}\}| \geq \sigma$. These inequalities imply $|\var(\ell_{1}) \cup \dots \cup \var(\ell_{{\frac{\sigma}{2}}}) \cup \{X_{1},\dots,X_{{\frac{\sigma}{2}}}\}| = \sigma$. Hence there exists variable set $B = \{X_{1},\dots,X_{{\frac{\sigma}{2}}}\} \sqcup \{Y_{1},\dots,Y_{{\frac{\sigma}{2}}}\}$ such that $\langle B \rangle = \langle \{\ell_{{1}}, \dots, \ell_{{\frac{\sigma}{2}}} \} \sqcup \{X_{1},\dots,X_{{\frac{\sigma}{2}}} \} \rangle$.  In particular, $\langle \{\ell_{{1}}, \dots, \ell_{{\frac{\sigma}{2}}} \} \sqcup \{X_{1},\dots,X_{{\frac{\sigma}{2}}} \} \rangle = \langle \var(h(A\vecw)) \rangle$.

Consider the variable $x_1 \in \vecx$. As $n-1 \geq \sigma - 2$ (as implied by $n \geq \sigma + 4$), then for $x_1 \in \vecx$ and $y_1 \in \vecy$, there exist, $h_1(\vecw), h_2(\vecw) \in Q$ such that $\var(h_1(\vecw)) \cap \var(h_2(\vecw)) = \{x_1,y_1\}$. Note that $\langle \var(h_1(\vecw)) \rangle \cap \langle \var(h_2(\vecw)) \rangle = \langle \var(h_1(\vecw)) \cap \var(h_2(\vecw)) \rangle$. Thus, 
    \[\dim \langle \var(h_1(\vecw)) \rangle \cap \langle \var(h_2(\vecw)) \rangle= \dim \langle \var(h_1(\vecw)) \cap \var(h_2(\vecw)) \rangle = 2.\] 
 The invertibility of $A$ and the argument of the first paragraph imply that there exist $\sigma$-size sets $B_1,B_2 \subseteq \vecw$, with $B_1 = \var(h_1(A\vecw))$ and $B_2 = \var(h_2(A\vecw))$, such that $\langle \var(h_1(\vecw)) \rangle$ and $\langle B_1 \rangle$ are isomorphic, and so are $\langle \var(h_2(\vecw)) \rangle$ and $\langle B_2 \rangle$. Similarly under $A$, $\langle \var(h_1(\vecw)) \rangle \cap \langle \var(h_2(\vecw)) \rangle$  is isomorphic to $\langle B_1 \rangle \cap \langle B_2 \rangle$. Therefore, $\dim \langle B_1 \rangle \cap \langle B_2 \rangle = 2$. As $A(x_1) = X_1$, and $A(y_1) \in \langle B_1 \rangle \cap \langle B_2 \rangle$ (by the argument in the first paragraph), therefore $A(y_1) = Y_1 + c_1X_1$, where $Y_1 \in \vecw$ is some scaled variable and $c_1 \in \F$.
 
 \noindent \textbf{Note:} For odd $\sigma$, the proof holds with some modification. First, $h(\vecw) \in Q$ is now of form $((x_{1}y_{1})\cdots (x_{{\frac{\sigma-1}{2}}} y_{{\frac{\sigma-1}{2}}})x_{{\frac{\sigma+1}{2}}})^\star$. The argument in the first paragraph then shows that for $h(A\vecw) = ((X_{1}\ell_{1})\cdots (X_{{\frac{\sigma-1}{2}}} \ell_{{\frac{\sigma-1}{2}}})X_{{\frac{\sigma+1}{2}}})^\star$, there exists a $\sigma$-size set $B \subseteq \vecw$, such that $B = \{X_{1},\dots,X_{{\frac{\sigma+1}{2}}}\} \sqcup \{Y_{1},\dots,Y_{{\frac{\sigma-1}{2}}}\}$ and $\langle B \rangle = \langle \{\ell_{{1}}, \dots, \ell_{{\frac{\sigma-1}{2}}} \} \sqcup \{X_{1},\dots,X_{{\frac{\sigma+1}{2}}} \} \rangle$. Secondly, for a variable $x_1 \in \vecx$, the existence of $h_1(\vecw), h_2(\vecw) \in Q$ such that $\var(h_1(\vecw)) \cap \var(h_2(\vecw)) = \{x_1,y_1\}$ is implied by $n-1 \geq \sigma-1$ (which itself is implied by $n \geq \sigma + 4$). With these changes, the remainder of the argument continues to hold. 
\subsection{Choosing the degrees}\label{etcsp-param}
In this section, we set the powers $\star$, as denoted in Section \ref{secConstructPoly}, for all polynomials in $P$, $Q$ and $R$ such that Conditions \ref{pow-degsep-CSP} and \ref{pow-finitechar-CSP}, specified in that section, are satisfied over any field. We specify the choice of $\star$ separately for characteristic $0$ fields and finite characteristic fields. Let $N := \binom{n+\sigma -5}{\sigma} + \binom{n}{\sigma/2}$, $i \in [N]$ and $k \in [m]$. For odd $\sigma$, $N := \binom{n+\sigma -5}{\sigma} + \binom{n}{\frac{\sigma+1}{2}}\frac{\sigma+1}{2}$.

\paragraph{Over characteristic $0$ fields.} For the polynomials in $P \sqcup Q$, arbitrarily order them and choose the powers to be of form $\sigma + i$. For this choice of the powers, every polynomial in $P \sqcup Q$ has corresponding degree $\sigma(\sigma + i)$ and is clearly degree separated from the other polynomials in $P \sqcup Q$. By Observation \ref{Obs: deg sep clause poly}, the degree of $R_k$ is at most $\sigma + 4$. As $i \geq 1$ and $\sigma > 2$ it can be easily observed that 
\[\sigma(\sigma + i) \geq \sigma(\sigma + 1) > \sigma + 4.\] 
Thus, for this choice of the powers, Condition \ref{pow-degsep-CSP} is satisfied over characteristic $0$ fields. The degree of $f$ then is $\sigma(\sigma + N) = O(n^{\sigma})$. 

\paragraph{Over finite characteristic fields.} Let the characteristic be $p > 0$. We assume $p > \sigma + 1$ to ensure Claim \ref{FullSupport} holds. If $p > \sigma+N$, the powers can be chosen just like in the characteristic $0$ case. Otherwise if $\sigma + 1 < p \leq \sigma + N$, then we choose the powers to be of form $p^t - 1$, where $t \in \N$, from the following $N$ disjoint intervals:  
\[[\sigma+1,p(\sigma+1)], \ [p(\sigma+1)+1 , p^2(\sigma+1)+p], \ [p^2(\sigma+1) + p + 1 , p^3(\sigma+1) + p^2 + p],\cdots\]
Order the $N$ polynomials in $P \sqcup Q$ arbitrarily and assign powers in each polynomial from the above $N$ intervals. Then, the degree of a polynomial in $P \sqcup Q$ lies in the range \[[\sigma(p^{i-1}(\sigma+1) + \sum_{l=0}^{i-2}p^l),\sigma(p^i(\sigma+1) + \sum_{l=1}^{i-1}p^l)].\] For distinct $i$ this range is disjoint implying the polynomials in $P \sqcup Q$ are degree separated. By Observation \ref{Obs: deg sep clause poly}, the degree of $R_k$ is at most $\sigma + 4$. Now, we show that the polynomials of $P \sqcup Q$ are degree separated from the $R_k$'s by showing that the lower bound on the degree of any polynomial in $P \sqcup Q$ is greater than the degree of any $R_k$.  As $i \geq 1$, $\sigma > 2$ and $p > 1$, it follows that
\[\sigma(p^{i-1}(\sigma+1) + \sum_{l=0}^{i-2}p^l) \geq \sigma(\sigma+1)  > \sigma + 4.\]
Thus, for this choice of powers, Conditions \ref{pow-degsep-CSP} and \ref{pow-finitechar-CSP} are satisfied. For $p > \sigma + N$, the degree of $f$ is $\sigma(\sigma + N) = O(n^\sigma)$  and for $\sigma + 1 < p \leq \sigma + N$, the degree of $f$ is 
\[O(\sigma(p^N(\sigma+1) + \sum_{l=1}^{N-1}p^l)) = O(p^N) = O((\sigma+N)^N) = O((\sigma + n^{\sigma})^{n^\sigma}).\]
Note that the degree of $f$ can be represented in $n^{O(1)}$ many bits as $\sigma$ is a constant. Hence, as remarked just after Theorem \ref{Theorem: ETsupport NP-hard}, we assume that over finite characteristic fields, the exponent vectors corresponding to the monomials of the input polynomials are given in binary.

\section{Missing proofs from Section \ref{section-shifteqvhard}} \label{Section: sparse-shift proofs}
\subsection{Proof of Observation \ref{Obs: deg separated poly-shift}}
    Let $i \in [n]$. Note that the $x_0$-degree of the monomials of $Q_{i,1}(\vecz)$ is $(2i-1)(d_2+1)$ while that of the monomials of $Q_{i,2}(\vecz)$ is $2i(d_2+1)$. Clearly, $Q_{i,1}(\vecz)$ and $Q_{i,2}(\vecz)$ are degree separated with respect to $x_0$. Thus, $Q_i$ is a sum of two polynomials degree separated with respect to $x_0$ and has degree $2i(d_2+1) + d_2$. It is also easy to observe that for $i,j \in [n]$, where $i < j$ without loss of generality, $Q_i(\vecz)$ is degree separated from $Q_j(\vecz)$ with respect to $x_0$. Similarly, $R_k(\vecz)$ is degree separated from $R_l(\vecz)$ with respect to $x_0$, for $k,l \in [m]$ and $k \neq l$.
    
    Now, let $i \in [n]$ and $k \in [m]$. The $x_0$-degree of a monomial of $Q_i(\vecz)$ is at least $(2i-1)(d_2+1) \geq d_2 + 1$ while that of a monomial of $R_k(\vecz)$ is $k$. As $d_2 + 1 \geq m(d_3+1)^2 + 2 > m \geq k$ by the conditions in \eqref{Sparse-shift-cond}, therefore $Q_i$ is degree separated from $R_k$. Lastly, $x_0^{d_1}$ is degree separated from the rest of the polynomials with respect to $x_0$ because $d_1 \geq 4n(d_2+1) + 2d_2 + 2$ while the highest $x_0$-degree of a monomial among $Q_i$'s and $R_k$'s is $2n(d_2+1)$. 

\subsection{Proof of Observation \ref{Obs: shift-poly degree}}

By Observation \ref{Obs: deg separated poly-shift} and the definition of $f$ in \eqref{Sparse-shift-poly}, the degree of $f$ is the maximum degree among $x_0^{d_1}$, $Q_i$'s and $R_k$'s, where $i \in [n]$ and $k \in [m]$. Now, the degree of $Q_i$ is $2i(d_2+1) + d_2$, that of $R_k$ is $k + 3d_3$ and $2i(d_2+1) + d_2  > k + 3d_3$. Further, $d_1 \geq 4n(d_2+1) + 2d_2 + 2$. Hence, the degree of $f$ is $d_1$. Finally, it follows from \eqref{Sparse-shift-cond} that $\frac{d_1}{2} > 2n(d_2+1) + d_2 > s$.

\subsection{Proof of Observation \ref{Obs: shift-poly sparsity}}
    By Observations \ref{Obs: deg separated poly-shift} and \ref{lem_deg_sep_sum},
    \[\cal S(f(\vecz)) = 1 + \sum_{i=1}^{n}\cal S(Q_i(\vecz)) + \sum_{k=1}^{m}\cal S(R_k(\vecz)).\]
    Let $i \in [n]$. By Observation \ref{Lemma: affine form sparsity}, $\cal S(Q_{i,1}) = \cal S(Q_{i,2}) = d_2 + 1$. Thus, $\cal S(Q_i) = 2d_2+2$. Now, let $k \in [m]$. By Observation \ref{Lemma: affine form sparsity} and the assumption that each clause in $\psi$ has $3$ distinct variables, we get that:
\[\cal S(R_k(\vecz)) = \cal S(x_0^{k})\prod_{j \in C_k} \cal S((x_j + (-1)^{a_{k,j}})^{d_3}) = (d_3 + 1)^3 \ \ \forall k \in [m].\]
    Thus,
    \[\cal S(f(\vecz)) = 1 + \sum_{i=1}^{n}\cal S(Q_i(\vecz)) + \sum_{k=1}^{m}\cal S(R_k(\vecz)) = 1 + n(2d_2 + 2) + m(d_3+1)^3.\]
This also shows that $\cal S(f) = (mn)^{O(1)}$ and $\cal S(f) > s$. For the support of $f$, note that $\supp(R_k) = 4$ for all $k \in [m]$ while $\supp(Q_i) = 2$ for all $i \in [n]$ and $\supp(x_0^{d_1}) = 1$. By Observation \ref{Obs: deg separated poly-shift}, $\supp(f) = \supp(R_k) = 4$.

\subsection{Proof of Lemma \ref{Lemma: x0 fixed-shift}}
    Let $b_0 = \vecb_{|x_0}$. If $b_0 \neq 0$, then by the binomial theorem
    \[(x_0 + b_0)^{d_1}= \sum_{i=0}^{d_1} \binom{d_1}{i} b_0^i x_0^{d_1-i}.\]
    The summands in the above expansion are degree separated and $\binom{d_1}{i} \neq 0$ for all $i \in [0,d_1]$, by the choice of $d_1$ (and Lucas's Theorem if the characteristic is finite). Therefore, $(x_0 + b_0)^{d_1}$ contains at least one monomial of degree $i$, for all $i \in [0,d_1]$. Since $\frac{d_1}{2} > 2n(d_2+1) + d_2 > s$ therefore at least $s+1$ monomials in $(x_0 + b_0)^{d_1}$ are degree separated from $Q_i(\vecz+\vecb)$'s and $R_k(\vecz+\vecb)$'s with respect to $x_0$. Hence $\cal S(f(\vecz + \vecb))  > s$, a contradiction. Thus, $\vecb_{|x_0} = 0$.

\subsection{Proof of Lemma \ref{Lemma: degree separation-shift}}

    Let $i \in [n]$. For $Q_{i,1}(\vecz + \vecb)$ and $Q_{i,2}(\vecz + \vecb)$, the $x_0$-degree of a monomial in the respective polynomials is $(2i-1)(d_2+1)$ and $2i(d_2+1)$ respectively. This clearly implies $Q_{i,1}(\vecz + \vecb)$ and $Q_{i,2}(\vecz + \vecb)$ are degree separated from one another with respect to $x_0$. For $i,j \in [n]$, with $i < j$ without loss of generality, $Q_i(\vecz + \vecb)$ is degree separated from $Q_j(\vecz + \vecb)$ with respect to $x_0$ as the highest $x_0$-degree in the former polynomial is $2i(d_2+1)$ while the lowest $x_0$-degree in the latter polynomial is $(2j-1)(d_2+1) \geq (2i+1)(d_2+1)$. Let $k, l \in [m]$ and $k < l$ without loss of generality. For $R_k(\vecz + \vecb)$ and $R_l(\vecz + \vecb)$, the $x_0$-degree of a monomial in respective polynomials is $k$ and $l$ implying $R_k$ and $R_l$ are also degree separated with respect to $x_0$.
    
    Lastly, let $i \in [n]$ and $k \in [m]$. The $x_0$-degree of a monomial of $Q_i(\vecz + \vecb)$ is at least $(2i-1)(d_2+1) \geq d_2 + 1$ while that of a monomial of $R_k(\vecz + \vecb)$ is $k \leq m$. As $d_2 + 1 \geq m(d_3+1)^2 + 2 > m \geq k$ by the conditions in \eqref{Sparse-shift-cond}. Therefore, $Q_i(\vecz+\vecb)$ is degree separated from $R_k(\vecz+\vecb)$'s with respect to $x_0$. Clearly, $x_0^{d_1}$ is degree separated with respect to $x_0$ from $Q_i(\vecz+\vecb)$ and $R_k(\vecz+\vecb)$ because $d_1 > 4n(d_2 + 1) + 2d_2$, while the highest $x_0$-degree of any monomial among the $Q_i(\vecz+\vecb)'s$ and $R_k(\vecz+\vecb)$'s is $2n(d_2+1)$.

\subsection{Proof of Lemma \ref{Qi sparsity analysis: shift}}

    From Lemma \ref{Lemma: degree separation-shift} and Observation \ref{lem_deg_sep_sum} it follows that $\cal S(Q_i(\vecz+\vecb)) = \cal S(Q_{i,1}(\vecz+\vecb))+\cal S(Q_{i,2}(\vecz+\vecb))$. The if direction of the lemma statement is easy to verify. For the only if direction consider the following cases of $\vecb$:
\begin{enumerate}   
    \item $\vecb_{|x_i} \notin \{-1,1\}$: In this case, $x_i + \vecb_{|x_i} - 1$ and $x_i + \vecb_{|x_i} + 1$ both have a non-zero constant. Then, by the binomial theorem, the choice of $d_2$ and Lucas's Theorem (for finite characteristic fields), it holds that $\cal S(Q_{i,1}(\vecz + \vecb)) = d_2+1$ and $\cal S(Q_{i,2}(\vecz + \vecb)) = d_2+1$ implying $\cal S(Q_{i}(\vecz + \vecb)) = 2d_2 + 2$.
    
    \item $\vecb_{|x_i} \in \{-1,1\}$: In this case, exactly one of $x_i + \vecb_{|x_i} - 1$ and $x_i + \vecb_{|x_i} + 1$ has a non-zero constant. Without loss of generality, let $\vecb_{|x_i} = 1$. Then, $\cal S(Q_{i,2}(\vecz + \vecb)) = d_2+1$ (by arguing as in the first case) while $\cal S(Q_{i,1}(\vecz + \vecb)) = 1$ implying $\cal S(Q_{i}(\vecz + \vecb)) = d_2 + 2$.
\end{enumerate}
It is clear from the above cases that if $\vecb$ is not as per the lemma statement, then $\cal S(Q_i(\vecz + \vecb)) = 2d_2 + 2$; otherwise, $\cal S(Q_{i}(\vecz + \vecb)) = d_2 + 2$.

\subsection{Proof of Lemma \ref{Lemma: equality condition Qi-shift}}
Suppose $\cal S(Q_j(\vecz + \vecb)) \neq d_2 + 2$ for some $j \in [n]$. Then, $\cal S(Q_j(\vecz + \vecb)) = 2d_2 + 2$ by Lemma \ref{Qi sparsity analysis: shift}. By the definition of $f$, Lemma \ref{Lemma: degree separation-shift} and the condition $d_2 \geq m(d_3+1)^2 + 1$, we get the following contradiction:
    \begin{equation*}
    \begin{split}
        \cal S(f(\vecz + \vecb)) &>  \cal S((x_0)^{d_1}) + \sum_{i=1, i \neq j}^{n} \cal S(Q_i(\vecz + \vecb)) + \cal S(Q_j(\vecz + \vecb)) \\  &\geq  1+(n-1)(2+d_2)+(2+2d_2) =1 + n(2+d_2) + d_2 > s.
    \end{split}
    \end{equation*}

\subsection{Proof of Lemma \ref{Lemma: unsat gap-shift}}
If $\vecb_{|x_0} \neq 0$, then it follows from the argument in the proof of Lemma \ref{Lemma: x0 fixed-shift} that $\cal S(f(\vecz + \vecb)) \geq \frac{d_1}{2}$. Henceforth, we assume  $\vecb_{|x_0} = 0$. Then, by Lemma \ref{Lemma: degree separation-shift}, it follows that
\[\cal S(f(\vecz + \vecb)) = \cal S(x_0^{d_1}) + \sum_{i=1}^{n} \cal S(Q_i(\vecz + \vecb)) + \sum_{k=1}^{m} \cal S(R_k(\vecz + \vecb)).\]
Now, note that 
\[R_k(\vecz + \vecb) = \prod_{j \in C_k}(x_j + \vecb_{|x_j} + (-1)^{a_{k,j}})^{d_3}.\]
As $\psi$ is unsatisfiable, there exists a $k \in [m]$ such that $\cal S(R_k(\vecz + \vecb)) = (d_3+1)^3$, which follows by the binomial theorem.